\documentclass[aos,preprint]{imsart}
\setattribute{journal}{name}{}

\RequirePackage[OT1]{fontenc}
\RequirePackage[round]{natbib}
\RequirePackage[colorlinks,citecolor=blue,urlcolor=blue, linkcolor=blue]{hyperref}
\usepackage{comment}
\usepackage{lineno}
\usepackage{amsmath, amsthm, amssymb, mathtools, bm, bbm, mathrsfs}
\usepackage{graphicx} 
\usepackage[linesnumbered,ruled,vlined]{algorithm2e}
\usepackage{xr} 
\usepackage[capitalise]{cleveref}
\newcommand{\possessivecite}[1]{\citeauthor{#1}'s \citeyearpar{#1}}
\Crefname{assumption}{Assumption}{Assumptions}
\usepackage{chngcntr}
\usepackage{apptools}
\AtAppendix{\counterwithin{lemma}{section}}
\AtAppendix{\counterwithin{figure}{section}}
\usepackage{caption}
\usepackage{booktabs, multirow}
\usepackage{enumerate}
\SetKwInput{KwInput}{Input}                \SetKwInput{KwOutput}{Output}              

\newtheorem{theorem}{Theorem}
\newtheorem{assumption}{Assumption}

\newtheorem{definition}{Definition}
\newtheorem{lemma}{Lemma}

\newtheorem{corollary}{Corollary}

\theoremstyle{definition}
\newtheorem{remark}{Remark}

\theoremstyle{definition}
\newtheorem{example}{Example}

\makeatletter
\newcommand\RedeclareMathOperator{\@ifstar{\def\rmo@s{m}\rmo@redeclare}{\def\rmo@s{o}\rmo@redeclare}}
\newcommand\rmo@redeclare[2]{\begingroup \escapechar\m@ne\xdef\@gtempa{{\string#1}}\endgroup
  \expandafter\@ifundefined\@gtempa
     {\@latex@error{\noexpand#1undefined}\@ehc}\relax
  \expandafter\rmo@declmathop\rmo@s{#1}{#2}}
\newcommand\rmo@declmathop[3]{\DeclareRobustCommand{#2}{\qopname\newmcodes@#1{#3}}}
\@onlypreamble\RedeclareMathOperator
\makeatother

\makeatletter
\newcommand{\longdash}[1][2em]{\makebox[#1]{$\m@th\smash-\mkern-7mu\cleaders\hbox{$\mkern-2mu\smash-\mkern-2mu$}\hfill\mkern-7mu\smash-$}}
\makeatother
\newcommand{\omitskip}{\kern-\arraycolsep}

\newcommand{\T}{\intercal}

\DeclareMathOperator*{\KL}{\mathcal{D}_\text{KL}}

\newcommand{\iid}{\overset{\textsf{iid}}{\sim}}

\DeclareMathOperator{\var}{\mathsf{var}}

\DeclareMathOperator{\proj}{\mathsf{Proj}}

\DeclareMathOperator{\E}{\mathbb{E}}

\DeclareMathOperator*{\argmax}{arg\,max}

\newcommand*\dd{\mathop{}\!\mathrm{d}}

\newcommand{\I}{\mathbb{I}}

\newcommand{\distconvto}{\overset{d}{\Rightarrow}}

\newcommand{\unif}{\mathsf{unif}}

\newcommand{\pois}{\mathsf{Pois}}

\newcommand{\N}{\mathcal{N}}

\RedeclareMathOperator{\Re}{\text{Re}}
\RedeclareMathOperator{\Im}{\text{Im}} 
\begin{document}

\begin{frontmatter}
\title{Empirical Bayes for Large-scale Randomized Experiments: a Spectral Approach}
\runtitle{A Spectral Approach for Large-scale Empirical Bayes}

\begin{aug}
\author{\fnms{F.~Richard}~\snm{Guo}\thanksref{m1,t2}\ead[label=e1]{ricguo@stat.washington.edu}},
\author{\fnms{James}~\snm{McQueen}\thanksref{m2}\ead[label=e2]{jmcq@amazon.com}}
\and
\author{\fnms{Thomas~S.}~\snm{Richardson}\thanksref{m1,m2}\ead[label=e4]{rchatho@amazon.com}\ead[label=e3]{thomasr@u.washington.edu}}
\runauthor{Guo, McQueen and Richardson}

\affiliation{University of Washington, Seattle\thanksmark{m1} and Amazon.com, Inc.\thanksmark{m2}}

\thankstext{t2}{Research conducted when the first author interned with Amazon.}

\address{Department of Statistics\\
University of Washington\\
Box 354322\\
Seattle, WA 98195 \\
\printead{e1}\\
\phantom{E-mail: }\printead*{e3}\\}

\address{500 9th Ave N\\
Seattle, WA 98109\\
\printead{e2}\\
\phantom{E-mail:\ }\printead*{e4}}
\end{aug}

\begin{abstract}
Large-scale randomized experiments, sometimes called A/B tests, are increasingly prevalent in many industries. Though such experiments are often analyzed via frequentist $t$-tests, arguably such analyses are deficient: $p$-values are hard to interpret and not easily incorporated into decision-making. As an alternative, we propose an empirical Bayes approach, which assumes that the treatment effects are realized from a ``true prior''. This requires inferring the prior from previous experiments. Following Robbins, we estimate a family of marginal densities of empirical effects, indexed by the noise scale. We show that this family is characterized by the heat equation. We develop a spectral maximum likelihood estimate based on a Fourier series representation, which can be efficiently computed via convex optimization. In order to select hyperparameters and compare models, we describe two model selection criteria. We demonstrate our method on simulated and real data, and compare posterior inference to that under a Gaussian mixture model of the prior.
\end{abstract}


\end{frontmatter}

\section{Introduction}
Consider a randomized experiment with a binary treatment and a continuous outcome. We use $c$ and $t$ to denote the control and treatment groups respectively. We wish to estimate the (true) average treatment effect (ATE) $\Delta := \mu_{t} - \mu_{c}$, the difference in the population means.
Subjects are randomly assigned: $n_{c}$ to control and $n_{t}$ to treatment. An unbiased estimator for $\Delta$ is the difference in the empirical means
\begin{equation}
\hat{\Delta} = \bar{X}_t - \bar{X}_c,
\end{equation}
which we will call the ``observed effect''. 

Suppose $n = n_t + n_c \rightarrow \infty$ and $n_t / n \rightarrow \gamma \in (0,1)$. By the central limit theorem, one can show that 
\begin{equation} \label{eqs:t-normal}
T := \frac{(\bar{X}_t - \mu_t) - (\bar{X}_c - \mu_c)}{\sqrt{s_t^2 / n_t + s_c^2 / n_c}} = \frac{\hat{\Delta} - \Delta}{\sqrt{s_t^2 / n_t + s_c^2 / n_c}} \distconvto \N(0,1),
\end{equation}
where $s_t^2$, $s_c^2$ are treatment/control group variances, and ``$\distconvto$'' denotes convergence in distribution. If the observations from treatment/control groups are further assumed to be normally distributed, then under finite $n$, approximately $T \sim t_{\nu}$, which is known as an unequal-variance or Welch's $t$-test \citep{welch1947generalization}. Formulae are available for approximating the degrees of freedom $\nu$ from data, e.g., those proposed by \citet{satterthwaite1946approximate} and \citet{welch1947generalization}. $t$-tests are routinely used for the analysis of randomized experiments. In this paper, we focus on the setting where $n$ is typically very large. Note that $\nu \rightarrow \infty$ as $n \rightarrow \infty$, so the normal limit in \cref{eqs:t-normal} is recovered.

For the analysis of a single experiment, it is common to compute a $p$-value from the $t$-statistic. However, in modern technological and industrial settings, a large number of randomized experiments are run every day. In this setting, continuing to use the $t$-test for each experiment separately is inefficient. As we will see, doing so essentially ignores the information from the population of experiments, even though two experiments may seem unrelated. A similar phenomenon, where estimation benefits from sharing information across seemingly unrelated data, is known as Stein's paradox; the reader is referred to \citet{stein1956inadmissibility} and \citet{efron1977stein} for more discussions. 
Also, while the $t$-test is proper for testing whether the true effect is below or above zero, it is inconvenient for industrial decision making, which may be targeted at a specified function that maps the decision and the true effect to a loss/utility.
Lastly, it is often difficult to communicate the precise interpretation of the $p$-value to a wider audience; see \citet{amrhein2019scientists} for a recent debate.  In the following, we argue that an empirical Bayes approach is a more appealing framework for the analysis of large-scale experiments that overcomes these issues. 

The rest of this paper is organized as follows. In \Cref{sec:empirical-bayes}, we review the foundational idea of Robbins and illustrate how an empirical Bayes analysis of experiments can be conducted and its benefits. In \Cref{sec:parametric}, we briefly review parametric estimation of the prior with normal mixtures. In \Cref{sec:nonparametric}, we show that the estimation problem is characterized by a heat equation and develop a nonparametric maximum likelihood estimate based on a spectral representation. We present an efficient algorithm, consider two model selection criteria, and prove consistency in terms of the method of sieves. In \Cref{sec:numerical}, we compare the performance of our spectral estimator with the mixture-of-Gaussian estimator on simulated and real data. Finally, discussion and bibliographic remarks are presented in \Cref{sec:discussion}.  \section{An empirical Bayes framework} \label{sec:empirical-bayes}
\subsection{The basic idea}
The study of empirical Bayes was pioneered by Herbert Robbins \citep{robbins1956empirical,robbins1964empirical,robbins1983some}. The motivation was to estimate parameters associated with ``many structurally similar problems'' \citep{lai1986contributions}. For example, in \citet{robbins1977prediction}, he considered predicting the number of accidents in the upcoming year based on the number of accidents from the current year for a population of $n$ taxi drivers. Suppose driver $i$ is associated with the rate of accidents $\lambda_i$, and we model the number of accidents $X_i \mid \lambda_i \sim \pois(\lambda_i)$ independently for each driver. Further, suppose that the rate of accidents is the same from year to year, then the expected number of accidents in the coming year is $\E [X_i^{\ast} \mid \lambda_i] = \lambda_i$, for which $X_i$ is the maximum likelihood estimate that treats each driver separately. The idea of Robbins is to assume that $\lambda_i$'s are drawn independently from some unknown distribution $G$, which we will call the ``true prior''. By doing so, he further showed that the posterior mean for $\lambda_i$ can be estimated \emph{nonparametrically} by the formula
\begin{equation*}
\widehat{\E}[\lambda_i \mid X_i=x] = (x + 1) n_{x+1}/ n_{x},
\end{equation*}
where $n_{x+1}$ is the number of drivers in the sample with $x+1$ accidents in the current year. As we can see, the formula utilizes information from other ``seemingly unrelated'' drivers. 

The idea has been recently applied to large-scale hypothesis testing for microarray data; see \citet{efron2003robbins} for an overview. In the context of large-scale experimentation, the idea seems readily applicable if we assume that the true treatment effects $\Delta_i$ are realized independently from some unknown distribution $G$. True effects $\Delta_i$ should be comparable across experiments for the population distribution $G$ to be meaningful. For the rest of the paper, we will assume that the true effects are measured in the same unit and are normalized to the same time duration. 

The actual data generating process for the outcome of an experiment also involves several nuisance parameters, for which we would like to \emph{refrain} from assuming an (even unknown) prior. For example, it seems unnatural to specify priors on the means and variances of treatment and control groups; and as we will see, it is also unnecessary if our goal is to infer the posterior over the true effect. Note that this deviates from the usual ``fully Bayesian'' analysis where priors are typically specified for nuisance parameters. 

\subsection{Nuisance-free asymptotic likelihood} We now consider a likelihood for the true effect that does not involve nuisance parameters, and hence frees us from specifying prior on the nuisance. Observe that \cref{eqs:t-normal} can be turned into an asymptotically efficient likelihood on $\Delta$, given by
\begin{equation}
L(\Delta; \hat{\Delta}, \hat{s}) = \phi_{\hat{s}}(\Delta - \hat{\Delta}), \quad \hat{s} = \sqrt{s_t^2 / n_t + s_c^2 / n_c},
\end{equation}
where $\hat{s}$ is the estimated asymptotic scale and $\phi_{s}(\cdot) = s^{-1} \phi(\cdot/s)$ is the Gaussian density function with variance $s^2$. In other words, the observational model relative to $\Delta$ is asymptotically a normal experiment;\footnote{Here the term ``experiment'' is a statistical experiment in the sense of \citet{van2000asymptotic}.} see \citet[Chap 9]{van2000asymptotic} for more on limits of experiments. Since the sample size in large-scale experimentation is typically very large, we make no distinction between $\hat{s}$ and the true scale $s := \sqrt{\sigma_t^2 / n_t + \sigma_c^2 / n_c}$, with $\sigma_t^2, \sigma_c^2$ being the population variances. In the following, we will treat the scale $s_i$ for an experiment $i$ as fixed and given. 

An experiment $i$ is \emph{represented} by the tuple $(\Delta_i, \hat{\Delta}_i, s_i)$. Following the idea of Robbins, we make the following assumption.
\begin{assumption} \label{assump:iid}
$\Delta_i \iid G$ for some distribution $G$.
\end{assumption}
This supposes an underlying population of experiments from which the experiments performed are randomly sampled. 
Then under the empirical Bayes framework, the observational model for an experiment is asymptotically an \emph{additive Gaussian noise} model with \emph{heterogenous scales}
\begin{equation}
\Delta_i \iid G, \quad \hat{\Delta}_i = \Delta_i + s_i Z_i, \label{eqs:observational}
\end{equation}
where $Z_i$ is a standard normal variable that is independent of $\Delta_i$ and independent across experiments. The prior $G$ is an unknown distribution on $\mathbb{R}$. Since the true effects $\{\Delta_i\}$ are unobserved, estimating $G$ from $\{(\hat{\Delta}_i, s_i)\}$ is a problem of deconvolution \citep{delaigle2008density}. 

We comment that this asymptotic treatment is also adopted by several authors, including \citet{deng2015objective, goldberg2017decision, azevedo2019b}. Among others, we note that \citet{deng2015objective} specifies the true prior on $\Delta_i / \tilde{s}_i$ instead of $\Delta_i$, with $\tilde{s}_i:=\sqrt{\sigma_t^2(1+n_t/n_c) + \sigma_c^2(1+n_c/n_t)}$. We find this specification not the most natural, since $\Delta_i / \tilde{s}_i$ is not an objective truth because $\tilde{s}_i$ depends on $n_t / n_c$ and hence the design of experiment. In our model \cref{eqs:observational}, the prior is only assumed on the objective quantity $\Delta_i$, while $s_i$ is treated as given. 

\subsection{Posterior analysis of experiments} It follows from \cref{eqs:observational} that if the prior $G$ is known, we can find the posterior distribution of the true effect $\Delta_i$. If $G$ has density $g$ (with respect to Lebesgue), then the posterior density is given by Bayes' rule
\begin{equation} \label{eqs:posterior-density}
p_{s_i}(\Delta_i \mid \hat{\Delta}_i) = \frac{g(\Delta_i) \phi_{s_i} (\Delta_i - \hat{\Delta}_i)}{p_{s_i}(\hat{\Delta}_i)},
\end{equation}
where 
\begin{equation} \label{eqs:marginal-density}
p_{s_i}(\hat{\Delta}_i) = \int g(u) \phi_{s_i}(\hat{\Delta}_i - u) \dd u
\end{equation}
is the marginal density for $\hat{\Delta}_i$. Technically, $s$ is not treated as random, but rather as a fixed covariate. Therefore, we write $s$ in the subscript.

Before going into the estimation of $G$, we briefly discuss what one should do if $G$ is \emph{known}. In short, we will show that if the true prior is employed in computing the posterior, then (i) the posterior has an exact frequency calibration over the sampling distribution of true effects, and (ii) the Bayes optimal decision for a loss function minimizes the expected loss with respect to that sampling distribution. Therefore, if a loss function is well-specified, then the standard Bayesian decision theory may be applied directly. The properties (i) and (ii) exactly address the issues raised in the Introduction against $t$-tests. 

\subsubsection{Frequency calibration of the posterior}
Using the true prior, the posterior probability $P_{s_i}(\Delta_i \in A \mid \hat{\Delta}_i)$ for a measurable set $A \subset \mathbb{R}$ possesses a frequency interpretation. In the following, we fix $s_i$. With $A$ also fixed, $G_{A}(x):=P_{s_i}(\Delta_i \in A \mid \hat{\Delta}_i = x): \mathbb{R} \rightarrow [0,1]$ is a measurable map from the observable to a posterior probability.

\begin{theorem} \label{thm:bayes-freq-coverage}
Given any $\varphi \in [0,1]$ and any measurable $A \subset \mathbb{R}$, if $P(G_{A}(\hat{\Delta}_i) = \varphi) > 0$, then it holds that
\begin{equation}
P \left(\Delta_i \in A \,\left|\, G_{A}(\hat{\Delta}_i) = \varphi \right.\right) = \varphi.
\end{equation}
\end{theorem}
Suppose $\varphi = 40\%$, this says that with respect to hypothetical replications over $(\Delta_i, \hat{\Delta}_i)$, out of those instances whose posterior probabilities of hypothesis $A$ is $40\%$, exactly $40\%$ of them actually have their associated $\Delta_i$ satisfying the hypothesis. We leave the proof to \cref{apx:bayes}. 
When $P(G_{A}(\hat{\Delta}_i) = \varphi) = 0$, by integrating the density over an infinitesimal interval and taking the limit, the statement can be generalized to the following. Note that the quantity should be interpreted as a conditional expectation. 
\begin{corollary}
Given any $\varphi \in [0,1]$ and any measurable $A \subset \mathbb{R}$, it holds that
\begin{equation*}
P\left(\Delta_i \in A \,\left|\, G_{A}(\hat{\Delta}_i)=\varphi \right. \right) = \varphi.
\end{equation*}
\end{corollary}

\subsubsection{Bayes optimal decision making}\label{sec:bayes-optimal-decision} When the true prior is used, the Bayes rule $d_{B}$ is optimal with respect to the sampling distribution. That is, if $d_{B}$ is used throughout, it will minimize the expected loss accumulated over future experiments; see \cref{apx:bayes} for a brief review of Bayesian decision theory.  
More concretely, as in \citet{azevedo2019b}, consider the decision over whether to ``launch'' a change/feature based on an experiment. The action space is $\mathcal{A} = \{0,1\}$, respectively denoting launching and not launching. Suppose in this context $\Delta$ denotes the gain in launching the treatment. Consider the following loss function
\begin{equation*}
l(\Delta, a) = \begin{cases} -(\Delta - C), &\quad a=1 \\
0, &\quad a=0 \end{cases},
\end{equation*}
where $C \geq 0$ is the cost for implementing the treatment. Then the Bayes optimal decision $d_{B}(\hat{\Delta})$ should minimize the posterior risk
\begin{equation*}
\E [l(\Delta, a) \mid \hat{\Delta}] = \begin{cases} -(\E[\Delta \mid \hat{\Delta}] - C), &\quad a = 1 \\
0, &\quad a=0
\end{cases}.
\end{equation*}
Clearly, we have
\begin{equation}
d_{B}(\hat{\Delta}) = \I\{\E[\Delta \mid \hat{\Delta}] > C\}, \label{eqs:bayes-rule-launch}
\end{equation}
which says that the change/feature should be launched whenever the posterior mean effect exceeds the cost.

 \section{Parametric estimation} \label{sec:parametric}
As we have seen, the key to implementing the framework is to estimate the underlying true prior $G$ with data $\{(\hat{\Delta}_i, s_i)\}$ from past experiments. A straightforward approach would be to estimate $G$ within a parametric family, which has been considered in the literature: \citet{goldberg2017decision} and \citet{azevedo2019b} considered estimating $G$ as a scaled, shifted $t$-distribution with unknown degrees of freedom; \citet{deng2015objective} modeled $G$ as a mixture of Gaussian densities with a point mass at zero. 

For completeness, in \cref{apx:sec:EM} we describe an expectation-maximization (EM) algorithm \citep{dempster1977maximum} for estimating $G$ with a mixture of $K$ Gaussian components, where $K$ is pre-specified or selected in some manner. However, we should emphasize that estimating $G$ within a parametric family is inevitably subject to \emph{model misspecification}, which motivates the development of nonparametric methods that make much weaker assumptions on the prior. It should be noted that fitting a parametric prior under additive noise is essentially fitting a latent variable model, and hence the log-likelihood is usually non-convex. The EM algorithm is often used, but it only converges to a local maximum, can have an extremely slow rate of convergence (when $K$ is over-specified) \citep{dwivedi2018singularity}, and can be numerically unstable as well \citep{archambeau2003convergence}. In the next Section, we propose a nonparametric estimator based on convex optimization, which is free from these issues.

 \section{Nonparametric spectral estimation} \label{sec:nonparametric}
In this Section, we develop a nonparametric estimation strategy based on a spectral characterization of the problem. Before describing the characterization, we first review empirical Bayes formulae that relate posterior cumulants to the derivatives of marginal densities. 

\subsection{Tweedie's formula}
Certain posterior quantities are of special interest to a decision maker. For example, as we saw from \cref{eqs:bayes-rule-launch}, the Bayes optimal decision rule for launching was determined by the posterior mean. Under our modeling assumption $\Delta \sim G, \ \hat{\Delta} \mid \Delta \sim \N(\Delta, s^2)$, \citet{robbins1956empirical} presented the following formula due to Maurice Tweedie
\begin{equation}
\E[\Delta \mid \hat{\Delta}] = \hat{\Delta} + s^2 \ell_{s}'(\hat{\Delta}), 
\end{equation}
where $\ell_s(\hat{\Delta}) := \log p_{s}(\hat{\Delta})$ is the logarithmic marginal density for the observed effect under a given scale $s$. The two terms on the RHS are, respectively, the unbiased estimate of $\Delta$ and a Bayesian shrinkage term that pulls the estimate towards a marginal mode ($\ell'_s(\hat{\Delta}) > 0$ if $\hat{\Delta}$ is on the left of a mode). Moreover, the formula only depends on the prior $G$ implicitly through the score (with respect to location) of the marginal density. As a result, one can estimate the posterior mean by directly estimating the marginal density, which circumvents estimating $G$ in the first place --- this strategy is called ``$f$-modeling'' by \citet{efron2014two}, as opposed to ``$g$-modeling'' that starts with estimating the prior. 

Robbins generalized \Cref{eqs:tweedie-mean} to exponential families.

\begin{lemma}[{\citet{robbins1956empirical}}] \label{lem:tweedie-robbins}
Suppose $\eta \sim G$ with density $g$ and 
\begin{equation*}
X \mid \eta \sim f_{\eta}(x) = e^{\eta x - \psi(\eta)} f_{0}(x),
\end{equation*}
where $f_0(x)$ is the conditional density $f_{\eta}(x)$ when $\eta = 0$, and $\psi(\eta)$ the cumulant generating function that ensures normalization. Let $f(x) = \int f_{\eta}(x) g(\eta) \dd \eta$ be the marginal density for $X$, and define $\lambda(x) = \log(f(x) / f_0(x))$. Then for $k=1,2,\dots$, the $k$-th posterior cumulant of $\eta$ given $X$ is the $k$-th derivative of $\lambda(x)$ evaluated at $X$. 
\end{lemma}

Specifically for $k=1,2$, it follows that
\begin{equation} \label{eqs:tweedie-first-two}
\E[\eta \mid X] = \lambda'(X), \quad \var[\eta \mid X] = \lambda''(X). 
\end{equation}
We specialize this result to the Gaussian case. 

\begin{lemma}[Tweedie's formulae] \label{lem:tweedie}
Under the additive Gaussian noise model \cref{eqs:observational}, we have
\begin{equation} \label{eqs:tweedie-mean}
\E_{s}[\Delta \mid \hat{\Delta}] = \hat{\Delta} + s^2 \ell_s'(\hat{\Delta}), 
\end{equation}
and
\begin{equation} \label{eqs:tweedie-var}
\var_{s}[\Delta \mid \hat{\Delta}] = s^2 \left(1 + s^2 \ell_s''(\hat{\Delta}) \right).
\end{equation}
\end{lemma}
\begin{proof}
In the context of \cref{lem:tweedie-robbins}, let $\eta = \Delta / s^2$. $G$ in the lemma becomes our ``$G$'' scaled by $s^{-2}$. The conditional density for $\hat{\Delta} \mid \eta $ can be expressed as 
\begin{equation*}
f_{\eta}(\hat{\Delta}) = \phi_{s} (\hat{\Delta} - \Delta) = \exp(\eta \hat{\Delta} - \eta^2 s^2 / 2) f_0(\hat{\Delta}), \quad f_0(\hat{\Delta}) = \phi_{s}(\hat{\Delta}).
\end{equation*}
Therefore, 
\begin{equation*}
\lambda(\hat{\Delta}) = \log f(\hat{\Delta}) - 
\log f_0(\hat{\Delta}) = \ell_{s}(\hat{\Delta}) + \hat{\Delta}^2 / (2 s^2) + \text{const},
\end{equation*}
where $\ell_s(\hat{\Delta}) = \log p_s(\hat{\Delta})$ from \cref{eqs:marginal-density}. It follows from \cref{eqs:tweedie-first-two} that 
\begin{equation*}
\E[\Delta \mid \hat{\Delta}] = s^2 \E[\eta \mid \hat{\Delta}] = s^2 \lambda'(\hat{\Delta}) = \hat{\Delta} + s^2 \ell_s(\hat{\Delta}),
\end{equation*}
and
\begin{equation*}
\var[\Delta \mid \hat{\Delta}] = \var[s^2 \eta \mid \hat{\Delta}] = s^4 \lambda''(\hat{\Delta}) = s^2(1 + s^2 \ell''_s(\hat{\Delta})).
\end{equation*}
\end{proof}

Tweedie's formulae suggest that the family of marginal densities $\{p_s(\hat{\Delta}): s \geq 0\}$ is a natural estimand, from which we can obtain estimates for the posterior mean and variance. This $f$-modeling (marginal-modeling) strategy is adopted by \citet{efron2003robbins,efron2011tweedie,efron2014two}, where he modeled $\log p(x)$ as a linear function of polynomials in $x$, fitted with a simple Poisson regression via Lindsey's method \citep{efron1996using}. However, his strategy is not applicable here since we have a density $p_s$ to be estimated for \emph{every} $s \geq 0$. The \emph{heterogeneity} in the likelihood poses a new challenge to $f$-modeling, which, to our knowledge, has not been addressed previously. Besides, as noted by \citet{wager2014geometric}, Efron's method does not constrain the modeled marginal densities to those representable under \cref{eqs:observational}, and hence is subject to misspecification and inefficiency.

\subsection{Characterization via the heat equation}
Estimating the density family $\{p_s: s \geq 0\}$ subject to our modeling assumption is equivalent to estimating the prior, since the prior density $g = p_0$. In the literature, recovering $G$ from observations with additive noise, with known or unknown heterogenous scales has been considered by \citet{delaigle2008density}, who proposed a deconvolution kernel estimate based on inverting an empirical characteristic function \citep{carroll1988optimal,StefanskiCarroll1990,fan1991asymptotic,fan1991global}. In addition to the usual problem of bandwidth selection, inverting an empirical estimate is not the most efficient. To improve efficiency, we consider estimating $\{p_s: s \geq 0\}$ based on maximum likelihood. In the following, we first answer a prerequisite question --- what is the space of all density families representable under the additive Gaussian noise model?

For ease of exposition, we make a change of variable $s = \sqrt{t}$ for $t \geq 0$, and reindex $p_t := p_{s=\sqrt{t}}$ whenever the letter $t$ is used. The modeling assumption poses non-trivial constraints on the density family $\{p_t: t \geq 0\}$. For example, the first two moments are related by
\begin{equation*}
\E_{t} \hat{\Delta} = \text{const}, \quad \var_{t} \hat{\Delta} = \var \Delta + t.
\end{equation*}
More generally, we have the following characterization.

\begin{theorem} \label{thm:heat}
The density family $\{p_t: t \geq 0\}$ on $\mathbb{R}$ can be represented as
\begin{equation}
p_t(x) = \int \phi_{\sqrt{t}}(x-u) g(u) \dd u \label{eqs:green}
\end{equation}
for a probability density $g$ on $\mathbb{R}$ if and only if
\begin{equation}
\frac{\partial}{\partial t} p_t(x) = \frac{1}{2} \frac{\partial^2}{\partial x^2} p_t(x), \quad t\geq0,\  x\in \mathbb{R}, \label{eqs:heat}
\end{equation}
and $p_0 = g$.
\end{theorem}

\Cref{eqs:heat} is the \emph{heat equation} with thermal diffusivity constant $1/2$, which describes how heat diffuses on $\mathbb{R}$ as time $t$ elapses from zero; see \cref{fig:diffusion-example} for an example. As a result, $p_t(x)$ obeys the \emph{maximum principle} \citep[page 54]{Evans2010PDE}, which says that the maximum of $p_t(x)$ within any interval will not exceed the maximum value previously encountered; and as a consequence, the number of local maxima in $p_t(x)$ is non-increasing in $t$. \cref{eqs:green} is the solution to the heat equation with \emph{initial condition} $p_0 = g$. The solution is based on convolving the initial condition with a Gaussian kernel with variance $t$, which is the Green function (also called the fundamental solution) to the heat equation, i.e., the solution when the initial condition is a delta function. 

\begin{figure}[!ht]
\includegraphics[width=0.65\textwidth]{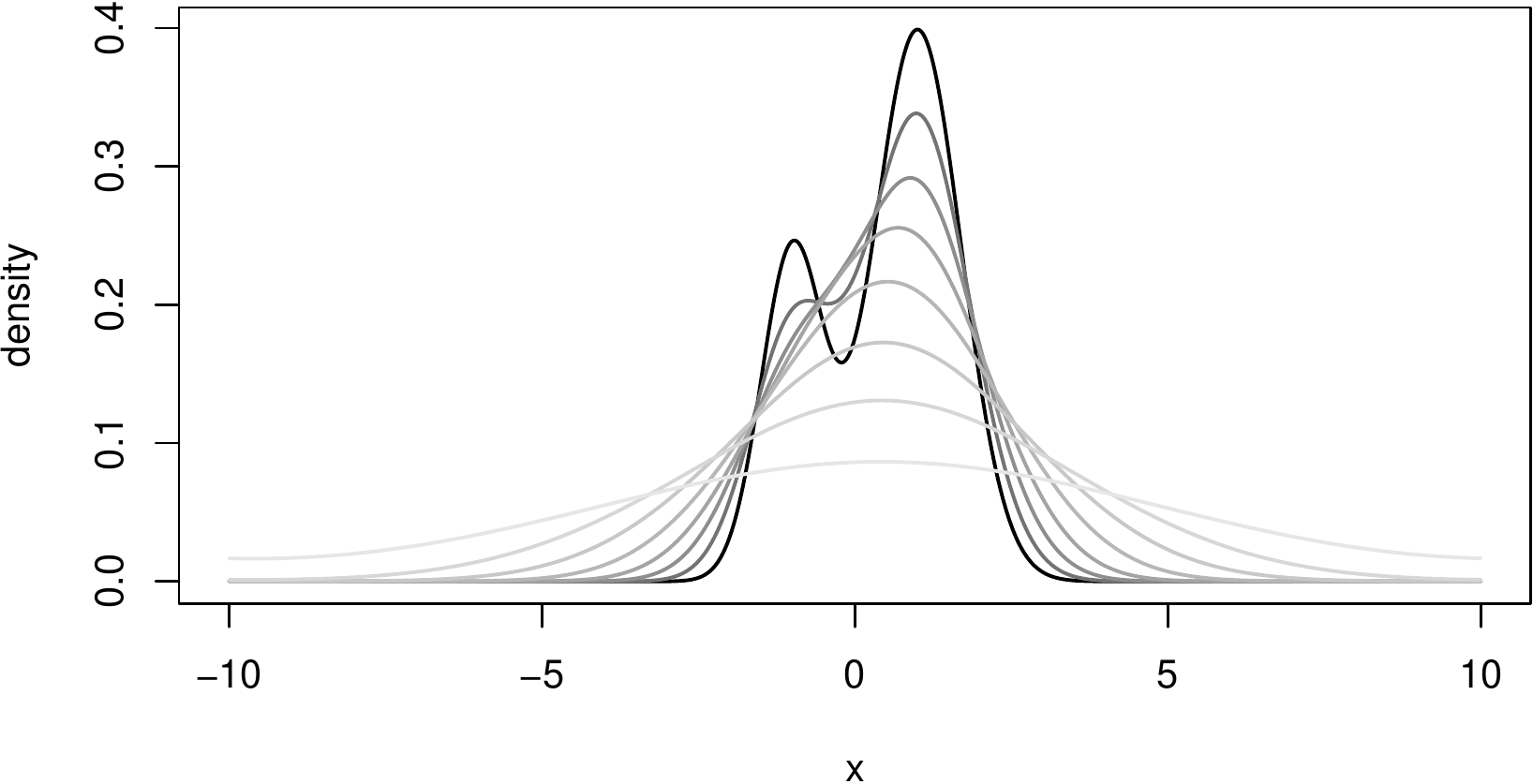}
\caption{Example of heat diffusion by \cref{eqs:heat}. $p_0 = g$ is a bimodal distribution (black), and over time it becomes smoother and flatter (lighter-colored as time elapses).}
\label{fig:diffusion-example}
\end{figure}

\begin{proof}
We need to show that \cref{eqs:green} is the unique solution to \cref{eqs:heat} with initial condition $p_0 = g$. We first verify that it is a solution. Fix $u \in \mathbb{R}$, consider $\varphi_{u}(x,t) := \phi_{\sqrt{t}}(x-u)$. We check that $\varphi_{u}(x,t)$ satisfies \cref{eqs:heat}. We have
\begin{equation*}
\frac{\partial}{\partial t} \varphi_{u}(x,t) = \frac{1}{2} \varphi_{u}(x,t) \left\{ \left( \frac{x-u}{t} \right)^2 - \frac{1}{t} \right\},
\end{equation*}
\begin{equation*}
\frac{\partial}{\partial x} \varphi_{u}(x,t) = - \varphi_{u}(x,t) \left (\frac{x - u}{t} \right),
\end{equation*}
and it follows that
\begin{equation*}
\frac{\partial^2}{\partial x^2} \varphi_{u}(x,t) = \varphi_{u}(x,t) \left \{ \left( \frac{x-u}{t} \right)^2 - \frac{1}{t} \right \} = 2 \frac{\partial}{\partial t} \varphi_{u}(x,t).
\end{equation*}
Then by linearity, $p_t(x) = \int \phi_{u}(x,t) g(u) \dd u$ also satisfies \cref{eqs:heat}. Also, $p_0(x) = \lim_{t \rightarrow 0+} p_t(x) = g(x)$. Further, by \citet[Theorem 7, Sec.~2.3]{Evans2010PDE} $p_t(x) = \int \phi_{u}(x,t) g(u)$ is the unique solution such that $p_t(x)$ is a density for every $t \geq 0$. 
\end{proof}

The general solution to the heat equation \cref{eqs:heat} can be characterized with Fourier transforms; see \cref{apx:fourier} for details. 
However, since the Fourier transform is infinite dimensional, representing the transform is as difficult as representing the density $g$ itself. To proceed, we need to truncate the domain from $\mathbb{R}$ to a compact interval, upon which the Fourier transform becomes a Fourier series. 

We remark that the connection between density estimation and the heat equation is also made by \citet{botev2010kernel} and previously by \citet{chaudhuri2000scale}. Contrary to our case, where $t$ corresponds to the variance of an additive noise, their consideration of $t$ is to optimize the amount of smoothing applied to the empirical measure for constructing a density estimator. 

\subsection{Toric formulation}
For tractability, we consider Fourier series representations, which are only associated with \emph{periodic} functions with a fixed period. By \cref{thm:heat}, $p_0$ equals the prior density $g$.

\begin{assumption} \label{assump:prior}
The prior $G$ has a uniformly continuous density $p_0(x)$ with respect to Lebesgue, supported on a subset of $[-L, L]$ for $0 < L < \infty$.
\end{assumption}
Recall that $p_0(x)$ is uniformly continuous if for every $\varepsilon > 0$, there exists $\delta > 0$ such that for all $x, y \in [-L,L]$ such that $|x-y| < \delta$, we have $|p_0(x) - p_0(y)| < \varepsilon$. Uniform continuity is a weak smoothness condition, which is implied by Lipschitz or H\"older continuity, and some smoothness assumption is usually required to establish consistency. Compact support is assumed for technical convenience (e.g., \citet{StefanskiCarroll1990} and \citet{meister2007deconvolving}), although it can be relaxed by letting $L$ grow with $n$. The domain is assumed to be symmetric about zero, which can be ensured by translation. Since we work with Fourier series, let us introduce $\bar{p}_0$ as the periodicized version of $p_0$:
\begin{equation}
\bar{p}_0(x) = p_0(x) \text{ for $x \in [-L, L]$}, \quad \bar{p}_0(x + 2L) = p_0(x) \text{ for all $x \in \mathbb{R}$}.
\end{equation}
Hence, any interval of length $2L$ is a \emph{torus}, and for convenience we will mostly use the torus $[-L,L]$. Next, we introduce an assumption that reduces the data generating process from $\mathbb{R}$ to the torus. 

\begin{assumption} \label{assump:torus}
The data is generated from $x_i \sim p_{t_i}$ independently, where $p_{t}$ is a probability density on $[-L,L]$ 
\begin{equation} \label{eqs:pt}
p_t(x) := \int_{-\infty}^{\infty} \bar{p}_0(u) \phi_{\sqrt{t}}(x-u) \dd u. 
\end{equation}
\end{assumption}
This a \emph{simplifying assumption} on the data generating mechanism, such that a heat diffusion on $\mathbb{R}$ is reduced to a heat diffusion on a compact domain with periodic boundary conditions. Note that the true data generating mechanism is $x_i \sim p_{t_i}^{\ast}$ with the following density on $\mathbb{R}$
\begin{equation} \label{eqs:true-pt}
p_{t}^{\ast}(x) = \int_{-L}^{L} p_0(u) \phi_{\sqrt{t}}(x - u) \dd u.
\end{equation}
Notice that $p_t$ and $p_t^{\ast}$ differ in terms of their domain of integration. It is easy to check that $p_t(x) \geq 0 $, has period $2L$, and integrates to one on $[-L,L]$:
\begin{equation*}
\begin{split}
\int_{-L}^{L} p_t(x) \dd x &= \int_{-L}^{L} \left \{\int_{-\infty}^{+\infty} \bar{p}_0(u)  \phi_{\sqrt{t}}(x-u) \dd u \right \} \dd x \\
&= \int_{-L}^{L} \left\{ \int_{-\infty}^{+\infty} \bar{p}_0(x - u) \phi_{\sqrt{t}}(u) \dd u  \right \} \dd x \\
&= \int_{-\infty}^{+\infty} \left( \int_{-L}^{L} \bar{p}_0(x - u) \dd x \right )  \phi_{\sqrt{t}}(u) \dd u = 1.
\end{split}
\end{equation*}
Further, the difference between $p_t$ and $p_t^{\ast}$ becomes negligible with a large enough $L$.  Let $t_{\max}$ be an upper bound on the variance $t$.

\begin{lemma} \label{lem:aliasing-error}
Given $t_{\max} < \infty$, for all $x \in [-L, L]$ and $t \in [0, t_{\max}]$ it holds that
\begin{equation}
0 < p_t(x) - p_t^{\ast}(x) < \sqrt{\frac{2}{\pi}} \frac{1}{\sqrt{t_{\max}} \left \{\exp(2L^2 / t_{\max}) - 1 \right\}}.
\end{equation}
\end{lemma}
\begin{proof}
By breaking the integral in \cref{eqs:pt} into intervals of length $2L$, we have
\begin{equation} \label{eqs:bound-pt-pt-true}
\begin{split}
p_t(x) &= \sum_{k=-\infty}^{+\infty} \int_{-L}^{L} \bar{p}_0(u - 2k L) \phi_{\sqrt{t}} (x - u - 2 k L ) \dd u \\
&= \sum_{k=-\infty}^{+\infty} \int_{-L}^{L} p_0(u) \phi_{\sqrt{t}} (x - u - 2 k L ) \dd u\\
& = p_t^{\ast}(x) + \int_{-L}^{L} p_0(u) \left \{\sum_{k \neq 0} \frac{1}{\sqrt{t}} \phi \left ( \frac{x - u - 2 k L}{\sqrt{t}} \right) \right\} \dd u,
\end{split}
\end{equation}
where we used the fact that $\bar{p}_0(u)$ has period $2L$. Define $H_r(x) := \sum_{k \neq 0} \phi(x - k r)$. Clearly for all $x \in \mathbb{R}$, we have
\begin{equation*}
\begin{split}
H_r(x) \leq H_r(0) = 2 \sum_{k=1}^{\infty} \phi(k r) &= \sqrt{\frac{2}{\pi}} \sum_{k=1}^{\infty} \exp(- k^2 r^2 / 2) \\
&< \sqrt{\frac{2}{\pi}} \sum_{k=1}^{\infty} \exp(- k r^2 / 2) \\
&= \sqrt{\frac{2}{\pi}} \frac{1}{\exp(r^2 / 2) - 1}.
\end{split}
\end{equation*}
Applying this to \cref{eqs:bound-pt-pt-true} with $r = 2L / \sqrt{t}$ and using $\int_{-L}^{L} p_0(u) \dd u =1$, we have
\begin{equation*}
0 < p_t(x) - p_t^{\ast}(x) < \sqrt{\frac{2}{\pi}} \frac{1}{\sqrt{t}\{\exp(2L^2/t) - 1\}}.
\end{equation*}
Let $\gamma(t) := \sqrt{t}\{\exp(2L^2/t) - 1\}$. We claim that $\gamma(t)$ is decreasing in $t \geq 0$. To see this, note that $\gamma'(t) = 2t^{-3/2}\left\{e^{2L^2/t}(t-4L^2) - t \right\}$. By rearranging, $\gamma'(t) < 0$ iff $e^{-2L^2/t} > 1 - 4L^2/t$, which is true by $e^{-2L^2/t} \geq 1 - 2L^2/t$. Therefore, for $t \in [0,t_{\max}]$ the RHS of the previous display is lower bounded by its value at $t_{\max}$. 
\end{proof}

By \cref{lem:aliasing-error}, one can always use a large enough $L$ such that \cref{assump:torus} approximately holds for the data at hand. In the following, we will work under \cref{assump:iid,assump:prior,assump:torus}. Furthermore, since a rescaling $x \leftarrow \pi x / L$ can be applied, without loss of generality, we suppose $L = \pi$ for convenience.

Under \cref{assump:prior,assump:torus}, \cref{thm:heat} still holds if we replace the domain $\mathbb{R}$ with the torus $[-\pi, \pi]$ (or any torus of length $2\pi$). 

\begin{lemma} \label{lem:heat-series}
The set of solutions, in the form of periodic functions in $x$ with period $2 \pi$, to the heat equation \cref{eqs:heat} with initial condition $\bar{p}_0$, has the following trigonometric series representation
\begin{equation}
p_t(x) \sim \sum_{k=-\infty}^{\infty} c_k e^{-k^2 t / 2} e^{ikx}, \quad c_k = \frac{1}{2 \pi} \int_{-\pi}^{\pi} p_0(x) e^{-ikx} \dd x.
\end{equation}
\end{lemma}

The symbol ``$\sim$'' means ``associated with'', since one needs to specify in what sense the series converges to the function $p_t(x)$. For example, a usual notion is the convergence in $L^2$; besides, in a pointwise sense, if $p_0(x)$ is continuous and the series converges uniformly in $x$, then ``$\sim$'' can be replaced by ``$=$''; see \citet[Theorem 6.3, Chapter I]{zygmund2002trigonometric}.

$\{c_k\}$ are complex Fourier coefficients of $p_0(x)$ satisfying complex conjugacy $c_{-k} = \overline{c_k}$, which holds since $p_0(x)$ is real. With $a_k:=2 \Re c_k$ and $b_k:=-2 \Im c_k$, the series in the previous display can be rewritten as
\begin{equation} \label{eqs:partial-sum}
p_t(x) \sim \frac{1}{2} a_0 + \sum_{k=1}^{\infty} (a_k \cos k x + b_k \sin k x) e^{-k^2 t / 2}.
\end{equation}

\subsection{Modeling with trigonometric polynomials}
Now we consider modeling $p_0(x)$, and hence $p_t(x)$ by \cref{lem:heat-series}, in terms of a partial sum of trigonometric polynomials. A trigonometric polynomial of order $N$ takes the form of
\begin{equation} \label{eqs:S-N}
S_N(x) = \sum_{k=-N}^{N} c_k e^{ikx} = \frac{1}{2} a_0 + \sum_{k=1}^{N} (a_k \cos kx + b_k \sin kx)
\end{equation}
for $(2N+1)$ complex Fourier coefficients $\{c_k\}$ satisfying $c_{-k} = \overline{c_k}$. In fact, $\{c_k\}$ are the Fourier coefficients associated with function $S_N(x)$, and are hence uniquely determined from $S_N(x)$. 

Let us fix an order $N$. Suppose $p_0(x) = S_N(x)$, by \cref{lem:heat-series}, the entire density family remains within trigonometric polynomials with order $N$, given by the \emph{closed form} expression
\begin{equation} \label{eqs:SN-family}
S_N(x; t) := \sum_{k=-N}^{N} c_k e^{-k^2 t / 2} e^{ikx}, \quad t \geq 0,
\end{equation}
where $S_N(x; 0) = S_N(x)$. 
In addition to $c_{-k} = \overline{c_k}$ for $k=1,\dots,N$, the coefficients are further subject to an equality constraint due to normalization
\begin{equation} \label{eqs:normalization-c0}
1 = \int_{-\pi}^{\pi} S_N(x) \dd x = (2\pi) c_0.
\end{equation}

Taking the normalization constraint into account, $S_N$ has $2N$ degrees of freedom. $S_N$ can also be uniquely parametrized by $(2N+1)$ interpolating nodes $\{(x_\nu, f_\nu): \nu=0,\pm 1,\dots,\pm N\}$ such that $S_N(x_\nu) = f_\nu$ for every $\nu$. In particular, we focus on equidistant nodes 
\begin{equation} \label{eqs:nodes}
x_\nu = \frac{2\pi \nu}{2N+1}, \quad \nu=0,\pm 1,\dots, \pm N
\end{equation}
on the torus $[-\pi , \pi]$. Given any real-valued $\{f_\nu\}$, there is a unique $S_N(x)$ that passes through every $(x_\nu, f_\nu)$, given by
\begin{equation} \label{eqs:dir-interpolation}
S_N(x) = \frac{1}{2N+1} \sum_{\nu=-N}^{N} f_\nu D_N(x - x_\nu),
\end{equation}
where $D_N(u)$ is the Dirichlet kernel of order $N$
\begin{equation} \label{eqs:dir-kernel}
D_N(u) = \sum_{k=-N}^{N} e^{i k u} = 1 + 2 \sum_{k=1}^{N} \cos ku = \frac{\sin((N + 1/2) u)}{\sin u/2},
\end{equation}
with $D_0(u) \equiv 1$. 
From these two expressions, one can easily check that $S_N(x)$ is an order-$N$ trigonometric polynomial. To see that $S_N(x)$ passes through every interpolating point, observe that $D_N(0) = 2N+1$ and $D_N(x_\nu) = 0$ for $\nu \neq 0$. By comparing the trigonometric terms between \cref{eqs:dir-interpolation} and \cref{eqs:S-N}, the Fourier coefficients are recovered as
\begin{equation} \label{eqs:DFT}
c_k = \frac{1}{2N+1} \sum_{\nu = -N}^{N} f_{\nu} e^{-\frac{2\pi}{2N+1} i k \nu}, \quad k=0,\pm 1,\dots,\pm N.
\end{equation}
In particular, the normalization constraint \cref{eqs:normalization-c0} translates to
\begin{equation*}
\frac{1}{2 \pi} = c_0 = \frac{1}{2N+1} \sum_{\nu=-N}^{N} f_{\nu}.
\end{equation*}
The full-rank linear transform \cref{eqs:DFT} is known as the discrete Fourier transform (DFT), which we will denote with shorthand $\{c_k\} = \text{DFT}_{N}(\{f_\nu\})$. $\text{DFT}_N$ and its inverse transform can be computed in $O((2N+1) \log (2N+1))$ time with fast Fourier transform (FFT) algorithms. To sum up, an order-$N$ trigonometric polynomial model of the density family $\{S_N(x; t): t \geq 0\}$ can be parametrized by $(2N+1)$ real-valued $\{f_\nu\}$ such that $\frac{2\pi}{2N+1} \sum_{\nu=-N}^{N} f_\nu = 1$. 

We remark that among other orthogonal polynomials, the trigonometric polynomials are the most natural choice for our problem. They are the eigenfunctions of the heat equation \cref{eqs:heat} and provide an orthogonal basis for representing $p_t(x)$ in closed form. Estimating a density with trigonometric polynomials/series was proposed by \citet{kronmal1968estimation} and has since been extensively studied; see also \citet{wahba1975optimal,walter1979probability,hall1981trigonometric,hall1986rate}. The estimator considered by these authors is primarily a Fourier partial sum sequence with $\{\hat{c}_k: k=\pm 1, \cdots, \pm N\}$ unbiasedly estimated by the sample average of the corresponding $e^{ikx}$. The performance of the estimator, as $N$ grows with $n$ properly, is analyzed in terms of the mean integrated square error relative to the true density. In the following, we show that maximum likelihood estimation can be performed on $S_N$ to improve efficiency. 

\subsubsection{Non-negativity} An obvious drawback of this type of model is that an estimated $S_N(x)$ is not guaranteed to be non-negative for a finite $N$, even if one places constraints $f_k \geq 0$. In the literature, the non-negativity constraint is sometimes relaxed to obtain a faster rate of convergence \citep{terrell1980improving}, or dismissed by arguing that the estimate will stay positive ``most of the time'' \citep{kronmal1968estimation}. Post-processing of an estimator to ensure non-negativity is also proposed; see \citet{kronmal1968estimation,hall1981trigonometric,gajek1986improving}.

In the context of ours, we find this issue particularly disturbing: (i) a negative value in \cref{eqs:SN-family} yields an invalid objective for maximum-likelihood estimation; (ii) $S_N(x; t)$ can have up to $2N$ zeros in $[-\pi, \pi]$ \citep[Theorem 10.1.7]{zygmund2002trigonometric}, and the posterior mean, by Tweedie's formula \cref{eqs:tweedie-mean}, will shoot to \emph{infinity} at these zeros because $\ell'(x) = S_N'(x;t) / S_N(x;t)$; (iii) $S_N(\hat{\Delta}; t)$ is the normalizing constant, when computing the posterior density of the true effect is desired. In the following, we present a simple parametrization that ensures a \emph{bona fide} density. 

Consider the Ces\`{a}ro sum of trigonometric polynomials
\begin{equation}
C_N(x) := \frac{S_0(x) + S_1(x) + \dots + S_N(x)}{N+1},
\end{equation}
which is the arithmetic mean of the first $N+1$ partial sums given by \cref{eqs:S-N} ($S_0(x) \equiv \frac{1}{2} a_0$). The Ces\`{a}ro sum $C_N(x)$ has stronger convergence guarantee compared to the symmetric partial sum $S_N(x)$. Fej\'{e}r's theorem \citep[Section 3.3]{zygmund2002trigonometric} states that as $N \rightarrow \infty$, $C_N(x)$ converges uniformly to $p_0(x)$ if $p_0$ is continuous. By definition, $C_N(x)$ is also a trigonometric polynomial of order $N$
\begin{equation} \label{eqs:cesaro}
C_N(x) = \sum_{k=-N}^{N} \left(1 - \frac{|k|}{N+1} \right) c_k e^{ikx},
\end{equation}
with complex Fourier coefficients reweighted by $1 - |k|/(N+1)$ accordingly. Similarly, let $C_N(x; t)$ be the arithmetic mean of $\{S_m(x; t): m=0,\dots,N\}$ given by \cref{eqs:SN-family}. By linearity, $C_N(x; t)$ satisfies the heat equation \cref{eqs:heat}. 

\begin{definition} \label{def:Fejer}
The Fej\'{e}r kernel is 
\begin{equation} 
K_N(u) := \frac{1}{N+1} \sum_{k=0}^{N} D_k(u),
\end{equation}
where $D_k(\cdot)$ is the Dirichlet kernel given by \cref{eqs:dir-kernel}.
\end{definition}

\begin{lemma}[{\citet[Sec. 3, Chapter III]{zygmund2002trigonometric}}] \label{lem:fejer}
The Fej\'{e}r kernel has the following properties.
\begin{enumerate}[(a)]
\item It holds that
\begin{equation} \label{eqs:fejer}
K_N(u) = \frac{1}{N+1} \left(\frac{\sin \frac{N+1}{2} u}{\sin u/2} \right)^2,
\end{equation}
where in particular $K_N(0) := \lim_{u \rightarrow 0} K_N(u) = N+1$. 
\item $K_N(u) \geq 0$.
\item $\int_{-\pi}^{\pi} K_N(u) \dd x = 2\pi$.
\item For $u \in [-\pi, \pi]$, $K_N(u)$ vanishes iff $u = \frac{2 \pi m}{N+1}$ for $m=\pm 1, \dots, \pm \lfloor \frac{N+1}{2} \rfloor$. 
\end{enumerate}
\end{lemma}

\begin{lemma} \label{lem:fejer-interpolation}
With $\{(x_\nu, f_\nu)\}_{\nu=-N}^{N}$ given by \cref{eqs:nodes,eqs:dir-interpolation}, $C_N(x)$ in \cref{eqs:cesaro} can be expressed as
\begin{equation}
C_N(x) = \frac{1}{2N+1} \sum_{\nu=-N}^{N} f_\nu K_N(x - x_\nu). \label{eqs:cesaro-kernel-form}
\end{equation}
\end{lemma}
\begin{proof}
Using $S_N(x) = \sum_{k=-N}^{N} c_k e^{-ikx}$ with $\{c_k\}$ given by \cref{eqs:DFT}, $C_N(x)$ in \cref{eqs:cesaro} can be expressed as
\begin{equation*}
\begin{split}
C_N(x) &= \frac{1}{2N+1} \sum_{k=-N}^{N} \left(1 - \frac{|k|}{N+1} \right) \sum_{\nu=-N}^{N} f_{\nu} e^{- \frac{i 2\pi k \nu}{2N+1}} e^{ikx} \\
&= \frac{1}{2N+1} \sum_{\nu=-N}^{N} f_{\nu} \left \{1 + \sum_{k=1}^{N} \left(1 - \frac{k}{N+1} \right) (e^{ik(x-x_{\nu})} + e^{-ik(x-x_{\nu})}) \right\} \\
&= \frac{1}{2N+1} \sum_{\nu=-N}^{N} f_{\nu} \left\{ 1 + 2 \sum_{k=1}^{N} \left(1 - \frac{k}{N+1} \right) \cos k(x - x_{\nu}) \right\},
\end{split}
\end{equation*}
where $x_{\nu} = \frac{2\pi \nu}{2N+1}$. By \cref{def:Fejer}, we also have
\begin{equation*}
K_N(u) = \frac{1}{N+1} \sum_{k=0}^{N} D_k(u) = \frac{1}{N+1} \left(N+1 + 2 \sum_{k=1}^{N} (N+1-k) \cos ku \right),
\end{equation*}
where we used \cref{eqs:dir-kernel} for $D_k(u)$. The lemma is proven by comparing the previous two displays.
\end{proof}

\begin{theorem} \label{thm:non-negative}
Suppose $S_N(x)$ is parametrized by $\{f_\nu\}$ satisfying $f_\nu \geq 0$ and $\frac{2\pi}{2N+1} \sum_{\nu=-N}^{N} f_\nu = 1$, then $C_N(x)$ in \cref{eqs:cesaro} is a valid density on $[-\pi, \pi]$. Further, if there exist $\nu \neq \nu'$ with $f_\nu, f_{\nu'}>0$, then $C_N(x)$ is strictly positive. 
\end{theorem}
\begin{proof}
From (b) and (c) of \cref{lem:fejer}, we know $\frac{1}{2\pi} K_N(u)$ is a density on $[-\pi, \pi]$. It follows from \cref{lem:fejer-interpolation} that $C_N(x)$ is a mixture of $\{K_N(u - x_\nu)\}$ with corresponding weights $\{\frac{2 \pi}{2N+1} f_\nu\}$. Given $f_\nu \geq 0$ and $\frac{2\pi}{2N+1} \sum_{\nu=-N}^{N} f_\nu = 1$, $C_N(x)$ is a density on $[-\pi, \pi]$. Further, suppose $f_\nu, f_{\nu'} > 0$ for $\nu \neq \nu'$. We prove by contradiction that $C_N(x)$ is positive. Suppose $C_N(y) = 0$. By \cref{lem:fejer}(b), $K_N(y - x_\nu)$ and $K_N(y - x_{\nu'})$ must both vanish. By \cref{lem:fejer}(d), we have $y - x_\nu = \frac{2\pi m}{N+1}$ and $y - x_{\nu'} = \frac{2 \pi m'}{N+1}$ for some non-zero integers $m, m'$. By taking the difference on \cref{eqs:nodes}, we have $\frac{2 \pi (\nu'-\nu)}{2N+1} = \frac{2 \pi (m' - m)}{N+1}$, namely $\frac{N+1}{2N+1}(\nu'-\nu) = m' - m$. We observe that $N+1$ and $2N+1$ are co-prime. Given that $\nu' \neq \nu$, the only possibility is that $\nu' - \nu$ cancels with $2N+1$. However, $|\nu' - \nu| \leq 2N$. By contradiction we conclude that $C_N(x)$ stays positive. 
\end{proof}

\begin{remark}
$C_N(x)$ is a generalized form of Jackson polynomial $J_{N, 2N+1}(x)$ based on ``nodes'' $\{(x_\nu, f_\nu)\}$; see \citet[Sec.~6, Chapter X]{zygmund2002trigonometric}. However, unlike $S_N(x)$, $C_N(x)$ does not pass through these ``nodes''. 
\end{remark}

To summarize, we have shown that a valid density family can be modeled by a family of trigonometric polynomials of order $N$
\begin{equation} \label{eqs:parametrization}
\begin{split}
& C_N(x; t) = \sum_{k=-N}^{N} \left(1 - \frac{|k|}{N+1} \right) c_k e^{-k^2 t/2} e^{ikx}, \quad \{c_k\} = \text{DFT}_{N}(\{f_\nu\}) \\
& \text{subject to} \ f_\nu \geq 0 \  (\nu=0,\pm 1, \dots,\pm N) \ \text{and } \ \frac{2\pi}{2N+1} \sum_{\nu=-N}^{N} f_\nu = 1.
\end{split}
\end{equation}
This density family satisfies the heat equation in \cref{thm:heat} with the domain replaced by the torus $[-\pi, \pi]$.

 \subsection{Maximum likelihood estimation} 
We are now ready to estimate the density family with trigonometric polynomials. For the sake of efficiency, we consider maximum likelihood estimation for the model in \Cref{eqs:parametrization}.
Let $f := \{f_\nu\}_{\nu=-N}^{N}$ be the parameter vector and let $C_{N,f}(x;t)$ be $C_N(x;t)$ defined through $f$. The log-likelihood of the dataset is
\begin{equation}
\ell_{n}(f) = \sum_{i=1}^{n} \log C_{N,f}(\hat{\Delta}_i; s_i^2).
\end{equation}
Since (i) $-\log(\cdot)$ is convex, (ii) $C_{N,f}(x; t)$ is linear in $\{c_k\}$, and (iii) $\{c_k\} = \text{DFT}_N(f)$ is linear and full-rank, indeed $(-\ell_n)$ is convex in $f$. Let $\mathcal{S}_{2N} \subset \mathbb{R}^{2N+1}$ denote the $2N$-dimensional unit simplex. The MLE can thus be obtained from the following convex optimization
\begin{equation} \label{eqs:opt}
\min \left\{-\sum_{i=1}^{n} \log C_{N,f}(\hat{\Delta}_i; s_i^2): \quad \kappa_N f \in \mathcal{S}_{2N} \right\},
\end{equation}
where $\kappa_N := 2\pi/(2N+1)$ is a constant. 

\subsubsection{Accelerated projected gradient} The log-likelihood is continuously differentiable in $f$. The constrained convex optimization can be solved by the more general proximal gradient method, formulated as
\begin{equation*}
\min_{f \in \mathbb{R}^{2N+1}} \ell_n(f) + \delta_{\kappa_N^{-1} \mathcal{S}_{2N}}(f),
\end{equation*}
where the indicator $\delta_{A}(x)$ takes $+\infty$ when $x \notin A$ and zero otherwise. 

We use the Fast Iterative-Shrinkage Thresholding Algorithm (FISTA) of \citet{beck2009fast}, which solves unconstrained convex optimization of the form $\min_{x} h(x) + g(x)$ for a smooth convex function $h$ and a non-smooth convex function $g$. In each iteration, the algorithm executes Nesterov's accelerated gradient step with $\nabla h$, followed by a proximal operator with respect to $g$ acting on the updated coordinate. In our case, $h = \ell_n$ and $g = \delta_{\kappa_{N}^{-1} \mathcal{S}_{2N}}$, and the proximal operator becomes a projection.

\begin{algorithm}[H]
\caption{Accelerated Projected Gradient for MLE} \label{alg:projected-gradient}
 \KwData{$(s_i, \hat{\Delta}_i)_{i=1}^{n}$}
 \KwInput{order of trigonometric polynomial $N$, domain half-length $L$, step size $\gamma_t > 0$}
 $x_0 \leftarrow \text{median}(\{\hat{\Delta}_i\})$ \;
 $\hat{\Delta}_i \leftarrow (\hat{\Delta}_i - x_0) \pi / L$, $s_i \leftarrow s_i \pi / L$ for all $i$ \;
 $f^{(0)} \leftarrow f^{(-1)} \leftarrow (1,\dots,1)^{\T} / \kappa_N$ \;
 \For{$m=1,2,\dots$ until convergence}{
  $ y \leftarrow f^{(m-1)} + \frac{m-2}{m+1} (f^{(m-1)} - f^{(m-2)} )$ \;
  $ f^{(m)} \leftarrow (\kappa_N)^{-1} \proj_{\mathcal{S}_{2 N}} \left(\kappa_N \{y + \gamma_m \nabla \ell_n(y)\} \right)$ \;
 }
 \KwOutput{MLE $\hat{f} = f^{(m)}$}
\end{algorithm} 

The step-size should be either set as a constant that is no greater than the inverse of a Lipschitz constant of $\ell_n$, or determined by a line search; see \citet{beck2009fast} for more details. The resulting algorithm achieves $\ell_{n}^{\ast} - \ell_{n}(f^{(m)}) \leq O(1/m^2)$, where $\ell_n^{\ast}$ denotes the maximum of $\ell_n$. The iterates $\ell_n(f^{(m)})$ are not monotonic in general, and a simple restarting trick from \citet{o2015adaptive} can be employed to suppress oscillation. 
We implemented the algorithm with automatic differentiation from \emph{Autograd} \citep{maclaurin2015autograd} for computing $\nabla \ell_n(f)$, which is able to differentiate through the FFT. The log-likelihood is can be efficiently evaluated with matrix-vector multiplications. 

The operator $\proj_{\mathcal{S}_{2N}}(\cdot)$ in \Cref{alg:projected-gradient} is the Euclidean projection of a vector onto the unit simplex. We implement it with the following algorithm from \citet{duchi2008efficient}; see also \citet{wang2013projection} for an exposition. 

\begin{algorithm}[H]
\caption{Projection of a vector onto the unit simplex} \label{alg:projection-simplex}
\KwInput{$y \in \mathbb{R}^N$}
Order $y$ as $y^{(1)} \geq \cdots \geq y^{(N)}$ \;
$\rho \leftarrow \max \left\{1 \leq j \leq N: y^{(j)} + \frac{1}{j}(1 - \sum_{i=1}^{j} y^{(i)}) > 0 \right\}$ \;
$ \lambda \leftarrow \left(1 - \sum_{i=1}^{\rho} y^{(i)} \right) / \rho $ \;
\KwOutput{$\bm{x}$ with $x_i = \max\{y^{(i)} + \lambda, 0\}$ for $i=1,\dots,N$.}
\end{algorithm}

We remark that \citet{wager2014geometric} also used DFT for a related problem, which, in our terms, is to estimate the marginal density $p_{s=1}$ when $s_{i} \equiv 1$ in the dataset (no heterogeneity in $s$). His approach is geometric, namely to find the $L^2$ projection of the empirical measure $\mathbb{P}_n$ onto the space of possible marginal densities under $t=1$. By Parseval's identity, the $L^2$ distance becomes the $\ell_2$ norm in the space of Fourier coefficients. Therefore, the projection can formulated as a quadratic program that can be efficiently solved. Besides, \citet{wager2014geometric} also uses simplex constraints on node points, but its effect is only approximate since his parametrization does not guarantee $\hat{p}_1$ to be non-negative. In contrast to our approach, \citet{wager2014geometric} does not consider heterogeneity in the error scale, and his estimator is not an MLE.

\subsection{Model selection} \label{sec:model-sel}
In practice, the order $N$ of trigonometric polynomial $C_N$ should be chosen empirically. This is a classical bias-variance trade-off: $N$ being too small introduces bias from missing high-frequency components; $N$ being too big introduces spurious oscillations in the estimate. This phenomenon is illustrated by simulation studies in \cref{sec:numerical}. In the following, we discuss two criteria for model selection: (i) the predicted log-likelihood, which is typically used for evaluating probabilistic models, and (ii) a square loss that is targeted at the accuracy of the estimated posterior mean, which exploits \possessivecite{hyvarinen2005estimation} score matching technique.

\subsubsection{Predicted log-likelihood}
One model selection criterion common to probabilistic modeling is the predicted log-likelihood on the held-out data. 
\begin{lemma} \label{lem:identifiable}
Fix any $s \geq 0$, $g$ is uniquely identifiable (up to Lebesgue almost everywhere) from the marginal $p_{s}(\cdot) = \int \phi_{s}(\cdot - u) g(u) \dd u$. 
\end{lemma}
\begin{proof}
By \cref{eqs:solution-fourier-infinite}, the characteristic function of $G$ is determined by  $\mathcal{F}(g)(\xi) = e^{\frac{1}{2} \xi^2 s^2} \mathcal{F}(p_s)(\xi)$, where $\mathcal{F}(p_s)$ is the characteristic function of $p_s$. $g$ is uniquely determined (up to Lebesgue almost everywhere) by $\mathcal{F}^{-1}\mathcal{F}(g)$.
\end{proof}

Suppose $p_s$ and $\hat{p}_s$ are the marginal densities under $G$ and $\widehat{G}$ respectively. It follows that the Kullback-Leibler divergence
\begin{equation*}
0 \leq \KL(p_s \| \hat{p}_s) = \int p_s (x) \log \frac{p_{s}(x)}{\hat{p}_{s}(x)} \dd x = C_{G} - \E_{p_s} \log \hat{p}_{s}(\hat{\Delta}),
\end{equation*}
where the equality holds if and only if $G = \widehat{G}$, the constant $C_G$ does not depend on $\widehat{G}$. This motivates model selection with the averaged predictive log-likelihood
\begin{equation}
\E_{s \sim Q} \E_{p_s} \log \hat{p}_{s}(\hat{\Delta}^{\ast}) = \frac{1}{n^{\ast}} \sum_{i=1}^{n^{\ast}} \log \hat{p}_{s_i^{\ast}} (\hat{\Delta}^{\ast}_i) + O_{p}(1/\sqrt{n^{\ast}}), \label{eqs:held-out-ll}
\end{equation}
which is uniquely maximized when $\widehat{G} =_{d} G$. Samples $(\hat{\Delta}_i^{\ast}, s_i^{\ast})$ come from a held-out dataset of size $n^{\ast}$. $Q$ is either the population distribution or simply an empirical distribution for $s$. A model with a higher predictive log-likelihood is preferred. 

\subsubsection{Score-matching for the posterior mean} As we have seen in \cref{sec:bayes-optimal-decision}, the posterior mean is crucial to deciding whether one should launch a change. We consider a particular loss targeted at accurately approximating the true posterior mean based on the score matching technique. The following lemma is adapted from \citet{hyvarinen2005estimation}.

\begin{lemma} \label{lem:hyvarinen}
Suppose $p$, $\hat{p}$ are periodic functions with period $2\pi$ and are densities on $[-\pi, \pi]$. Suppose $p$ is continuously differentiable and $\hat{p}$ is twice continuously differentiable. Let $\ell = \log p$ and $\hat{\ell} = \log \hat{p}$. It holds that 
\begin{equation}
\E_{p} (\ell' - \hat{\ell}')^2 = \E_{p} \left\{ (\hat{\ell}')^2 + 2 \hat{\ell}'' \right\} + C_p,
\end{equation}
when all the terms on the RHS are finite, where the constant $C_p = \int_{-\pi}^{\pi} p(x) (\ell'(x))^2 \dd x$ only depends on $p$.
\end{lemma}
\begin{proof} We have
\begin{equation*}
\begin{split}
\E_{p} (\ell' - \hat{\ell}')^2 &= \int_{-\pi}^{\pi} \left ( \frac{p'(x)}{p(x)} - \frac{\hat{p}'(x)}{\hat{p}(x)} \right )^2 p(x) \dd x \\
&= \int_{-\pi}^{\pi} p(x) (\ell'(x))^2 \dd x + \int_{-\pi}^{\pi} p(x) (\hat{\ell}'(x))^2 \dd x - 2 \int_{-\pi}^{\pi} \frac{p'(x) \hat{p}'(x)}{\hat{p}(x)} \dd x,
\end{split}
\end{equation*}
where via integration by parts we obtain
\begin{equation*}
\begin{split}
\int_{-\pi}^{\pi} \frac{p'(x) \hat{p}'(x)}{\hat{p}(x)} \dd x &= \int_{-\pi}^{\pi} \frac{\hat{p}'(x)}{\hat{p}(x)} \dd p(x) = p(x) \hat{\ell}'(x)|^{+\pi}_{-\pi} - \int_{-\pi}^{\pi} p(x) \hat{\ell}''(x) \dd x \\
&= -\int_{-\pi}^{\pi} p(x) \hat{\ell}''(x) \dd x,
\end{split}
\end{equation*}
where the boundary difference term vanishes due to the periodicity of $2 \pi$ . It follows that 
\begin{equation*}
\E_{p} (\ell' - \hat{\ell}')^2 = \E_{p} \left \{ (\hat{\ell}')^2 + 2 \hat{\ell}'' \right \} + \E_{p} (\ell'_{p})^2.
\end{equation*}
\end{proof}

\begin{lemma}
Fix any $s>0$. Suppose $p_0$ is the true prior and $\E_s[\Delta \mid \hat{\Delta}]$ is the posterior mean under $p_0$. Suppose $\hat{p}_0$ is the estimated prior, and $\hat{\E}_{s}[\Delta \mid \hat{\Delta}]$ is the posterior mean under $\hat{p}_0$. Under \cref{assump:iid,assump:prior,assump:torus}, it holds that 
\begin{equation*}
\E_{\hat{\Delta} \sim p_s} \left (\E_{s}[\Delta \mid \hat{\Delta}] - \widehat{\E}_{s}[\Delta \mid \hat{\Delta}] \right)^2 = s^4 \E_{s} \left \{\left(\hat{\ell}'_s(\hat{\Delta}) \right)^2 + 2 \hat{\ell}''_{s}(\hat{\Delta}) \right\} + C_{p_0,s},
\end{equation*}
when every term on the RHS is finite. $C_{p_0,s}$ does not depend on $\hat{p}_0$.
\end{lemma}
\begin{proof}
Suppose $\ell_{s}$ and $\hat{\ell}_{s}$ are the logarithmic marginal densities under $p_0$ and $\hat{p}_0$ respectively. By Tweedie's formula \cref{eqs:tweedie-mean}, the square loss in approximating the posterior mean is
\begin{equation}
\left (\E_{s}[\Delta \mid \hat{\Delta}] - \widehat{\E}_{s}[\Delta \mid \hat{\Delta}] \right)^2 = s^4 \left(\ell'_{s}(\hat{\Delta}) - \hat{\ell}'_{s} (\hat{\Delta}) \right)^2. \label{eqs:square-loss-mean}
\end{equation}
For $s>0$, $p_s(x)$ and $\hat{p}_s(x)$ are infinitely many times differentiable by the property of Gaussian convolution. Further, under \cref{assump:prior,assump:torus}, $p_s$ and $\hat{p}_s$ have period $2\pi$ and are densities on the torus $[-\pi, \pi]$. Hence, we can apply \cref{lem:hyvarinen} and obtain 
\begin{equation*}
\begin{split}
\E_{\hat{\Delta} \sim p_s} \left (\E_{s}[\Delta \mid \hat{\Delta}] - \widehat{\E}_{s}[\Delta \mid \hat{\Delta}] \right)^2 &= s^4 \E_{p_s} \left(\ell'_{s}(\hat{\Delta}) - \hat{\ell}'_{s} (\hat{\Delta}) \right)^2 \\
&= s^4 \E_{s} \left \{\left(\hat{\ell}'_s(\hat{\Delta}) \right)^2 + 2 \hat{\ell}''_{s}(\hat{\Delta}) \right\} + C_{p_0,s},
\end{split}
\end{equation*}
where $C_{p_0,s} = s^4 \E_{p_s} (\ell'_s(\Delta))^2$ does not depend on $\hat{p}_0$.
\end{proof}

Since the $\E_{s}[\Delta \mid \hat{\Delta}]$ has the same unit as $s$, we scale the square loss at $s$ by $s^{-2}$ to make it dimensionless. Dropping the irrelevant constant, we use the following averaged loss for model selection
\begin{equation}
\begin{split}
& \quad \E_{s \sim Q} \E_{\hat{\Delta} \sim p_s} s^{-2} \left (\E_{s}[\Delta \mid \hat{\Delta}] - \hat{\E}_{s}[\Delta \mid \hat{\Delta}] \right)^2 - C_{Q,p_0}\\
&= \E_{s \sim Q} s^2 \E_{p_s} \left \{\left(\hat{\ell}'_s(\hat{\Delta}) \right)^2 + 2 \hat{\ell}''_{s}(\hat{\Delta}) \right\} \\
&= \frac{1}{n^{\ast}} \sum_{i=1}^{n^\ast} s_i^{\ast 2}   \left \{\left(\hat{\ell}'_{s_i^{\ast}}(\hat{\Delta}_i^{\ast}) \right)^2 + 2 \hat{\ell}''_{s_i^{\ast}} (\hat{\Delta}_i^{\ast}) \right \} + O_p(1/\sqrt{n^{\ast}}).
\end{split}
\end{equation}

\subsection{Consistency} \label{sec:consistency}
Let $\hat{p}_{0} := \widehat{C}_N(x; t=0)$ be the MLE for the prior density from solving \cref{eqs:opt}. To achieve consistency, we have to enlarge the class of densities optimized over by letting $N \rightarrow \infty$ as $n \rightarrow \infty$. This type of estimation is called Grenander's method of sieves \citep{grenander1981abstract}. For maximum likelihood, when the class of functions (``sieves'') is dense relative to the space of true density functions, the resulting estimator is consistent under weak conditions. One condition that we require is that $t$ is upper-bounded in data. 

\begin{assumption} \label{assump:t-bounded}
$t \leq t_{\max}$ almost surely for some $t_{\max} < \infty$. 
\end{assumption}

Let $\|\hat{p} - p\|_{\infty}:= \sup_{x} \|\hat{p}(x) - p(x)\|$. 

\begin{theorem} \label{thm:consistency}
Under \cref{assump:iid,assump:prior,assump:torus,assump:t-bounded}, for $N \rightarrow \infty$ as $n \rightarrow \infty$, it holds that $\|\hat{p}_0 - p_0 \|_{\infty} \rightarrow_{p} 0$. 
\end{theorem}
The consistency can be shown by verifying the generic conditions from \citet[Page 5590]{chen2007large}. We delegate the proof to \cref{apx:sec:consistency}. It is worth mentioning that the sieve MLE is also in general asymptotically efficient in the Fisher sense; see \citet{shen1997methods}. 

Despite the fact that uniform convergence of $\hat{p}_0$ need not imply the convergence of $\hat{p}'_0$ as $n \rightarrow \infty$ (see \citet[
Example 7.5]{rudin1964principles}), by the property of Gaussian convolution, we in fact have uniform convergence of the estimated posterior mean functions. 

\begin{corollary} \label{cor:unif-conv}
Given any $t>0$, $\|\hat{p}_t - p_t\|_{\infty} \rightarrow_{p} 0$ and $\|\hat{p}_t' - p_t'\|_{\infty} \rightarrow_{p} 0$. 
\end{corollary}
\begin{proof} With $s = \sqrt{t}$, we have
\begin{equation*}
\begin{split}
\|\hat{p}_t - p_t\|_{\infty} &= \sup_{x} \left |\int(\hat{p}_0(u) - p_0(u)) \phi_{s}(x-u) \dd u \right | \\
& \leq \|\hat{p}_0 - p_0\|_{\infty} \sup_{x}\left| \int \phi_{t}(x-u) \dd u \right| = \|\hat{p}_0 - p_0\|_{\infty} \rightarrow_{p} 0,
\end{split}
\end{equation*}
and 
\begin{equation} \label{eqs:unif-conv}
\begin{split}
\|\hat{p}_t' - p_t'\|_{\infty} &= \sup_{x} \left|\int (\hat{p}_0(u) - p_0(u)) \phi'_s(x-u) \dd u \right| \\
& \leq \|\hat{p}_0 - p_0\|_{\infty} \sup_{x} s^{-2} \int \phi_{s}(x-u) |x-u| \dd u \\
& = \|\hat{p}_0 - p_0\|_{\infty} s^{-2} \int \phi_s(z) |z| \dd z \rightarrow_{p} 0.
\end{split}
\end{equation}
\end{proof}
Recall that $\E_{t} (\Delta \mid \hat{\Delta}=\cdot)$ and $\widehat{\E}_{t} (\Delta \mid \hat{\Delta}=\cdot)$ are posterior mean functions $[-L,L] \rightarrow [-L,L]$ under $p_0$ and $\hat{p}_0$ respectively. 
\begin{lemma} \label{lem:unif-consistency-posterior-mean}
Given any $t > 0$, under the same conditions as \cref{thm:consistency}, $\|\widehat{\E}_{t} (\Delta \mid \hat{\Delta}=\cdot) - \E_{t} (\Delta \mid \hat{\Delta} = \cdot)\|_{\infty} \rightarrow_{p} 0$.
\end{lemma}
\begin{proof}
By \cref{lem:tweedie}, it suffices to show 
\begin{equation*}
\sup_{x \in [-L,L]} \left |\frac{\hat{p}_t'(x)}{\hat{p}_t(x)} - \frac{p_t'(x)}{p_t(x)} \right | \rightarrow_{p} 0.
\end{equation*}
For any $x \in [-L,L]$, by Taylor expansion on $f(a,b) = a/b$ we have
\begin{equation*}
\begin{split}
\left |\frac{\hat{p}_t'(x)}{\hat{p}_t(x)} - \frac{p_t'(x)}{p_t(x)}\right| &= \left|\frac{1}{p_t(x)} (\hat{p}'_t(x) - p'_t(x)) - \frac{p'_t(x)}{p_t^2(x)} (\hat{p}_t(x) - p_t(x)) \right| \\
& \quad \quad + o\left(\|\hat{p}'_t - \hat{p}'_t\|_{\infty} + \|\hat{p}_t - p_t\|_{\infty}\right).
\end{split}
\end{equation*}
Note that on $[-L,L]$, $p_t$ is lower-bounded by a positive constant (see \cref{eqs:marginal-dens-lb} in \cref{apx:sec:consistency}), and $|p_t'|$ is upper bounded by a positive constant (a computation similar to \cref{eqs:unif-conv}). The previous display is thus upper bounded by
\begin{equation*}
\left \|\frac{\hat{p}_t'(x)}{\hat{p}_t(x)} - \frac{p_t'(x)}{p_t(x)}\right\|_{\infty} \leq c_1 \|\hat{p}_t' - p_t'\|_{\infty} + c_2 \|\hat{p}_t - p_t\|_{\infty} + o\left(\|\hat{p}'_t - \hat{p}'_t\|_{\infty} + \|\hat{p}_t - p_t\|_{\infty}\right)
\end{equation*}
for $0 < c_1, c_2 < \infty$. The result then follows from \cref{cor:unif-conv}.
\end{proof} \section{Numerical results} \label{sec:numerical}
We present some results on simulated examples and real large-scale experiments. 

\subsection{Simulations} In the following we consider several simulation studies where $p_0$ is chosen to be a known prior. 

\subsubsection{Uniform distribution} We set $p_0 = \unif(-4,4)$. The domain half-length is chosen to be $L = 8$. We simulate $n=2,000$ data points $(\hat{\Delta}_i, s_i)$ with $s_i \sim \unif(0,1)$. The reader is referred to \citet{wager2014geometric} for a similar example where $s_i$ is fixed to 1. To select $N$, we use Monte Carlo cross-validation that randomly splits between the training set ($90\%$ of data) and the test set ($10\%$ of data). The random split is repeated 100 times. We select $N$ from options $\{4, 6, 12, 16, 24, 32, 48, 64\}$. See \cref{tab:model-sel} for the results based on the two model-selection criteria proposed in \cref{sec:model-sel}; $N=32$ is selected by both criteria (the highest predicted log-likelihood and the lowest score-matching loss). 

\begin{table}[!htb]
\caption{Model selection results (standard errors are shown in brackets)}
\label{tab:model-sel}
\scalebox{0.75}{
\begin{tabular}{cccccccccccc}
\toprule
Dataset & Method & $N$ & 4 & 6 & 12 & 16 & 24 & 32 & 48 & 64 \\
\midrule
\multirow{4}{*}{Uniform} & \multirow{4}{*}{\textsf{Spectral}} & \multirow{2}{*}{log-likelihood} & -2.264 & -2.240 & -2.213 & -2.205 & -2.197 & \textbf{-2.193} & -2.197 & -2.199 \\
 &  & & (0.002) & (0.002) & (0.002) & (0.003) & (0.003) & (0.003) & (0.003) & (0.003) \\
& & \multirow{2}{*}{score-matching} & -0.082 & -0.090 & -0.101 & -0.105 & -0.107 & \textbf{-0.114} & -0.108 & -0.101 \\
& & & (0.001) & (0.002) & (0.002) & (0.002) & (0.003) & (0.002) & (0.003) & (0.004) \\
\midrule
 & \multirow{4}{*}{\textsf{Spectral}} & \multirow{2}{*}{log-likelihood} & -2.149 & -2.075 & -1.992 & -1.980 & -1.973 & -1.962 & \textbf{-1.955} & -1.967 \\
Mixture of &  & & (0.003) & (0.003) & (0.003) & (0.003) & (0.003) & (0.004) & (0.004) & (0.004) \\
2 Gaussians & & \multirow{2}{*}{score-matching} & -0.139 & -0.205 & -0.261 & -0.261 & -0.273 & -0.286 & \textbf{-0.294} & -0.284 \\
& & & (0.004) & (0.002) & (0.004) & (0.004) & (0.004) & (0.004) & (0.004) & (0.005) \\
\midrule
& & $N$ & 512 & 1024 & 1536 & 2048 & 3072 & 4096 \\
\midrule
\multirow{4}{*}{Amazon} & \multirow{4}{*}{\textsf{Spectral}} & \multirow{2}{*}{log-likelihood} & 0.070 & 0.072 & 0.072 & \textbf{0.074} & 0.072 & 0.072 \\
 &  & & (0.001) & (0.001) & (0.001) & (0.001) & (0.001) & (0.001) \\
& & \multirow{2}{*}{score-matching} & -0.73 & -0.75 & -0.75 & \textbf{-0.77} & -0.71 & -0.74 \\
& & & (0.01) & (0.01) & (0.01) & (0.01) & (0.01) & (0.01) \\
\midrule
& & $K$ & 1 & 2 & 3 & 4 & 5 \\
\midrule
\multirow{4}{*}{Amazon} &  & \multirow{2}{*}{log-likelihood} & -0.148 & 0.068 & \textbf{0.072} & 0.069 & 0.069 \\
 & \textsf{Gaussian} & & (0.001) & (0.001) & (0.001) & (0.001) & (0.001) \\
& \textsf{mixture} & \multirow{2}{*}{score-matching} & -0.12 & -0.75 & \textbf{-0.78} & \textbf{-0.78} & \textbf{-0.78}\\
& & & (0.002) & (0.01) & (0.01) & (0.01) & (0.01) \\
\bottomrule
\end{tabular}}
\end{table}

\begin{figure}[!hb]
\centering
\includegraphics[width=0.3\textwidth]{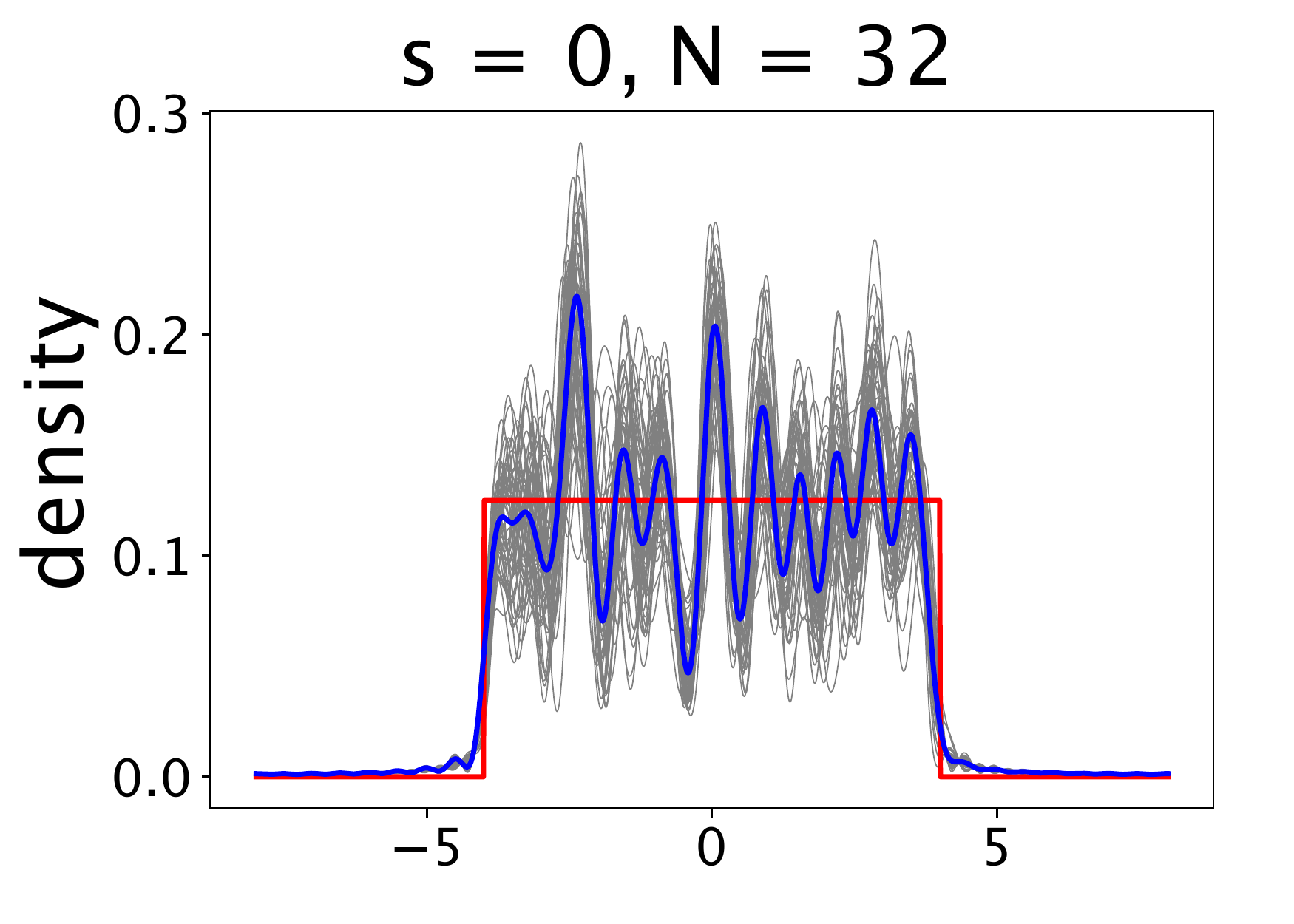}
\includegraphics[width=0.3\textwidth]{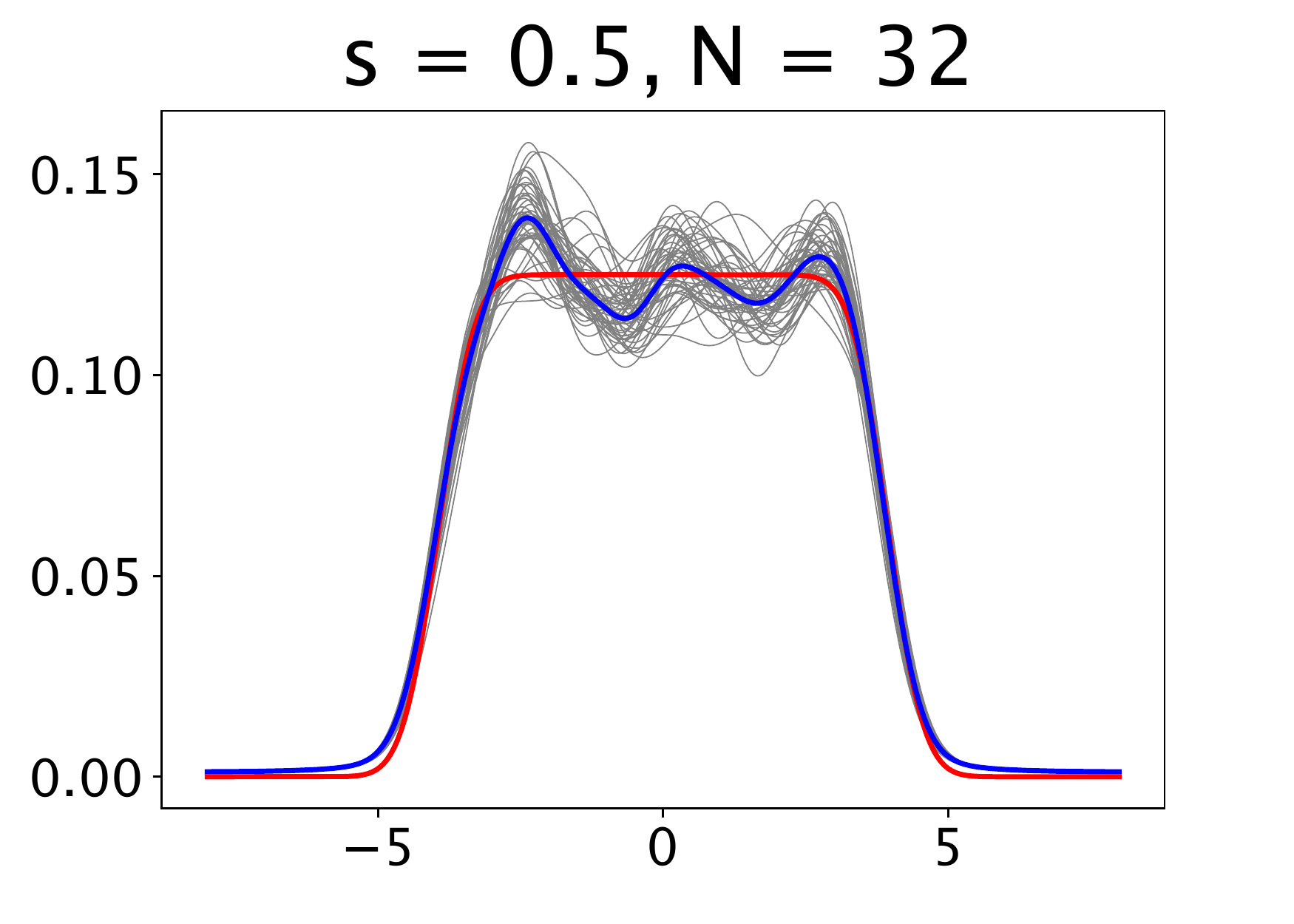}
\includegraphics[width=0.3\textwidth]{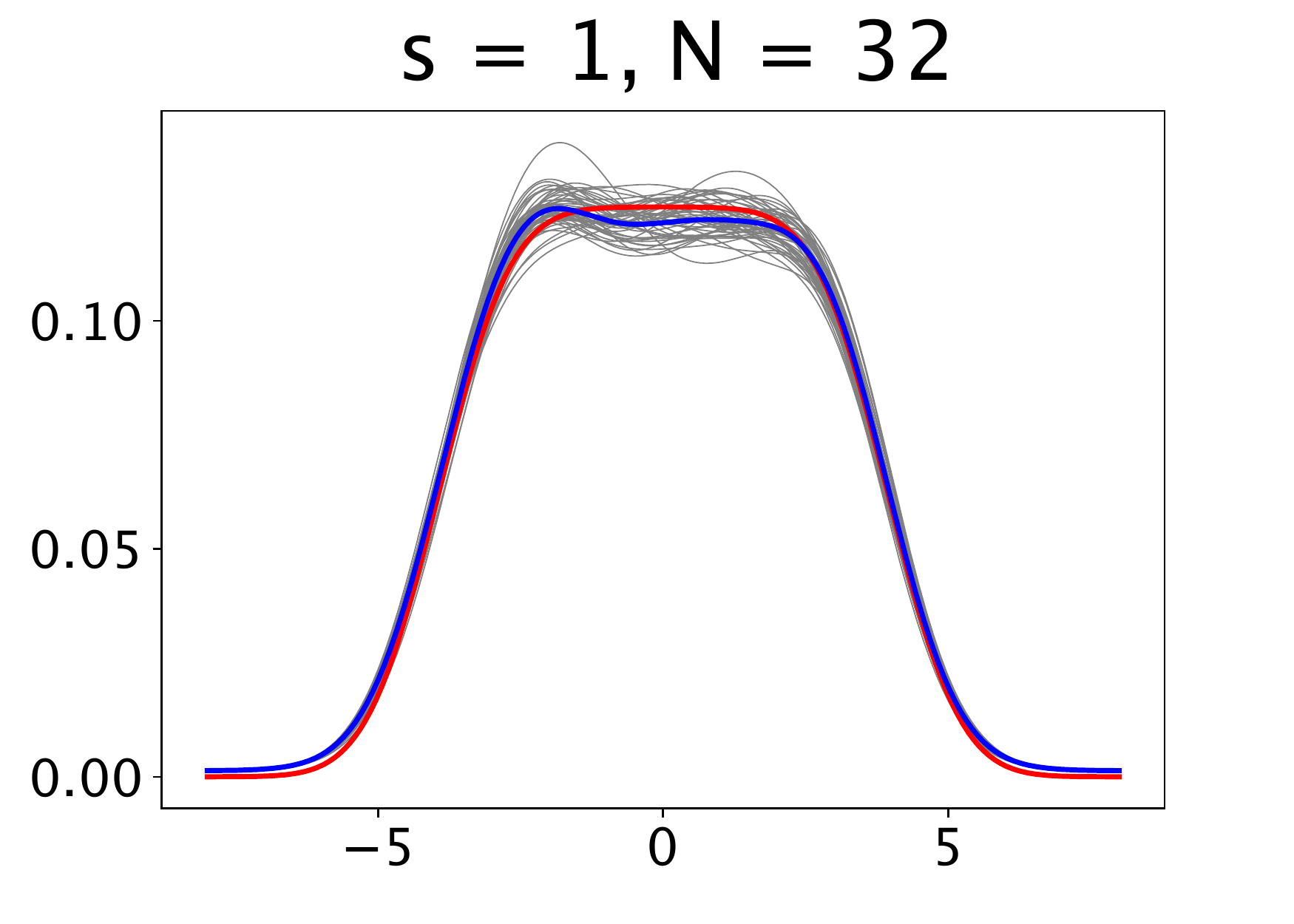}
\includegraphics[width=0.29\textwidth]{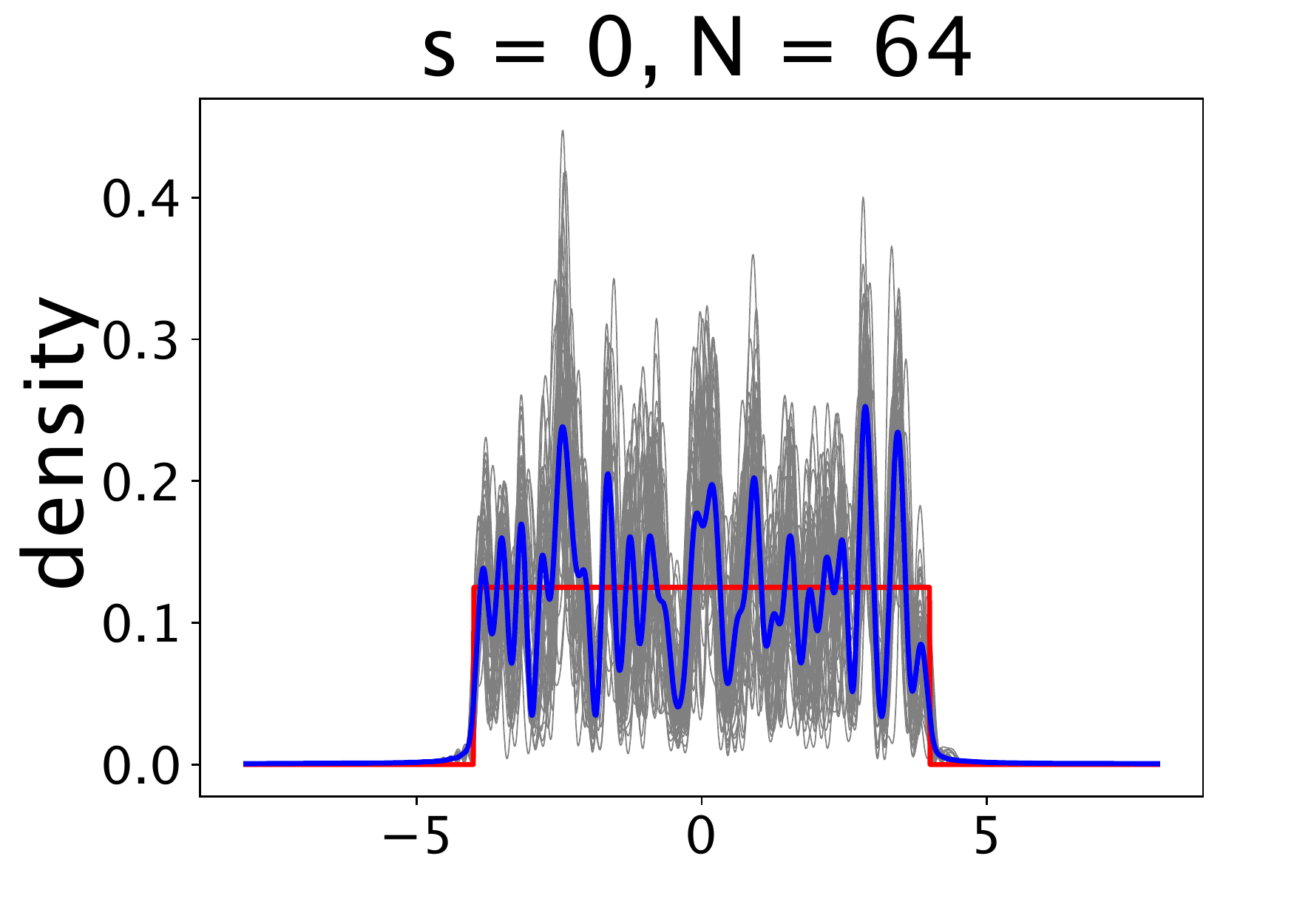}
\includegraphics[width=0.3\textwidth]{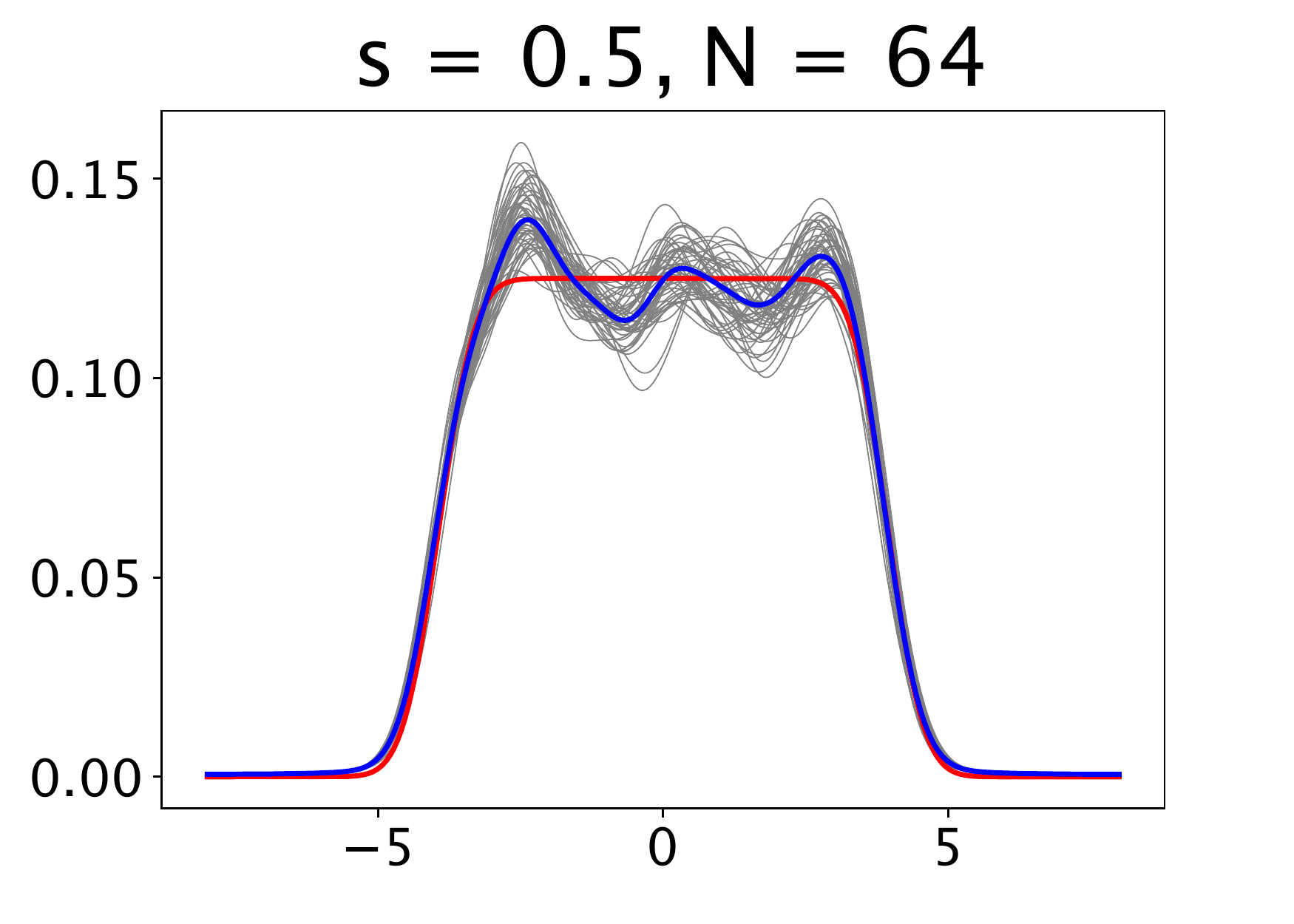}
\includegraphics[width=0.3\textwidth]{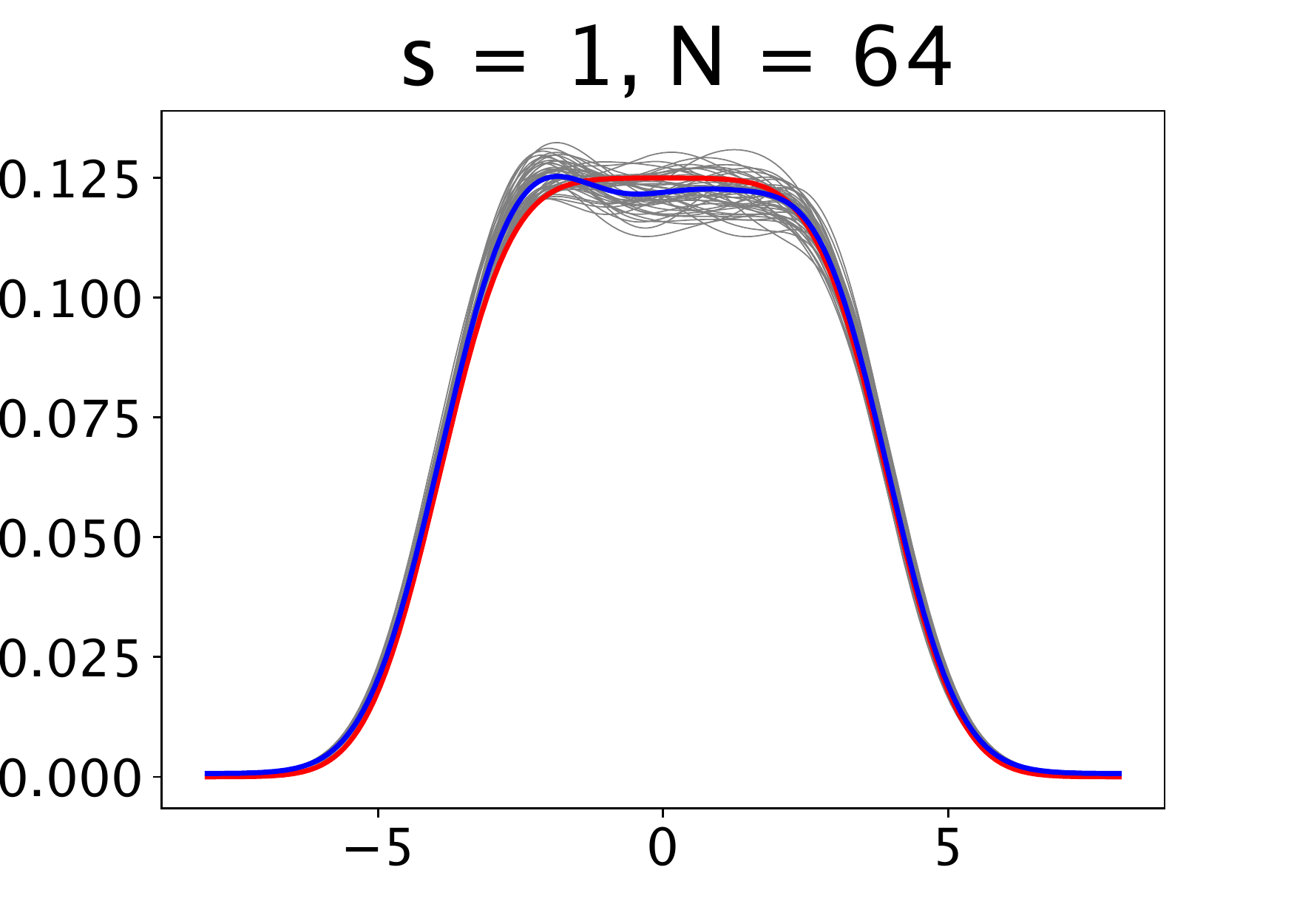}
\caption{Comparison of estimates under $p_0 = \unif(-4,4)$. Estimated (blue) versus true densities (red) $p_s(\cdot)$ for $s=0,0.5,1$ are shown (top: using $N=32$ as selected from cross validation, bottom: using a larger $N=64$). The grey curves are 50 bootstrap estimates. The error in high-frequency components diminishes quickly as $s$ grows.}
\label{fig:simu-unif}
\end{figure}

\Cref{fig:simu-unif} compares the estimates from $N=32$ (chosen by cross-validation) and from a larger $N=64$. Choosing $N=64$ introduces more high-frequency oscillations in the estimate for $p_0$ (under-smoothed). The difference between the two estimates diminishes as we compare $\hat{p}_s$ for a larger $s$. The high-frequency errors are damped very quickly; see \cref{lem:heat-series}. 

\subsubsection{Mixture of Two Gaussians}
We consider a mixture of two Gaussians
\begin{equation*}
p_0 = 0.3 \N(-1.5, 0.2^2) + 0.7 \N(2, 1.0^2)
\end{equation*}
and set $L = 8$. We simulate 1,000 samples with $s_i \sim \unif(0, 1.5)$. See \cref{fig:mog-data} for the marginal densities and the scatterplot of a simulated dataset. Again we cross-validate on the order of trigonometric polynomials $N$. \cref{tab:model-sel} shows the predicted log-likelihood and the score-matching loss computed from held-out data ($10\%$ of samples), based on 200 random splits of the same dataset. $N=48$ is selected in terms of both criteria. 
In \cref{fig:simu-mog}, we compare the estimates from $N=6$ (too small, over-smoothed) and from $N=48$ (selected by cross-validation) from 50 realizations of the sampling distribution. As a reference, we also include estimates from fitting the true model.

\begin{figure}[!hb]
\centering
\includegraphics[width=0.45\textwidth]{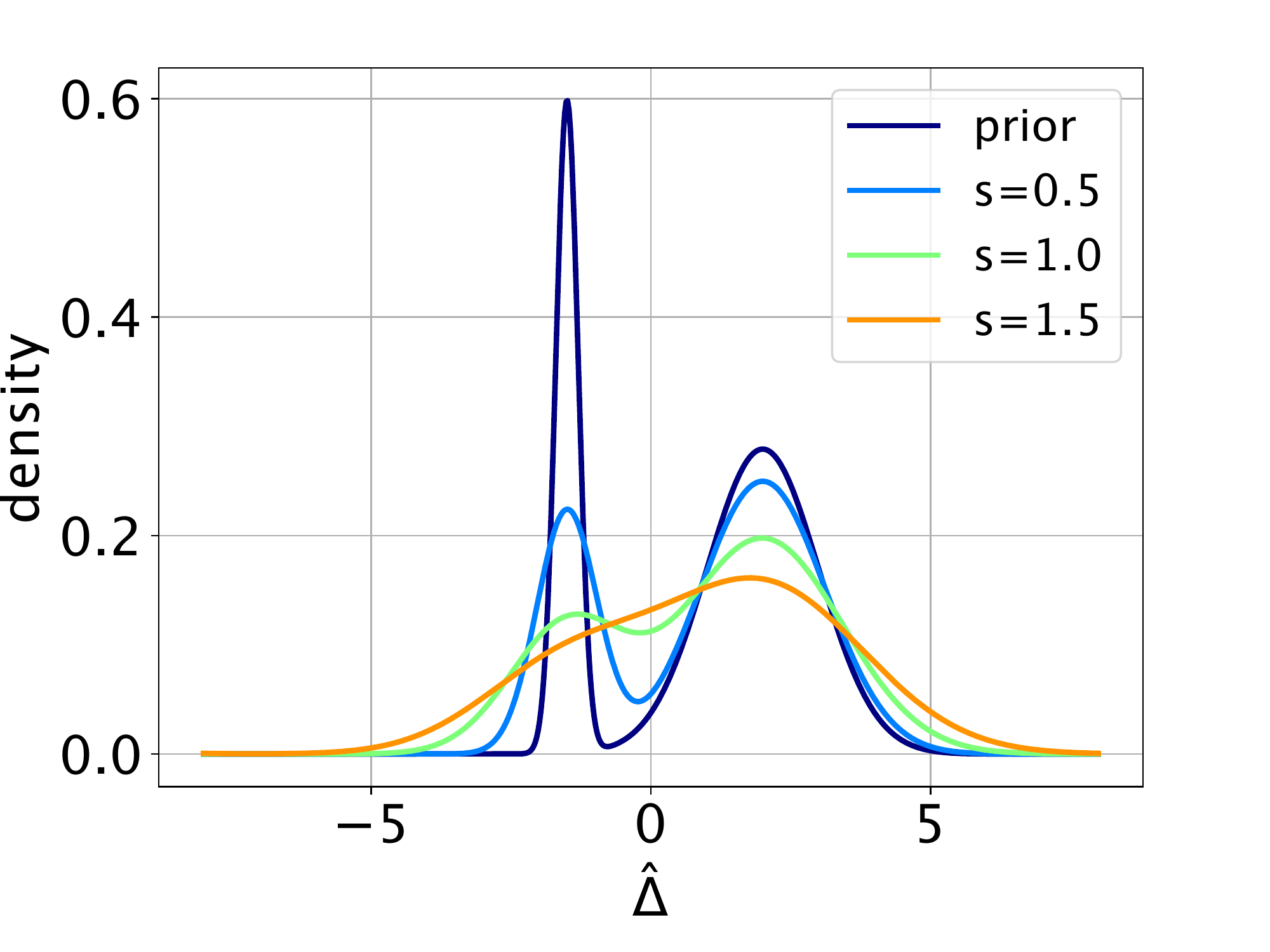}
\includegraphics[width=0.45\textwidth]{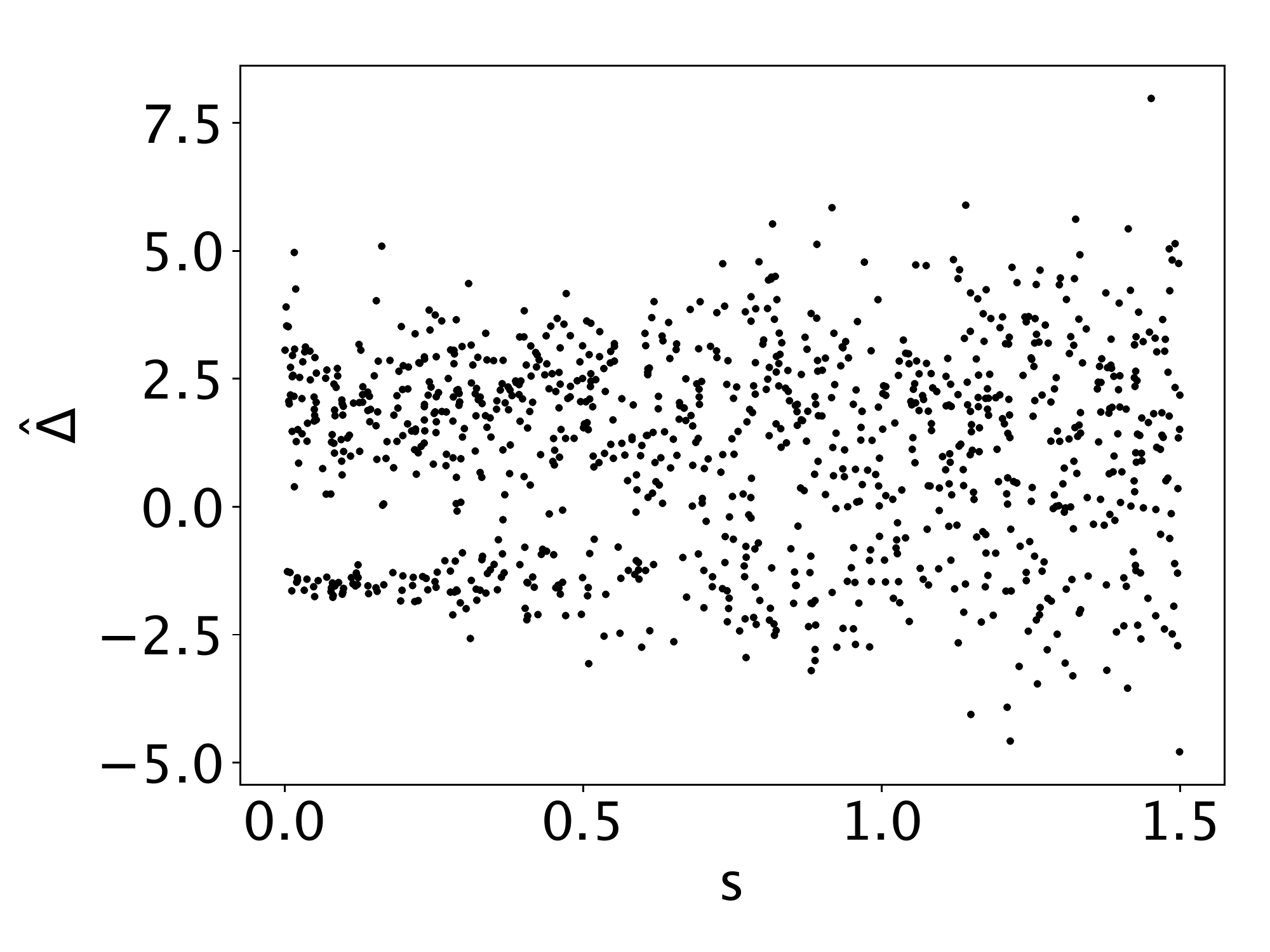}
\caption{The marginal densities $p_s(\hat{\Delta})$ (left) and the simulated data points $(s_i, \hat{\Delta}_i)$, where $p_0$ is a mixture of two Gaussians.}
\label{fig:mog-data}
\end{figure}

\begin{figure}[!htb]
\centering
\includegraphics[width=0.3\textwidth]{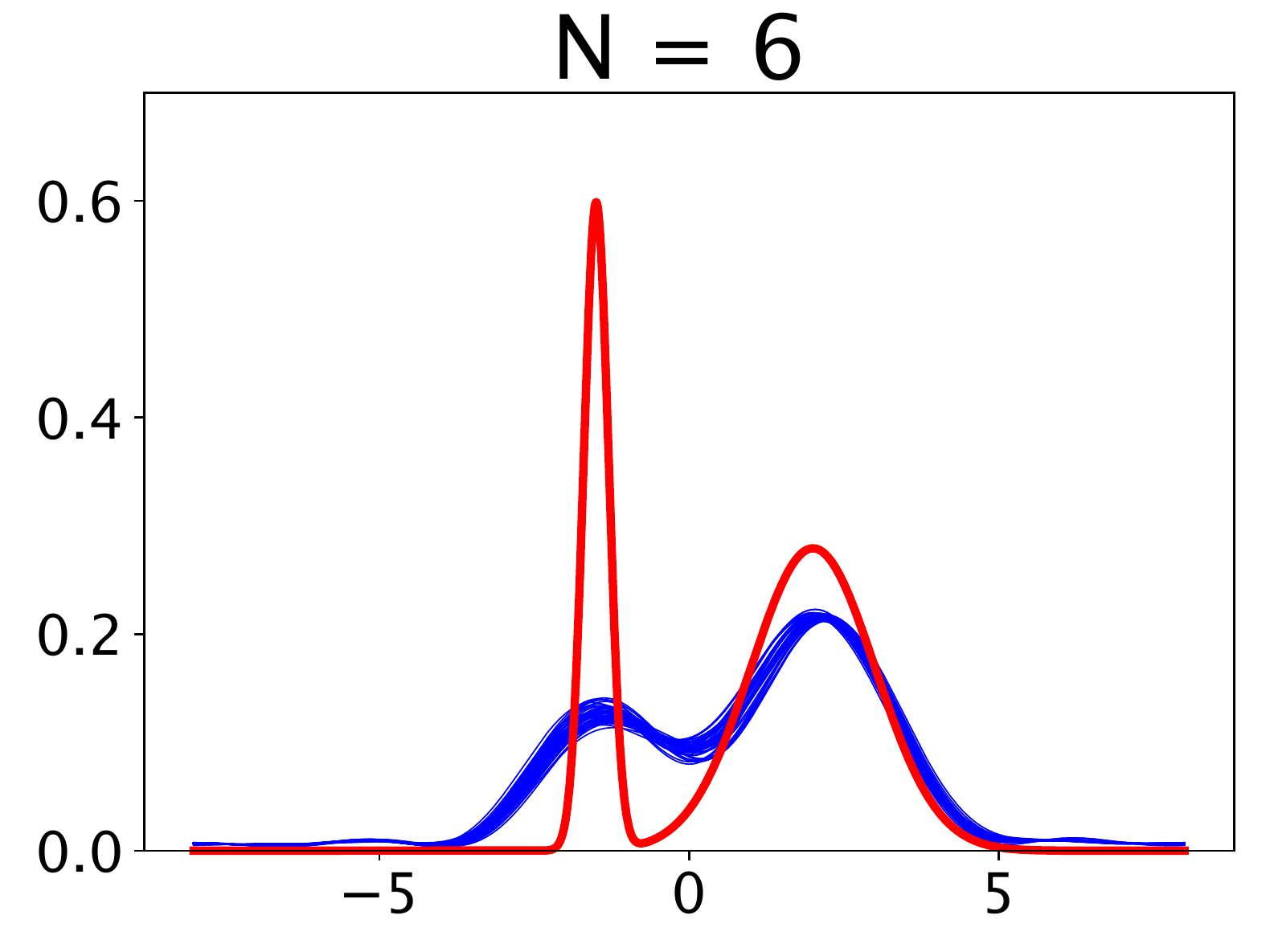}
\includegraphics[width=0.3\textwidth]{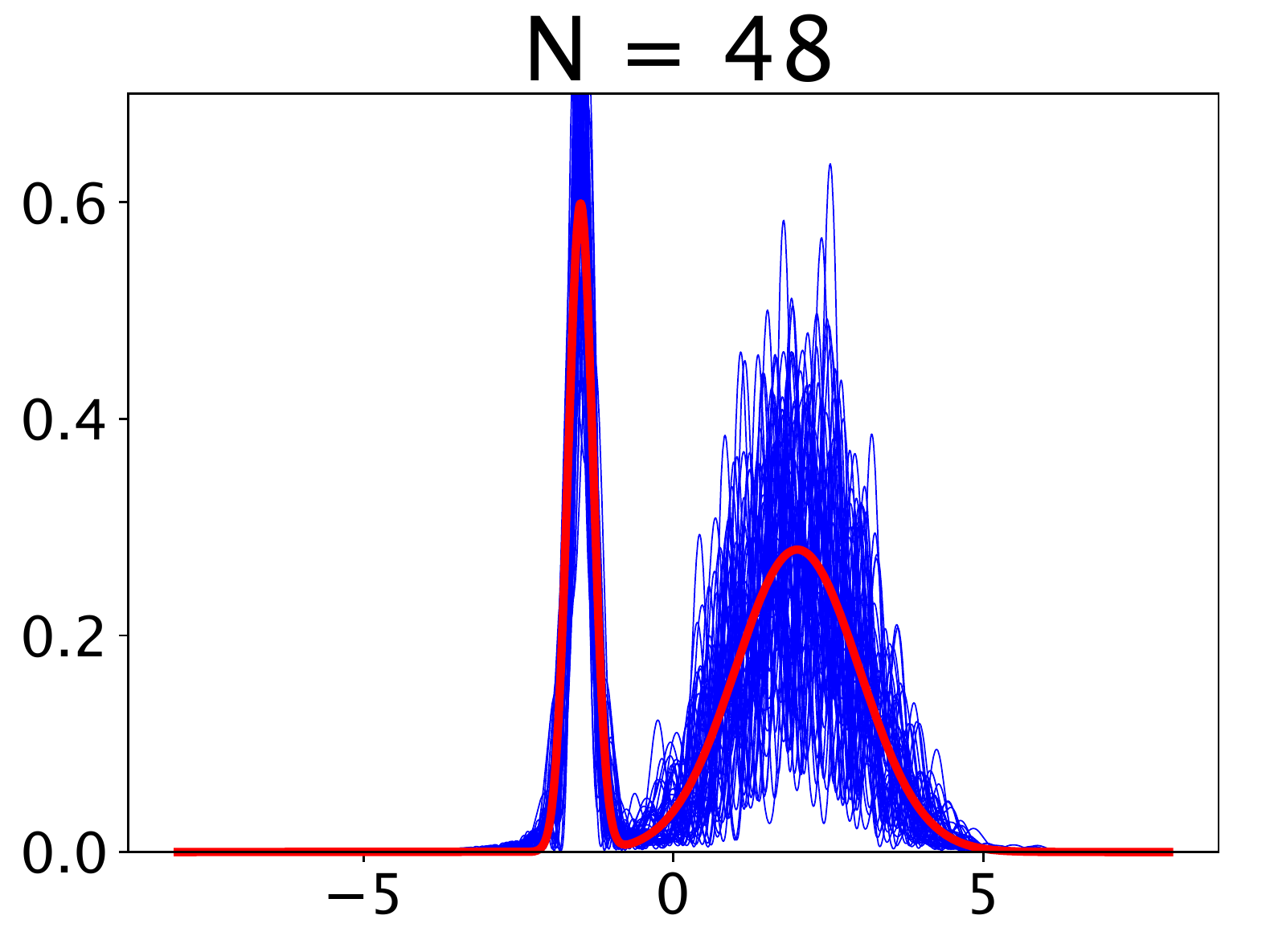}
\includegraphics[width=0.3\textwidth]{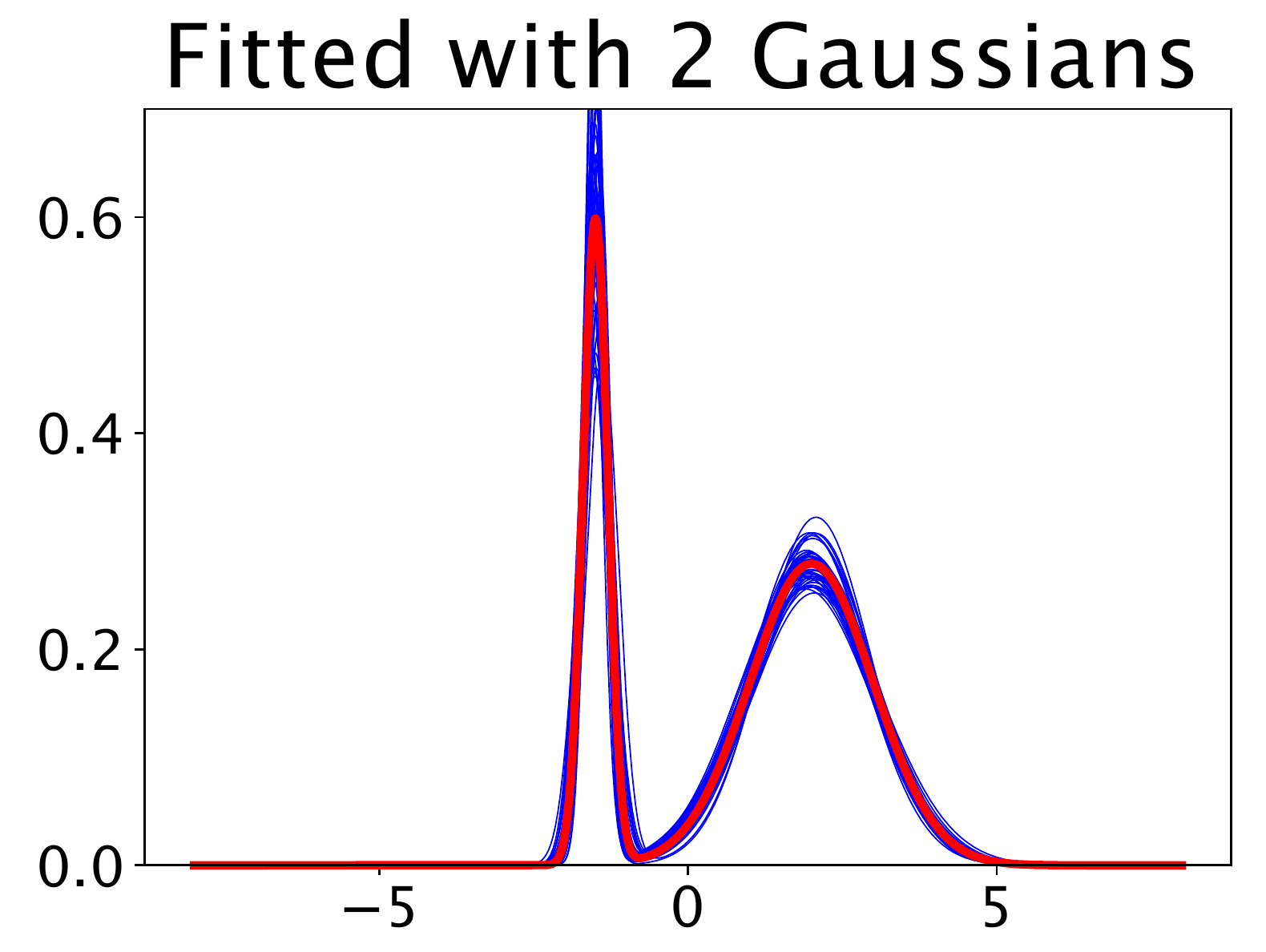}
\caption{Estimates compared to the true density when $p_0$ is a mixture of two Gaussians. The first two plots come from our method with $N=6$ (over-smoothed) and $N=48$ (selected by cross-validation). The last plot corresponds to fitting the true model. To illustrate the performance of the estimators over hypothetical replications, the blue curves are estimated from 50 simulated datasets, each consisting of 1,000 data points.}
\label{fig:simu-mog}
\end{figure}

\subsection{Large-scale experimentation at Amazon} We apply our method to large-scale A/B tests run at Amazon. \cref{fig:xs-amzn} shows a subset of 680 past experiments coming from the same population, where $\hat{\Delta}_i$ is the empirical estimate of some effect measured in some standardized duration and unit. We choose an appropriately large $L$, which leaves 8 data points off the domain; those data points are approximately projected to the boundary via $(\hat{\Delta}, s) \leftarrow (\text{sign}(\hat{\Delta}) L, L s / |\hat{\Delta}|)$ (the transform is exact if the associated true effect is zero). 

\begin{figure}[!htb]
\centering
\includegraphics[width=0.8\textwidth]{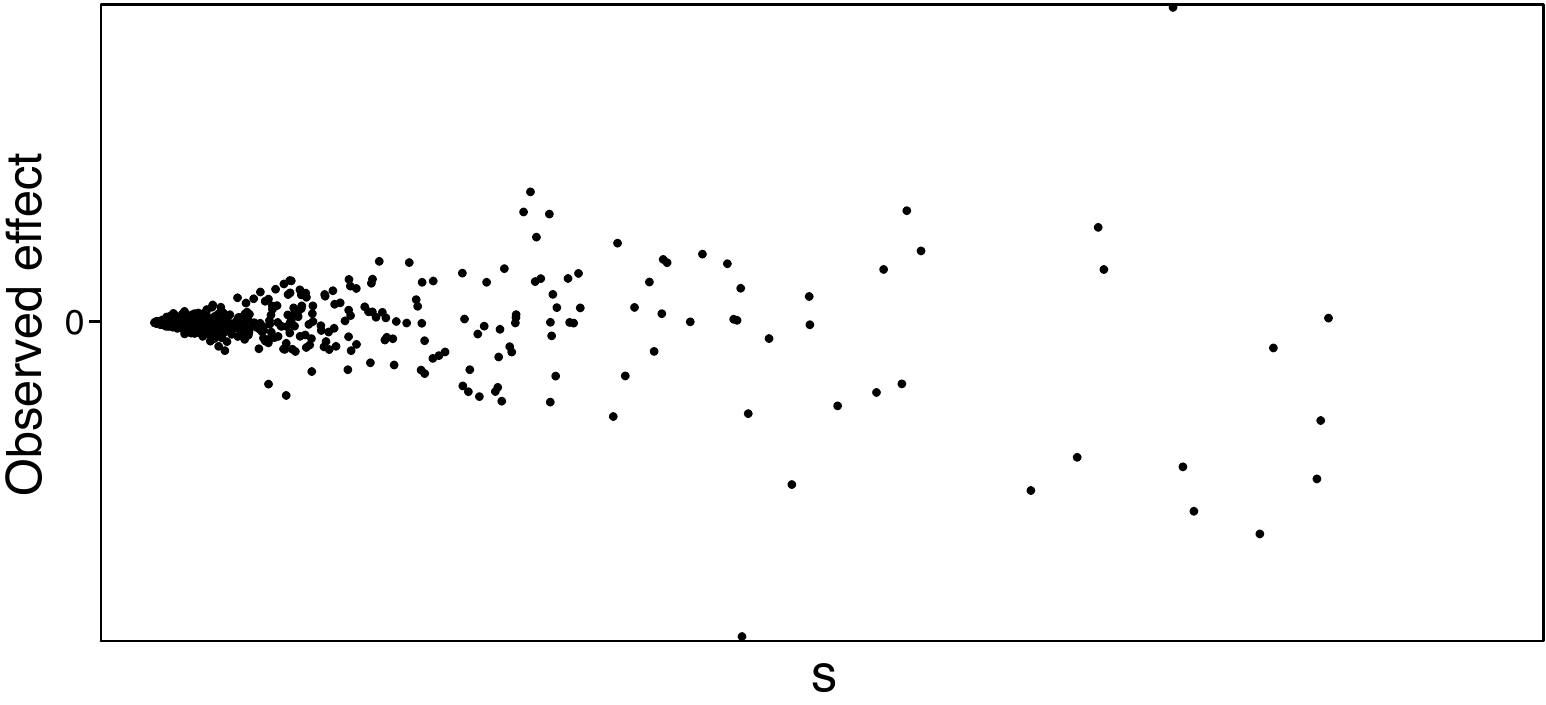}
\caption{680 experiments $(\hat{\Delta}_i, s_i)$ from the same population run at Amazon.}
\label{fig:xs-amzn}
\end{figure}

We select $N$ from $\{512, 1024, 1536, 2048, 3072, 4096\}$. We ran cross-validation 400 times randomly holding out $10\%$ of data. \Cref{tab:model-sel} shows the two model selection criteria. Both the predicted log-likelihood and the score-matching loss prefer $N=2,048$ from the list of options. We fit the full dataset with this selected $N$. \Cref{fig:prior-amzn} displays the estimated prior density. The pointwise $95\%$ confidence bands are estimated from 500 bootstrap replicates.

\begin{figure}[!htb]
\centering
\includegraphics[width=0.8\textwidth]{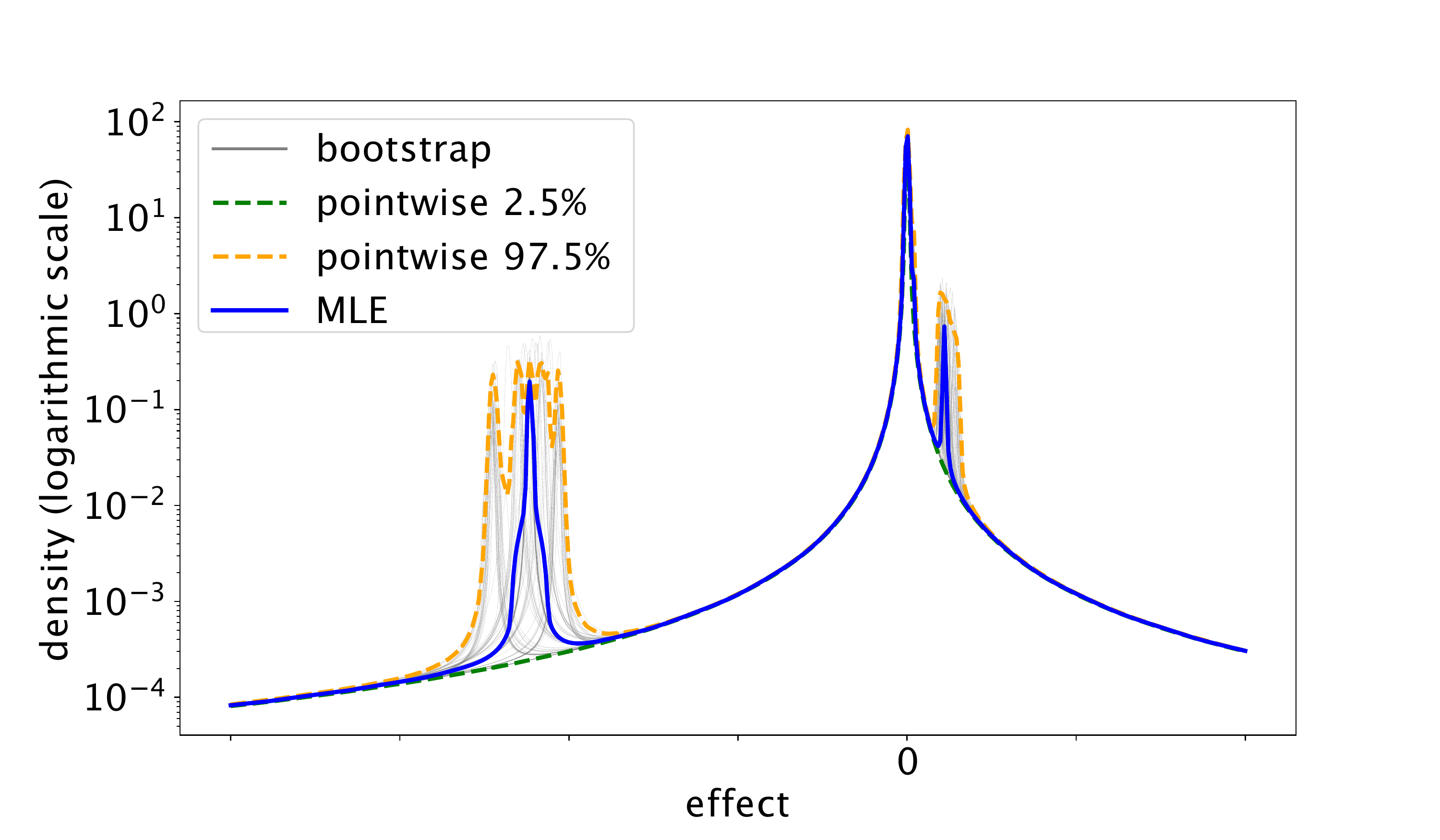}
\caption{The prior over the true effect estimated from a set of experiments run at Amazon. The pointwise confidence bands are estimated from bootstrap. The density is drawn in the logarithmic scale; see \cref{apx:fig:prior-amzn} in the Appendix for the linear scale. Also compare to \cref{apx:fig:prior-amzn-mog}, which is fitted with a mixture of Gaussians.}
\label{fig:prior-amzn}
\end{figure}

\paragraph{Comparison} We compare to fitting the prior with a mixture of $K$ Gaussian distributions; see \cref{apx:sec:EM} for the fitting algorithm. We select $K$ using the same cross-validated criteria. As shown in \cref{tab:model-sel}, $K=3$ is selected. The prior density fitted is shown in \cref{apx:fig:prior-amzn-mog} in the Appendix. The difference between our method and the mixture of Gaussians is apparent when plotting on the logarithmic scale --- our method fits heavier tails (areas with large $|\Delta|$); compare \cref{fig:prior-amzn} and \cref{apx:fig:prior-amzn-mog}.

As a consequence of the heavier tails in the estimated prior, our method imposes \emph{milder} shrinkage when making posterior inference. To illustrate this effect, in \cref{fig:shrinkage-amzn} we plot the amount of shrinkage in the posterior mean corresponding to different noise scales $s$. By Tweedie's formula \cref{eqs:tweedie-mean}, the amount of shrinkage upon observing $\hat{\Delta} = x$ and $\hat{s}=s$ is
\begin{equation*}
\hat{\E}_{s}[\Delta \mid \hat{\Delta}=x] - x = s^2 \hat{\ell}'_s(x) = s^2 \hat{p}'_s(x) / \hat{p}_s(x).
\end{equation*}
Note that strict positivity of $\hat{p}_s(x)$ proved in \cref{thm:non-negative} guarantees that the estimated posterior mean is finite. Also, \cref{lem:unif-consistency-posterior-mean} guarantees that the shrinkage curves will uniformly converge to the corresponding true curves as $n \rightarrow \infty$.
In the bottom panel of \cref{fig:shrinkage-amzn}, we compare to the shrinkage functions from estimating the prior as a mixture of three Gaussians. We can see that the mixture-of-Gaussian model imposes much \emph{stronger} shrinkage compared to our method (the dashed diagonal lines in the plots represent the strongest possible shrinkage that always shrinks the posterior mean to exactly zero).

\begin{figure}[!htb]
\centering
\includegraphics[width=0.75\textwidth]{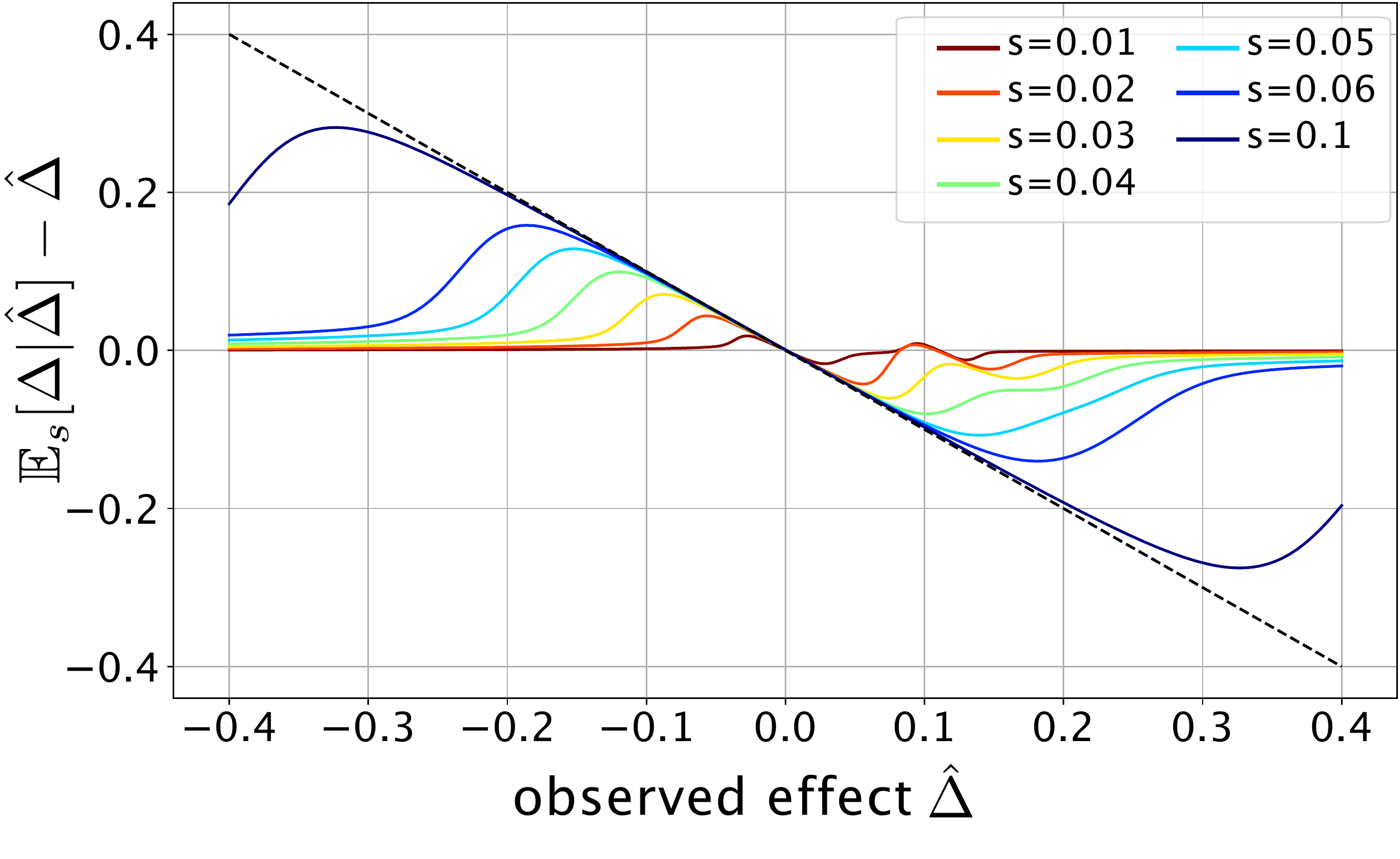}
\includegraphics[width=0.75\textwidth]{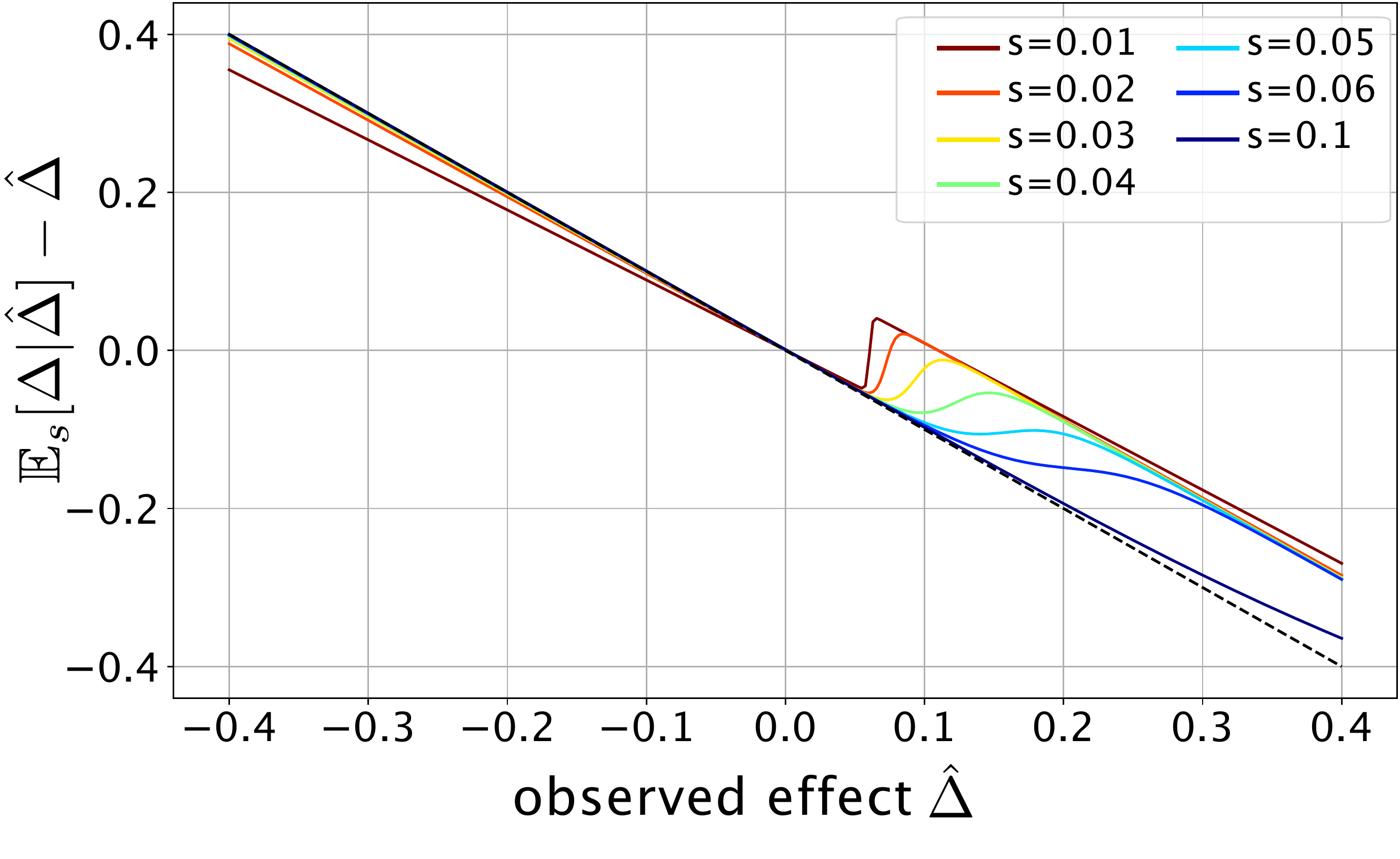}
\caption{The amount of shrinkage in the posterior mean function $\E_{s}[\Delta \mid \hat{\Delta}] - \hat{\Delta} = s^2 \ell'_s(\hat{\Delta})$ estimated from Amazon data (top: our spectral method, bottom: prior fitted with a mixture of three Gaussians). Colors correspond to different values of $s$. The strongest shrinkage is $y=-\hat{\Delta}$ as $s \rightarrow \infty$ (dashed diagonal line), which always sets the posterior mean to zero. See also \Cref{apx:fig:shrinkage-amzn} for a wider range of $\hat{\Delta}$.}
\label{fig:shrinkage-amzn}
\end{figure}

Additionally, the posterior density of $\Delta$ upon observing $s$ and $\hat{\Delta}=x$ can be computed via 
\begin{equation}
\hat{p}_{s}(\Delta \mid \hat{\Delta}=x) = \hat{p}_0(\Delta) \phi_{s}(\Delta - x) / \hat{p}_{s}(x).
\end{equation}
See \cref{apx:fig:posterior-dens-amzn} for an example of the estimated posterior density. Again, we can observe that the prior fitted with mixture of Gaussians implies stronger shrinkage; note that under the mixture of Gaussian model, little mass accumulates around $\hat{\Delta}$ for any of the noise scales considered. 

\begin{figure}[!htb]
\centering
\includegraphics[width=0.75\textwidth]{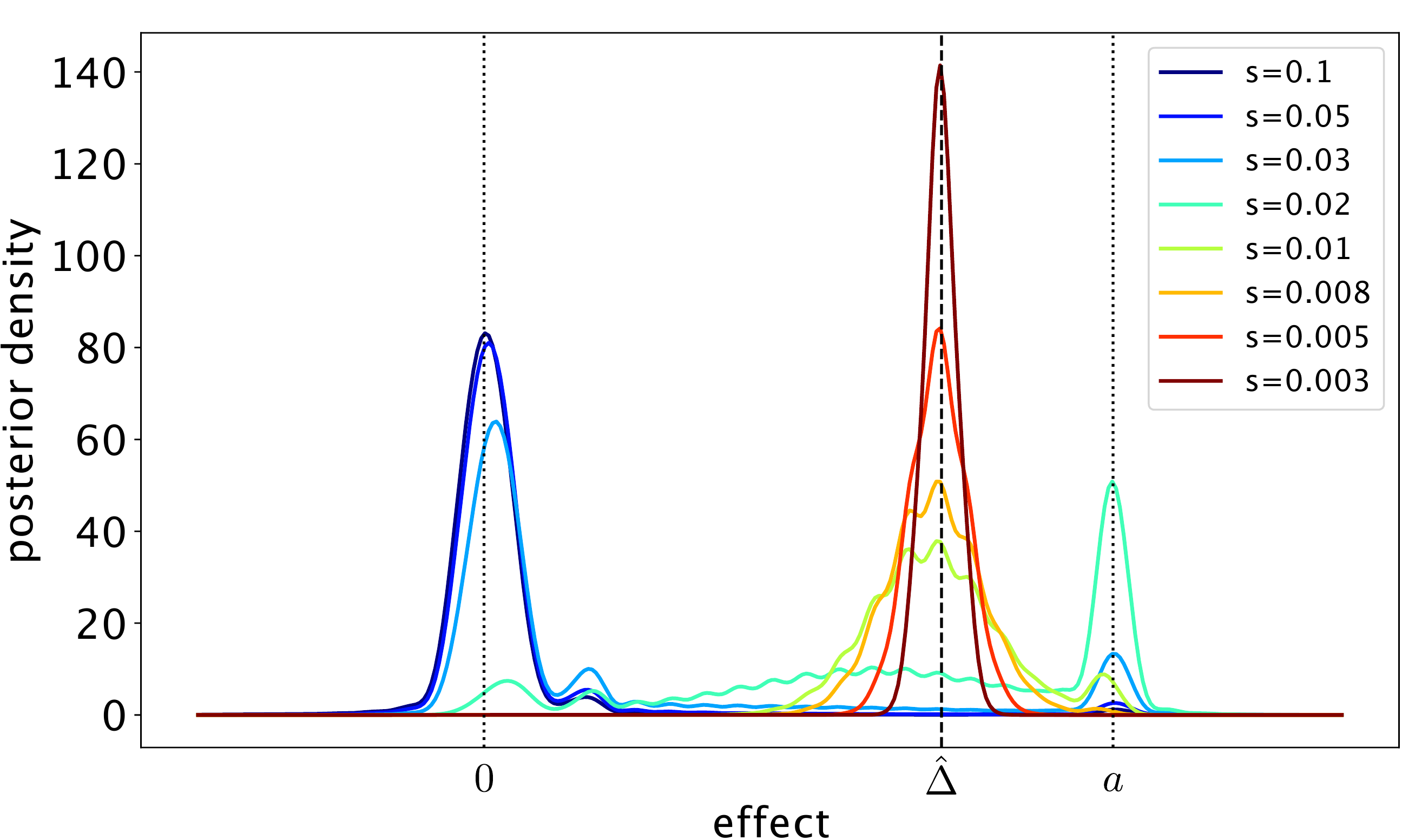}
\includegraphics[width=0.75\textwidth]{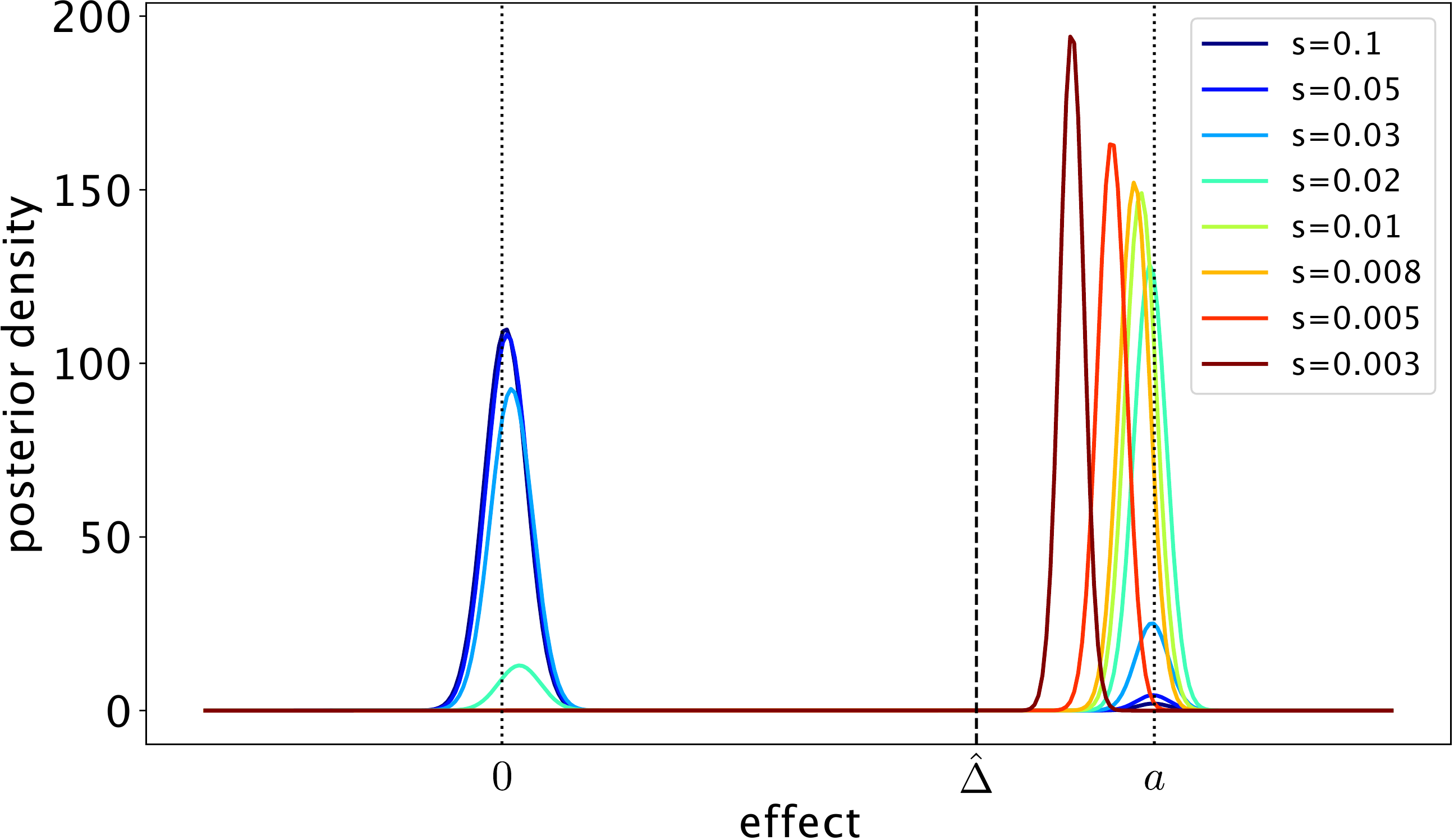}
\caption{Estimated posterior density $p_{s}(\Delta \mid \hat{\Delta})$ for Amazon data under different $s$ (top: our spectral method, bottom: prior fitted with a mixture of three Gaussians). The observed $\hat{\Delta}$ is marked as dashed. As the error scale $s$ decreases from $0.1$ (corresponding to increasing experimental sample size), the posterior is initially concentrated at zero, but then gradually shifts towards $a$ (a prior mode near $\hat{\Delta}$), and should finally concentrate around $\hat{\Delta}$. It is clear that the mixture-of-Gaussian prior imposes stronger shrinkage.}
\label{apx:fig:posterior-dens-amzn}
\end{figure}

 \section{Concluding remarks} \label{sec:discussion}
We have developed a new, principled and intuitive framework for the analysis of large-scale randomized experiments. We first characterized the density family arising from the problem with a PDE (the heat equation), which unifies ``$f$-modeling'' (marginal) and ``$g$-modeling'' (prior) approaches towards empirical Bayes \citep{efron2014two} under (asymptotic) Gaussian likelihood. Second, we estimated the density family with trigonometric polynomials, which are eigenfunctions of the heat equation. Third, we introduced a novel parametrization of non-negative trigonometric polynomials that ensures a bona fide probability density, which further guarantees that the implied posterior mean is finite. Fourth, we presented an efficient convex optimization algorithm for maximum likelihood estimation. Moreover, towards model selection, we connected the square loss in estimating the posterior mean to \possessivecite{hyvarinen2005estimation} score-matching via Tweedie's formula. Lastly, we showed that our estimator, as a sieve MLE, is uniformly consistent in estimating the prior density and the posterior mean functions. Our methodology provides a simple and scalable approach to analyzing large-scale randomized experiments, that offers many advantages relative to commonly used $t$-tests or fitting the prior to a parametric mixture model with the EM algorithm.

It is worth mentioning that the methodology developed here is applicable to more general settings, where one wants to estimate the prior distribution for some parameter $\theta \in \mathbb{R}$ from many asymptotically normal estimators $\hat{\theta}_i$, which correspond to the parameters $\theta_i$ realized independently from the prior. Here is an example. 

\begin{example}
Suppose an urn is filled with an infinite number of coins, and the probability of heads for a random coin is $\theta \in [0,1]$. For each experiment $i$, one can draw a coin from the urn, flip $n_i$ times and observe the number of heads $m_i$. The MLE for $\theta_i$ is $\hat{\theta}_i = m_i / n_i$, which is asymptotically normal when $n_i \rightarrow \infty$. Under a large number of coin tosses for each experiment, the observational model is equivalent to $\theta_i \iid G$, $\hat{\theta}_i = \theta_i + s_i Z_i$ with $s_i \approx \sqrt{\hat{\theta}_i (1 - \hat{\theta}_i) / n_i}$; compare with \cref{eqs:observational}.
\end{example}

We conclude with a few remarks. First, closely related to our work, \citet{walter1981orthogonal,walter1991bayes} and \citet{carrasco2011spectral} considered estimating the prior with orthogonal polynomials chosen with respect to the Gaussian convolution/integration kernel, which are Hermite polynomials; see \citet[Example 1]{carrasco2011spectral} and \citet[Appendix B.1]{walter1991bayes}. Here we take a differential perspective instead of an integral one, and the resulting polynomials are trigonometric. Second, in terms of nonparametric maximum likelihood, \citet{laird1978nonparametric} and \citet{leonard1984some} considered estimating the prior in the form of a mixture of delta functions, fitted with some iterative schemes such as the EM. Third, we leave the removal or relaxation of our technical assumptions to further studies, such as the toric simplification of the domain (\cref{assump:torus}). 
\section*{Acknowledgements}
RG thanks Hongxiang Qiu for pointing to references on sieve estimation, and Don Percival for detailed comments.

\appendix
\section{Bayesian analysis of experiments} \label{apx:bayes}
\subsection{Proof of \cref{thm:bayes-freq-coverage}}
\begin{proof}
By Bayes' theorem, 
\begin{equation*}
P \left(\Delta_i \in A \,\left|\, G_{A}(\hat{\Delta}_i) = \varphi \right. \right) = \frac{P(\Delta_i \in A, G_{A}(\hat{\Delta}_i) = \varphi)}{P(G_{A}(\hat{\Delta}_i) = \varphi)},
\end{equation*}
where the denominator is positive by assumption. By the law of total probability, the numerator becomes
\begin{equation*}
\begin{split}
P(\Delta_i \in A, G_{A}(\hat{\Delta}_i) = \varphi) &= \int_{A} P\left(G_{A}(\hat{\Delta}_i) = \varphi \,\left|\, \Delta_i = u \right. \right) g(u) \dd u \\
&= \int_{A} \left( \int_{\{x: G_{A}(x) = \varphi \}} \phi_{s_i}(x - u) \dd x \right ) g(u) \dd u \\
&= \int_{A} \int_{\{x: G_{A}(x) = \varphi\}} \frac{\phi_{s_i}(x - u) g(u)}{p_{s_i}(x)} p_{s_i}(x) \dd u \dd x \\
&\overset{(i)}{=} \int_{\{x: G_{A}(x) = \varphi\}} \left ( \int_{A} \frac{\phi_{s_i}(x - u) g(u)}{p_{s_i}(x)} \dd u \right) p_{s_i}(x) \dd x \\
&\overset{(ii)}{=} \int_{\{x: G_{A}(x) = \varphi\}} G_{A}(x) p_{s_i}(x) \dd x \\
&= \varphi \int_{\{x: G_{A}(x) = \varphi\}} p_{s_i}(x) \dd x \\
&= \varphi P(G_{A}(\hat{\Delta}_i) = \varphi),
\end{split}
\end{equation*}
where (i) uses Fubini's theorem, and (ii) uses \cref{eqs:posterior-density}. The result is proven by noting that $P(G_{A}(\hat{\Delta}_i) = \varphi)$ cancels with the denominator.
\end{proof}

\subsection{Bayes optimal decision making}
In the context of \cref{sec:bayes-optimal-decision}, suppose $\mathcal{A}$ is the collection of possible actions, and we have a loss function $l(\Delta, a): \mathbb{R} \times \mathcal{A} \rightarrow \mathbb{R}$. Without loss of generality, consider a deterministic decision function $d(\hat{\Delta}): \mathbb{R} \rightarrow \mathcal{A}$. The expected loss, with respect to the sampling distribution of future experiments, of making decision $d$ is
\begin{equation*}
\begin{split}
\E l(\Delta, d(\hat{\Delta})) &= \int g(\Delta) \left [ \int l(\Delta, d(\hat{\Delta})) p(\hat{\Delta} \mid \Delta) \dd \hat{\Delta} \right ] \dd \Delta \\
&= \int p(\hat{\Delta}) \left [\int l(\Delta, d(\hat{\Delta})) p(\Delta \mid \hat{\Delta}) \dd \Delta \right ] \dd \hat{\Delta} \\
& \geq \int p(\hat{\Delta})  \left [\int l(\Delta, d_B(\hat{\Delta})) p(\Delta \mid \hat{\Delta}) \dd \Delta \right ] \dd \hat{\Delta},
\end{split}
\end{equation*}
where the Bayes decision rule $d_B(\hat{\Delta})$ minimizes the posterior risk $\int l(\Delta, d(\hat{\Delta})) p(\Delta \mid \hat{\Delta}) \dd \Delta$ for every $\hat{\Delta}$. 

\section{General solution to the heat equation} \label{apx:fourier}
For a function $h: \mathbb{R} \rightarrow \mathbb{C}$ and $h \in L^1(\mathbb{R})$, the Fourier transform $\mathcal{F}h: \mathbb{R} \rightarrow \mathbb{C}$ is defined as 
\begin{equation*}
\mathcal{F}h(\xi) := \frac{1}{\sqrt{2 \pi}} \int_{\mathbb{R}} e^{-i x \xi} h(x) \dd x. 
\end{equation*}
And its inverse Fourier transform is defined as
\begin{equation*}
\mathcal{F}^{-1}h(\xi) := \frac{1}{\sqrt{2 \pi}} \int_{\mathbb{R}} e^{i x \xi} h(x) \dd x. 
\end{equation*}
It holds that $h = \mathcal{F}^{-1} \mathcal{F} h$. The Fourier transform for $p_t(x)$ is therefore
\begin{equation*}
(\mathcal{F}p_{t})(\xi) = \frac{1}{\sqrt{2 \pi}} \int_{\mathbb{R}} e^{-i x \xi} p_{t}(x) \dd x,
\end{equation*}
which is also called the characteristic function of $p_t$ in statistics.
Taking Fourier transform on both sides of \cref{eqs:heat} with respect to $x$, we have
\begin{equation*}
\mathcal{F} \left(\frac{\partial p_{t}(x)}{\partial t} \right)(\xi) = \frac{1}{2} \mathcal{F} \left(\frac{\partial^2 p_t(x)}{\partial x^2} \right)(\xi).
\end{equation*}
Using the property that $\mathcal{F}\left(\frac{\partial^m h}{\partial x^m}\right)(\xi) = (i \xi)^{m} \mathcal{F}h(\xi)$, the previous display becomes
\begin{equation*}
\frac{\partial}{\partial t} \mathcal{F}(p_t)(\xi) = -\frac{1}{2} \xi^2 \mathcal{F}(p_t)(\xi),
\end{equation*}
which is an ordinary differential equation in $t$. The solution is
\begin{equation}
\mathcal{F}(p_t)(\xi) = \mathcal{F}(p_0)(\xi) e^{-\frac{1}{2} \xi^2 t} = \mathcal{F}(g)(\xi) e^{-\frac{1}{2} \xi^2 t}, \label{eqs:solution-fourier-infinite}
\end{equation}
where $\mathcal{F}(g)$ is the Fourier transform of the prior density. The factor $e^{-\frac{1}{2} \xi^2 t}$ describes how components corresponding to different $\xi$ are damped over time. The inverse Fourier transform of \cref{eqs:solution-fourier-infinite} leads to the solution \cref{eqs:green} in terms of Green function. From the previous display, it becomes clear that we can evaluate the density family $p_t(x) = (\mathcal{F}^{-1}\mathcal{F}p_t)(x)$ at any $x \in \mathbb{R}$,  $t > 0$ if we can represent $\mathcal{F}(g)(\xi)$ for $\xi \in \mathbb{R}$; and the density family can be approximated if we can approximate $\mathcal{F}(g)$.

\section{Proof of consistency} \label{apx:sec:consistency}
Our proof of consistency will be based on the following result. 
\begin{lemma}[{\citet[Page 5590]{chen2007large}}] \label{lem:Chen}
Let $(\Theta, d)$ be a metric space. Let $\{\Theta_N\}$ be a sequence of sieves such that $\Theta_N \subseteq \Theta$. Let $\hat{\theta}_n$ be an (approximate) maximizer of the sample criterion function within sieve $\Theta_N$, namely one that satisfies
\begin{equation} \label{eqs:approximate-mle}
Q_n(\hat{\theta}_n) \geq \sup_{\theta \in \Theta_{N}} Q_n(\theta) - o_p(1), 
\end{equation}
for some $N \rightarrow \infty$ as $n \rightarrow \infty$. 
Then $d(\hat{\theta}_n, \theta_0) \rightarrow_{p} 0$ under the following conditions. 
\begin{enumerate}[(a)]
\item (i) $Q(\theta)$ is continuous at $\theta_0$ in $\Theta$, $Q(\theta_0) > -\infty$; (ii) for all $\varepsilon > 0$, $Q(\theta_0) > \sup_{\theta \in \Theta: d(\theta, \theta_0) \geq \varepsilon} Q(\theta)$. 

\item For any $\theta \in \Theta$ there exists $\theta_N \in \Theta_N$ such that $d(\theta, \theta_N) \rightarrow 0$ as $N \rightarrow \infty$. 

\item For each $k \geq 1$, (i) $Q_n(\theta)$ is measurable for all $\theta \in \Theta$, and (ii) $Q_n(\theta)$ is upper semicontinuous on $\Theta_k$ under metric $d(\cdot, \cdot)$. 

\item $\Theta_k$ is compact under $d(\cdot, \cdot)$ for every $k \geq 1$.

\item For every $k \geq 1$, $\sup_{\theta \in \Theta_k} \|Q_n(\theta) - Q(\theta)\| \rightarrow_{p} 0$. 
\end{enumerate}
\end{lemma}

\begin{remark}
\citet{chen2007large} also requires that $\Theta_{N} \subseteq \Theta_{N+1}$ for every $N$. But establishing consistency does not require the nestedness of sieves; see also \citet[Lemma A1]{newey2003instrumental} and \citet[Conditions (A-D)]{shen1997methods}.
\end{remark}

For the ease of presentation, we will adopt the standard notations from \citet{chen2007large}. Without loss of generality, let us assume $L = \pi$. Let $\Theta$ denote the parameter space under \cref{assump:prior}, namely
\begin{equation} \label{eqs:param-space}
\Theta := \left\{\theta(x): \theta \text{ is a uniformly continuous density on $[-\pi, \pi]$} \right\}.
\end{equation}
For $\theta \in \Theta$, 
\begin{equation} \label{eqs:p-theta-x-t}
p_{\theta, t} := \int_{-\infty}^{\infty} \bar{\theta}(u) \phi_{\sqrt{t}} (x - u) \dd u
\end{equation}
is the implied marginal density on $[-\pi, \pi]$, where $\bar{\theta}: \mathbb{R} \rightarrow \mathbb{R}$ is the periodicized version of $\theta$ satisfying $\bar{\theta}(x + 2 \pi) = \bar{\theta}(x)$. 
Let us define the population ``criterion function'' as the expected log-likelihood $Q(\theta): \Theta \rightarrow \mathbb{R}$, given by
\begin{equation} \label{eqs:Q}
Q(\theta) := \E_{t} \E_{X \mid t} \log p_{\theta, t}(X).
\end{equation}
\begin{lemma} \label{lem:identification}
Under $\theta_0$, $Q(\theta)$ is uniquely maximized at $\theta_0$. 
\end{lemma}
\begin{proof}
It suffices to show that for any fixed $t \geq 0$, $Q(\theta; t):= \E_{X\mid t} \log p_{\theta, t}(X)$ is uniquely maximized at $\theta = \theta_0$. Note that, $-Q(\theta; t) = \KL(p_{\theta_0, t} \| p_{\theta, t}) + \text{const}$, which is uniquely minimized when $p_{\theta, t} = p_{\theta_0, t}$. By \cref{lem:identifiable}, $p_{\theta, t} = p_{\theta_0, t}$ iff $\theta_0 = \theta$ Lebesgue almost everywhere. Further, since $\theta, \theta_0$ are continuous by definition of $\Theta$, $\theta_0 = \theta$ Lebesgue almost everywhere iff $\theta_0 = \theta$. 
\end{proof}
The sample version of the criterion function is defined as
\begin{equation} \label{eqs:Qn}
Q_n(\theta) := \mathbb{P}_n \log p_{\theta, t}(X),
\end{equation}
where $\mathbb{P}_n$ is the empirical measure over $(t, X)$. The maximum likelihood estimator $\widehat{C}_{N}(x; t=0)$ of the prior density can be written as 
\begin{equation} \label{eqs:theta-n}
\hat{\theta}_n = \argmax_{\theta \in \Theta_N} Q_n(\theta),
\end{equation}
where $N$ grows with $n$ and $\Theta_N$ is the class of densities representable under our parametrization. More precisely, by \cref{eqs:cesaro-kernel-form} and the conditions $f_\nu \geq 0$, $\frac{2 \pi}{2N + 1} \sum_{\nu=-N}^{N} f_\nu = 1$, the class can be expressed as
\begin{equation}
\Theta_N = \left\{\sum_{\nu=-N}^{N} \gamma_\nu \widetilde{K}_{N} (x - x_\nu): \gamma \in \mathcal{S}_{2N} \right \},
\end{equation}
where $\gamma_\nu = 2 \pi f_\nu / (2N+1)$, $\widetilde{K}_{N}(\cdot) = K_{N}(\cdot) / (2\pi)$ and $x_\nu = \frac{2 \pi \nu}{2 N + 1}$. Recall that $\mathcal{S}_{2N}$ is the $2N$-dimensional unit simplex. 

\begin{lemma} \label{lem:sieve-compact}
$\Theta_N$ is compact in $\|\cdot\|_{\infty}$ for every $N \geq 1$.
\end{lemma}
\begin{proof}
The metric space $(\Theta_N, \|\cdot\|_{\infty})$ is compact iff it is sequentially compact \citep[Theorem 4.3.14]{kumaresan2005topology}. Now we show sequential compactness. Since $\mathcal{S}_{2N}$ is compact, for any sequence $\gamma^{(n)}$ in $\mathcal{S}_{2N}$ there exists a subsequence $\gamma^{(m_n)} \rightarrow \gamma \in \mathcal{S}_{2N}$. Let $\theta_\gamma$ be the density that corresponds to $\gamma$. We have 
\begin{equation*}
\|\theta_{\gamma^{(m_n)}} - \theta_{\gamma}\|_{\infty} = \left\|\sum_{\nu} \left(\gamma_\nu^{(m_n)} - \gamma_\nu \right) \widetilde{K}_{N}(x - x_k)\right\|_{\infty} \leq \bar{K}_{N} \|\gamma^{(m_n)} - \gamma\|_{\infty} \rightarrow 0,
\end{equation*}
where $\bar{K}_{N}:=\widetilde{K}_{N}(0)=\frac{N+1}{2\pi} < \infty$; see \cref{lem:fejer}(a). 
\end{proof}

$\{\Theta_N\}$ is a sequence of sieves in the sense of \citet{grenander1981abstract}. We show that the  sieves are dense in $\Theta$ with respect to $\|\cdot\|_{\infty}$. 
\begin{lemma} \label{lem:sieve-dense}
For every $\theta \in \Theta$, there exists $\theta_N \in \Theta_N$ such that $\|\theta_N - \theta\|_{\infty} \rightarrow 0$. 
\end{lemma}
\begin{proof}
Let us first consider the sieves without the simplex constraint. Define
\begin{equation*}
\widetilde{\Theta}_N := \left\{\sum_{\nu=-N}^{N} \gamma_\nu \widetilde{K}_{N} (x - x_\nu): \gamma \in \mathbb{R}^{2N+1} \right \}.
\end{equation*}
Consider $\tilde{\theta}_N := \tilde{\theta}_{\gamma^{(N)}} \in \widetilde{\Theta}_N$ determined by $\gamma^{(N)}_\nu = \frac{2 \pi}{2N + 1} \theta(x_\nu) \geq 0$. \citet[Theorem 6.3, Chapter X]{zygmund2002trigonometric} shows that (i) $\|\tilde{\theta}_N - \theta\|_{\infty} \rightarrow 0$, and (ii) $\tilde{\theta}_N$ remains within the same bounds as $\theta$. Now let us consider $\theta_N := \theta_{\gamma^{(N)}} \in \Theta_N$ with normalized weights $\gamma^{(N)} = \tilde{\gamma}^{(N)} / \|\tilde{\gamma}^{(N)}\|_{1} \in \mathcal{S}_{2N}$. Using (ii), we have
\begin{equation*}
\|\tilde{\theta}_N - \theta_N\|_{\infty} = \left|1 - \|\tilde{\gamma}^{(N)}\|_{1}^{-1} \right| \|\tilde{\theta}_N\|_{\infty} \leq \left|1 - \|\tilde{\gamma}^{(N)}\|_{1}^{-1} \right| c, 
\end{equation*}
where $c$ is the maximum density of $\theta$. Clearly, $c < \infty$ since $\theta$ is a continuous density on $[-\pi, \pi]$. Note that
\begin{equation*}
1 = \int_{-\pi}^{\pi} \theta(x) \dd x = \sum_{\nu=-N}^{N} \int_{x_\nu}^{x_{\nu+1}} \theta(x) \dd x = \frac{2\pi}{2N+1} \sum_{\nu=-N}^{N} \theta(x_\nu')
\end{equation*}
for $x_\nu' \in [x_\nu, x_{\nu+1}]$ by the mean value theorem. We have
\begin{equation*}
\begin{split}
\left | \|\tilde{\gamma}^{(N)}\|_{1} - 1 \right| &= \left |\frac{2 \pi}{2N+1} \sum_{\nu=-N}^{N}(\theta(x_\nu) - \theta(x_\nu')) \right| \\
& \leq 2 \pi \sup_{|x - x'| \leq 2 \pi / (2N+1)} |\theta(x) - \theta(x')| \rightarrow 0,
\end{split}
\end{equation*}
by uniform continuity of $\theta$. It follows that $\|\tilde{\theta}_N - \theta_N\|_{\infty} \rightarrow 0$. Finally, $\|\theta_N - \theta\|_{\infty} \leq \|\theta_N - \tilde{\theta}_N\|_{\infty} + \|\tilde{\theta}_N - \theta\|_{\infty} \rightarrow 0$, using (i) above.
\end{proof}

Next, we show a uniform law of large numbers for $\Theta_N$ using the following lemma.
\begin{lemma}[{Theorem 19.4 of \citet{van2000asymptotic}}] \label{lem:GC-bracketing}
Every class $\mathcal{F}$ of measurable functions such that $N_{[~]}(\varepsilon, \mathcal{F}, L_1(P)) < \infty$ for every $\varepsilon > 0$ is $P$-Glivenko-Cantelli.
\end{lemma}
A function class $\mathcal{F}$ is called $P$-Glivenko-Cantelli if $\sup_{f \in \mathcal{F}} | (\mathbb{P}_n - P) f| \rightarrow_{p} 0$, where $\mathbb{P}_n$ is the empirical measure of $P$ under $n$ iid samples. Given two functions $l$ and $u$, let $[l, u]:=\{f: l(x) \leq f(x) \leq u(x) \}$. $[l,u]$ is called an $\varepsilon$-bracket in $L_1(P)$ if $P|l - u| \leq \varepsilon$. The bracketing number $N_{[~]}(\varepsilon, \mathcal{F}, L_1(P))$ is the smallest cardinality $|I|$ of a set of $\varepsilon$-brackets that covers $\mathcal{F}$, i.e., $\mathcal{F} \subseteq \bigcup_{i \in I} [l_i, u_i]$. Similarly, we use $N(\varepsilon, \mathcal{F}, \|\cdot\|)$ to denote the covering number of $\mathcal{F}$, namely the smallest cardinality of $\{f_i: i \in I\}$ such that $\mathcal{F} \subseteq \bigcup_{i \in I} \{f: \|f - f_i\| \leq \varepsilon\}$. See \citet[Chapter 19]{van2000asymptotic} for more background. 

\begin{lemma} \label{lem:sieve-GC}
Under \cref{assump:t-bounded}, $\sup_{\theta \in \Theta_N} \left|Q_n(\theta) - Q(\theta) \right|\rightarrow_{p} 0$ as $n \rightarrow \infty$ for every $N \geq 1$.
\end{lemma}
\begin{proof}
By \cref{assump:t-bounded}, let 
\begin{equation*}
\mathcal{F}_N := \left \{p_{\theta}(x,t): [-\pi, \pi] \times [0, t_{\max}] \rightarrow \mathbb{R} \mid \theta \in \Theta_N \right \},
\end{equation*}
where $p_{\theta}(x,t)$ is defined by \cref{eqs:p-theta-x-t}. By definition of $Q$ and $Q_n$ (see \cref{eqs:Q,eqs:Qn}), to show the lemma is to show that $\log \mathcal{F}_N := \{\log f: f \in \mathcal{F}_N\}$ is $P$-Glivenko-Cantelli for every $N \geq 1$. 

Fix any $N \geq 1$. By \cref{lem:GC-bracketing}, we can show the result by showing finite bracketing number for $\log \mathcal{F}_N$ for every $\varepsilon > 0$. First, we show finite bracketing number for $\mathcal{F}_N$. Note that $N(\varepsilon, \Theta_N, \|\cdot\|_{\infty}) < \infty$ for every $\varepsilon > 0$ by \cref{lem:sieve-compact}. Take $\{\theta_i: i \in I\}$ to be an $\varepsilon/2$-cover of $\Theta_N$ in $\|\cdot\|_{\infty}$. For any $p_{\theta_1},p_{\theta_2} \in \mathcal{F}_N$, it holds that
\begin{equation*}
\begin{split}
\|p_{\theta_1} - p_{\theta_2}\|_{\infty} &= \left \| \int_{-\infty}^{+\infty} (\bar{\theta}_1(u) - \bar{\theta}_2(u)) \phi_{\sqrt{t}} (x - u) \dd u \right \|_{\infty} \\
&\leq \|\theta_1(u) - \theta_2(u)\|_{\infty} \left|\int_{-\infty}^{+\infty} \phi_{\sqrt{t}} (x - u) \dd u  \right | = \|\theta_1(u) - \theta_2(u)\|_{\infty}.
\end{split}
\end{equation*}
Hence, for any $p_{\theta} \in \mathcal{F}_N$, there exists some $i \in I$ such that $\|p_{\theta} - p_{\theta_i} \|_{\infty} \leq \varepsilon / 2$. Under \cref{assump:t-bounded}, we claim that $\mathcal{F}_N$ is bounded from above and below by positive constants. The upper bound is given by
\begin{equation*}
\int_{-\infty}^{+\infty} \bar{\theta}(u) \phi_{\sqrt{t}}(x-u) \dd u \leq \|\theta\|_{\infty} \leq \bar{K}_{N},
\end{equation*}
where $\bar{K}_N$ appeared in the proof of \cref{lem:sieve-compact}. And the lower bound comes from
\begin{equation} \label{eqs:marginal-dens-lb}
\begin{split}
\int_{-\infty}^{+\infty} \bar{\theta}(u) \phi_{\sqrt{t}}(x-u) \dd u &\geq \int_{-\pi}^{+\pi} \theta(u) \phi_{\sqrt{t}}(x-u) \dd u \\
& \geq \inf\{\phi_{\sqrt{t}}(x): t \in [0, t_{\max}], x \in [-\pi, \pi] \} \\
& \geq \phi_{\sqrt{t_{\max}}} (2 \pi) > 0.
\end{split}
\end{equation}
For every $\theta_i$ in the cover, we construct a pair of brackets
\begin{equation*}
\begin{split}
l_i(x, t) &:= \left( \int_{-\infty}^{+\infty} \bar{\theta}_i(u) \phi_{\sqrt{t}}(x-u) \dd u - \varepsilon / 2 \right) \vee \phi_{\sqrt{t_{\max}}} (2 \pi), \\
u_i(x, t) &:= \left( \int_{-\infty}^{+\infty} \bar{\theta}_i(u) \phi_{\sqrt{t}}(x-u) \dd u + \varepsilon / 2 \right) \wedge \bar{K}_N.
\end{split}
\end{equation*}
By construction, $\mathcal{F}_N \subseteq \bigcup_{i \in I} [l_i, u_i]$ and $P |u_i - l_i| \leq \varepsilon$. Therefore, we have
\begin{equation*}
N_{[~]}(\varepsilon, \mathcal{F}_N, L_1(P)) \leq N(\varepsilon, \Theta_k, \|\cdot\|_{\infty}) < \infty.
\end{equation*}
Further, since $0 < \phi_{\sqrt{t_{\max}}} (2 \pi) \leq \mathcal{F}_N \leq \bar{K}_N < \infty$, $\log \mathcal{F}_N$ is a Lipschitz transform on $\mathcal{F}_N$ and it follows that $N_{[~]}(\varepsilon, \log \mathcal{F}_N, L_1(P)) < \infty$. By \cref{lem:GC-bracketing}, $\log \mathcal{F}_N$ is $P$-Glivenko-Cantelli. 
\end{proof}

Finally, consistency is established as follows. 
\begin{proof}[Proof of \cref{thm:consistency}]
We prove consistency by verifying the conditions in \cref{lem:Chen}, with $\Theta$ defined in \cref{eqs:param-space} and metric $d(f, g) := \|f-g\|_{\infty}$. Clearly, $\hat{\theta}_n$ as defined in \cref{eqs:theta-n} satisfies \cref{eqs:approximate-mle}. Condition (a) is satisfied due to (i) the continuity of $Q$ and (ii) \cref{lem:identification}. By \cref{lem:sieve-dense}, Condition (b) is satisfied. Condition (c) is clearly satisfied by the definition of $Q_n$ in \cref{eqs:Qn}. Finally, Condition (d) and (e) are verified by \cref{lem:sieve-compact} and \cref{lem:sieve-GC} respectively. 
\end{proof}

\section{Gaussian mixture model of the prior} \label{apx:sec:EM}
\subsection{The EM algorithm}
We first consider a parametric model by assuming that the prior $G$ is a mixture of $K$ Gaussians, namely
\begin{equation*}
G(\Delta) = \sum_{k=1}^{K} \alpha_k \N(\Delta; \mu_k, V_k),
\end{equation*}
where $\alpha_k \geq 0$, $\sum_{k=1}^{K} \alpha_k = 1$ and $V_k \geq 0$. 

The induced marginal likelihood on $\hat{\Delta}$ is again a mixture of Gaussians
\begin{equation*}
\hat{\Delta} \sim \sum_{k=1}^{K} \alpha_k \N(\mu_k, V_k + s^2).
\end{equation*}
And the marginal log-likelihood is
\begin{equation}
\ell_n = \sum_{i=1}^{N} \log \left\{ \sum_{k=1}^{K} \alpha_k \N(\hat{\Delta}_i; \mu_k, V_k + s^2) \right\}.  \label{eqs:marginal-ll-mog}
\end{equation}
Note that $\ell_n$ is not concave in $\{(\alpha_k, \mu_k, V_k^{-1})\}$. 

\paragraph{$\bm{K=1}$} In the simplest case of fitting one Gaussian, setting $\nabla \ell_n = 0$ yields
\begin{equation*}
\hat{\mu} = \frac{\sum_{i=1}^{N} (\hat{V} + s_i^2)^{-1} \hat{\Delta}_i}{\sum_{i=1}^{N} (\hat{V} + s_i^2)^{-1}}, \quad \hat{V} = \frac{1}{N} \sum_{i=1}^N (\hat{\Delta}_i - \hat{\mu})^2 - \frac{1}{N} \sum_{i=1}^N s_i^2.
\end{equation*}
A fixed-point can be found by iterating the two equations. Note the estimates are different from naively fitting a normal from $\{\hat{\Delta}_i\}$. 

\paragraph{$\bm{K > 1}$} It is not straightforward to optimize \cref{eqs:marginal-ll-mog} due to non-convexity. Instead, an EM algorithm \citep{dempster1977maximum} can be constructed by imputing latent component indicators $z_{ik}$ and unobserved true effects $\Delta_i$, such that we can iteratively maximize a lower-bound for $\ell_n$. This algorithm, in the more general multivariate setting, was derived by \citet{bovy2011extreme}. 

The E-step entails updating the following $N \times K$ matrices
\begin{equation*}
q_{ik} \leftarrow \frac{\alpha_k \N(\hat{\Delta}_i; \mu_k, V_k + s_i^2)}{\sum_{l=1}^K \alpha_l \N(\hat{\Delta}_i; \mu_l, V_l + s_i^2)},
\end{equation*}
\begin{equation*}
b_{ik} \leftarrow \frac{s_i^2 \mu_k + V_k \hat{\Delta}_i}{s_i^2 + V_k}, \quad B_{ik} \leftarrow \frac{s_i^2 V_k}{s_i^2 + V_k}.
\end{equation*}
The M-step updates the estimates for $k=1,\dots,K$ by
\begin{equation*}
\alpha_{k} \leftarrow \frac{1}{N} \sum_{i=1}^{N} q_{ik}, \quad \mu_k \leftarrow \frac{\sum_{i=1}^N q_{ik} b_{ik}}{\sum_{i=1}^N q_{ik}}, \quad V_k \leftarrow \frac{\sum_{i=1}^{N} q_{ik} \{(\mu_{k} - b_{ik})^2 + B_{ik}\}}{\sum_{i=1}^{N} q_{ik}}.
\end{equation*}
The EM algorithm converges to a local maximum of $\ell_n$. To approach the global maximum, one can run the algorithm multiple times with random initializations.

\subsection{Posterior inference}
Under a mixture of Gaussian prior, the posterior for $\Delta$ is also a mixture of Gaussians
\begin{equation*}
\Delta \mid \hat{\Delta}, s \sim \sum_{k=1}^K \alpha_k' \N(\mu_k', V_k'),
\end{equation*}
where the updated parameters are given by
\begin{equation*}
\alpha_k' \propto \frac{\alpha_k}{\sqrt{s^2 + V_k}} \exp \left \{ - \frac{(\mu_k - \hat{\Delta})^2}{2(s^2 + V_k)} \right \}, \quad \sum_{k=1}^{K} \alpha_k' = 1
\end{equation*}
and
\begin{equation*}
\mu_k' = \frac{s^2 \mu_k + V_k \hat{\Delta}}{s^2 + V_k}, \quad V_k' = \frac{V_k s^2}{s^2 + V_k}.
\end{equation*}
Suppose $\mu_1 = V_1 = 0$, i.e., the first component is a point mass at zero. Then $\alpha_1'$ can be interpreted as the posterior probability that the true effect is zero.

\section{Additional plots} \label{apx:sec:plots}
We provide additional plots \cref{apx:fig:prior-amzn,apx:fig:prior-amzn-mog,apx:fig:shrinkage-amzn}.

\begin{figure}[!htb]
\centering 
\includegraphics[width=0.85\textwidth]{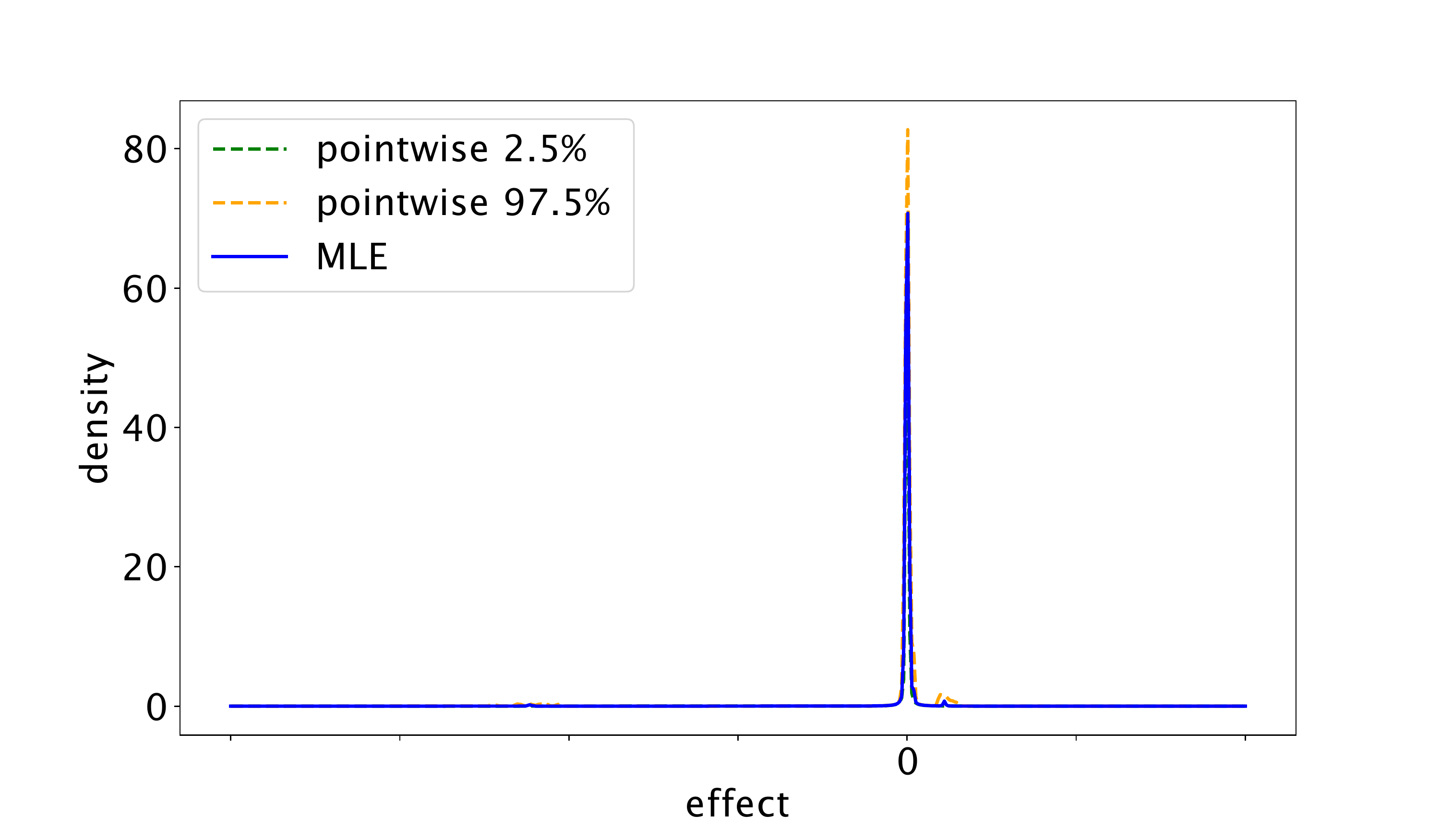}
\caption{The prior over the true effect estimated from a set of experiments run at Amazon using the spectral method developed in this paper. The pointwise confidence bands are estimated from bootstrap replications. See also \cref{fig:prior-amzn} in the main text for the logarithmic scale.}
\label{apx:fig:prior-amzn}
\end{figure}

\begin{figure}[!htb]
\centering
\includegraphics[width=0.75\textwidth]{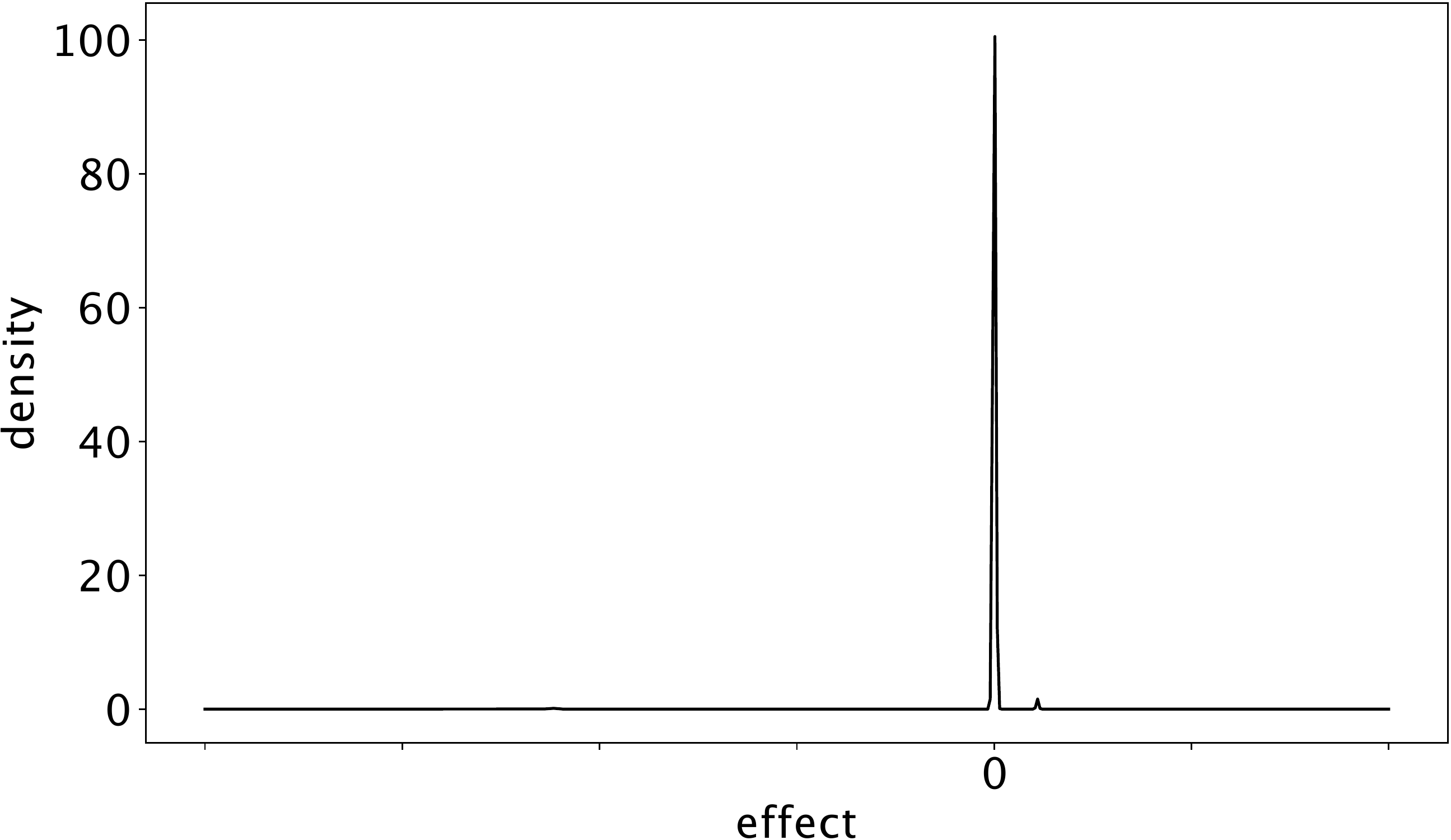}
\includegraphics[width=0.75\textwidth]{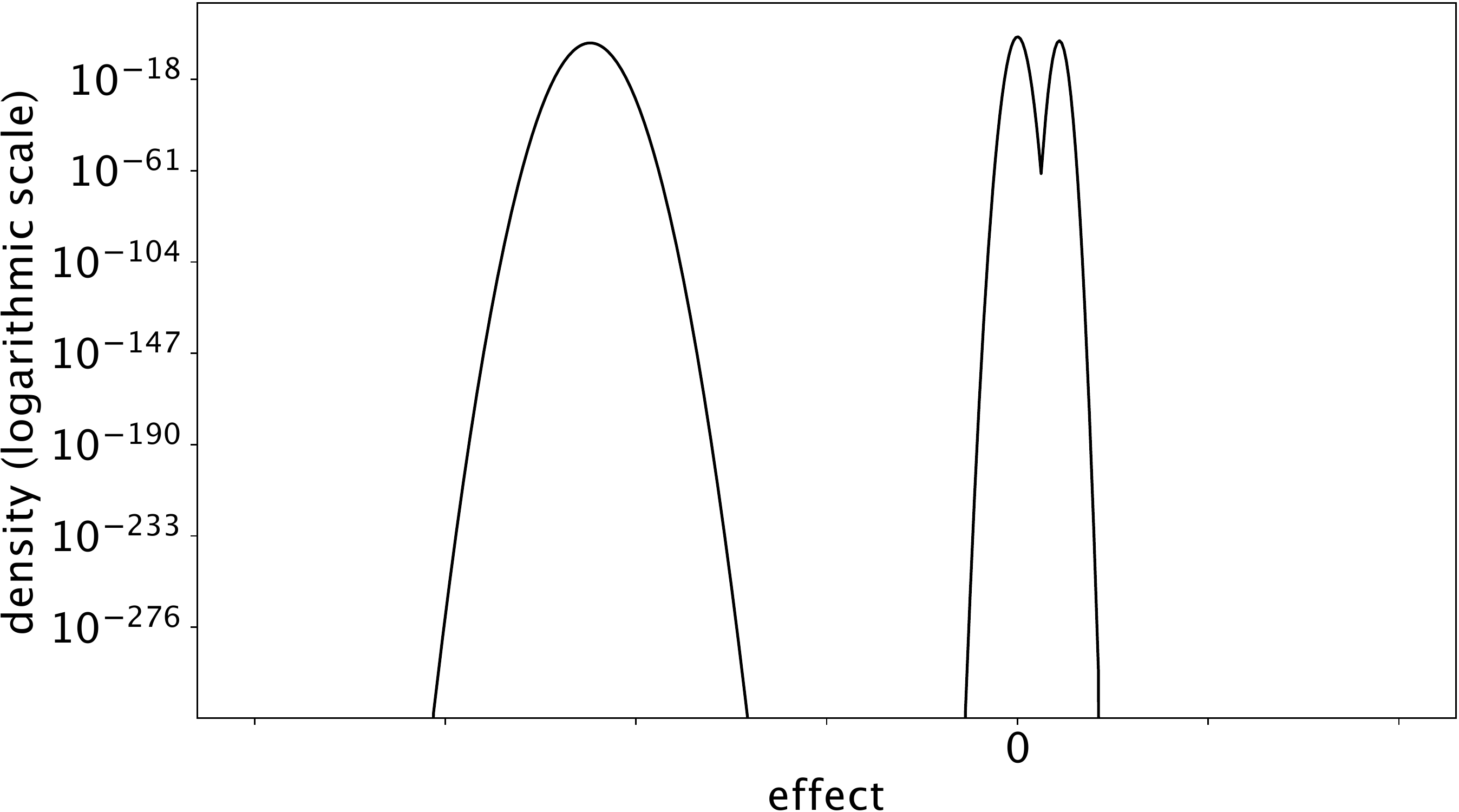}
\caption{The prior over the true effect estimated from a set of experiments run at Amazon fitted with a mixture of three Gaussians (top: in linear scale; bottom: in logarithmic scale). Compare to \cref{fig:prior-amzn} and \cref{apx:fig:prior-amzn} estimated by the spectral method.}
\label{apx:fig:prior-amzn-mog}
\end{figure}

\begin{figure}[!htb]
\centering
\includegraphics[width=0.75\textwidth]{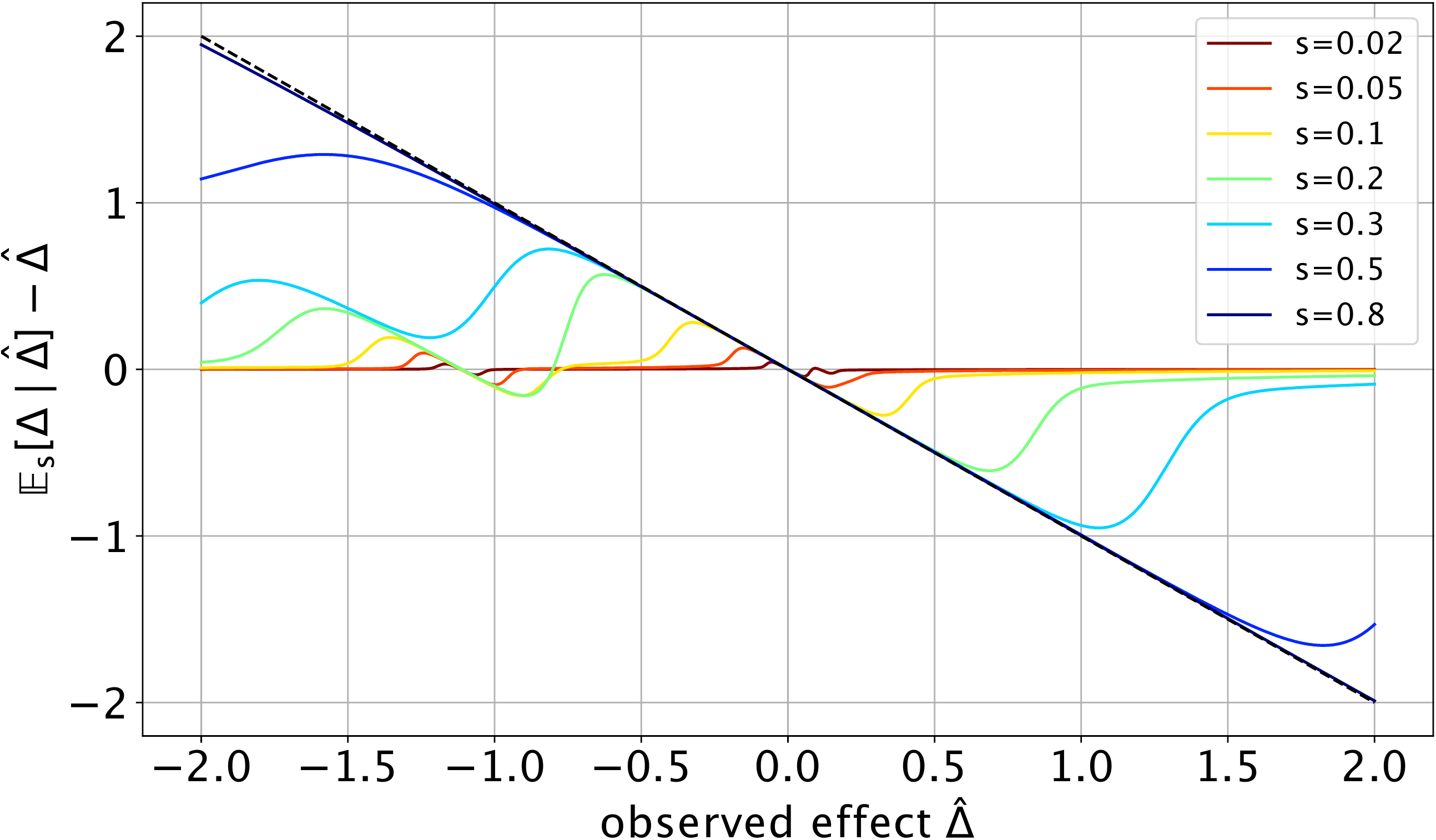}
\includegraphics[width=0.75\textwidth]{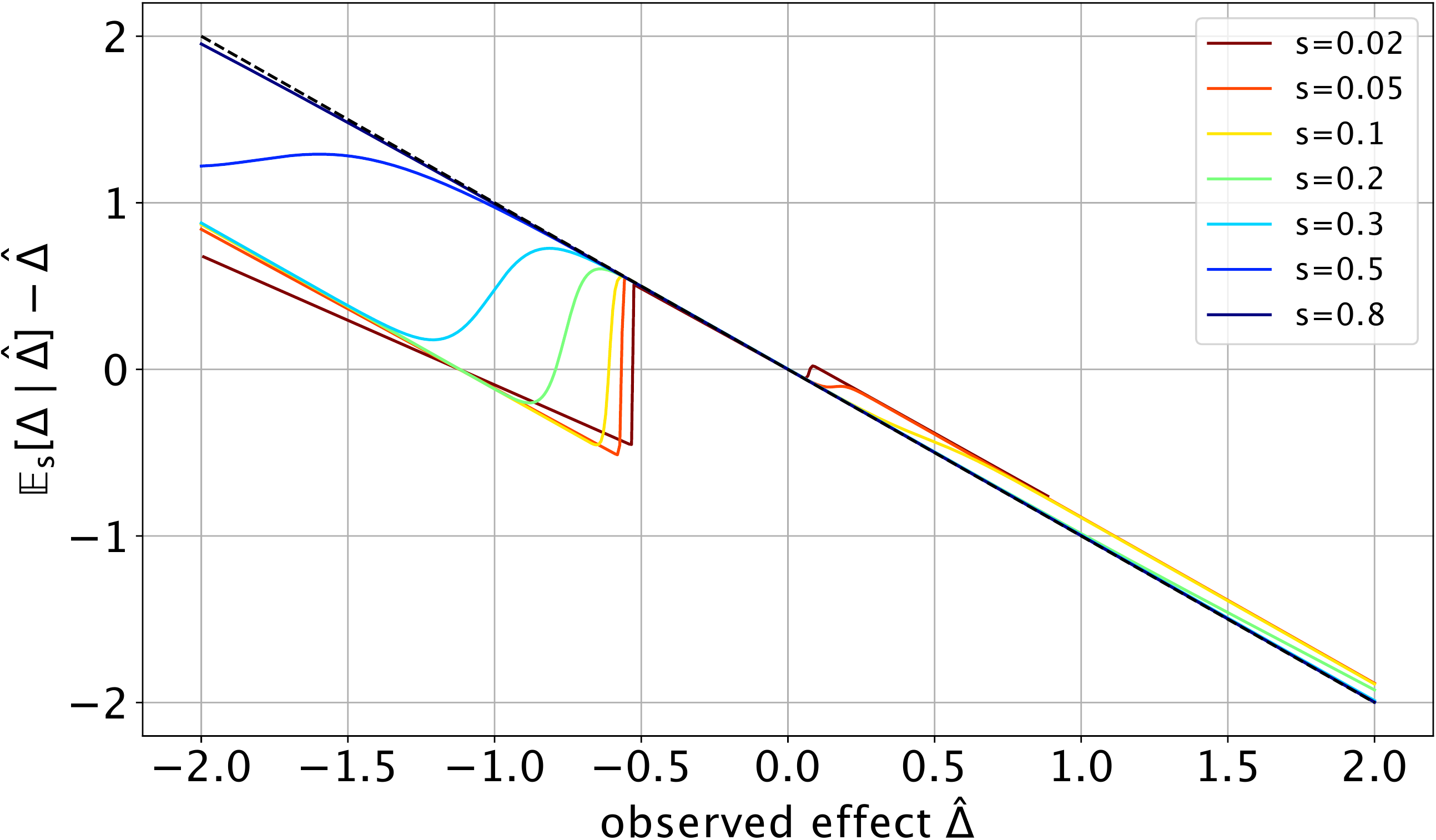}
\caption{The amount of shrinkage in the estimated posterior mean function $\E_{s}[\Delta \mid \hat{\Delta}] - \hat{\Delta} = s^2 \ell'_s(\hat{\Delta})$ for different $s$ (top: our spectral method, bottom: fitting the prior with a mixture of three Gaussians). The strongest shrinkage is $y=-\hat{\Delta}$ as $s \rightarrow \infty$ (dashed diagonal line), which always sets the posterior mean to zero. See also \cref{fig:shrinkage-amzn} for more details around zero.}
\label{apx:fig:shrinkage-amzn}
\end{figure} 

\bibliographystyle{plainnat}

\begin{thebibliography}{57}
\providecommand{\natexlab}[1]{#1}
\providecommand{\url}[1]{\texttt{#1}}
\expandafter\ifx\csname urlstyle\endcsname\relax
  \providecommand{\doi}[1]{doi: #1}\else
  \providecommand{\doi}{doi: \begingroup \urlstyle{rm}\Url}\fi

\bibitem[Amrhein et~al.(2019)Amrhein, Greenland, and
  McShane]{amrhein2019scientists}
Valentin Amrhein, Sander Greenland, and Blake McShane.
\newblock Scientists rise up against statistical significance.
\newblock \emph{Nature}, 2019.

\bibitem[Archambeau et~al.(2003)Archambeau, Lee, and
  Verleysen]{archambeau2003convergence}
C{\'e}dric Archambeau, John~Aldo Lee, and Michel Verleysen.
\newblock On convergence problems of the {EM} algorithm for finite gaussian
  mixtures.
\newblock In \emph{ESANN}, volume~3, pages 99--106, 2003.

\bibitem[Azevedo et~al.(2019)Azevedo, Alex, Montiel~Olea, Rao, and
  Weyl]{azevedo2019b}
Eduardo~M Azevedo, Deng Alex, Jose Montiel~Olea, Justin~M Rao, and E.~Glen
  Weyl.
\newblock {A/B} testing with fat tails.
\newblock \emph{Available at SSRN 3171224}, 2019.

\bibitem[Beck and Teboulle(2009)]{beck2009fast}
Amir Beck and Marc Teboulle.
\newblock A fast iterative shrinkage-thresholding algorithm for linear inverse
  problems.
\newblock \emph{SIAM Journal on Imaging Sciences}, 2\penalty0 (1):\penalty0
  183--202, 2009.

\bibitem[Botev et~al.(2010)Botev, Grotowski, and Kroese]{botev2010kernel}
Zdravko~I Botev, Joseph~F Grotowski, and Dirk~P Kroese.
\newblock Kernel density estimation via diffusion.
\newblock \emph{The Annals of Statistics}, 38\penalty0 (5):\penalty0
  2916--2957, 2010.

\bibitem[Bovy et~al.(2011)Bovy, Hogg, and Roweis]{bovy2011extreme}
Jo~Bovy, David~W Hogg, and Sam~T Roweis.
\newblock Extreme deconvolution: Inferring complete distribution functions from
  noisy, heterogeneous and incomplete observations.
\newblock \emph{The Annals of Applied Statistics}, 5\penalty0 (2B):\penalty0
  1657--1677, 2011.

\bibitem[Carrasco and Florens(2011)]{carrasco2011spectral}
Marine Carrasco and Jean-Pierre Florens.
\newblock A spectral method for deconvolving a density.
\newblock \emph{Econometric Theory}, 27\penalty0 (3):\penalty0 546--581, 2011.

\bibitem[Carroll and Hall(1988)]{carroll1988optimal}
Raymond~J Carroll and Peter Hall.
\newblock Optimal rates of convergence for deconvolving a density.
\newblock \emph{Journal of the American Statistical Association}, 83\penalty0
  (404):\penalty0 1184--1186, 1988.

\bibitem[Chaudhuri and Marron(2000)]{chaudhuri2000scale}
Probal Chaudhuri and James~Steven Marron.
\newblock Scale space view of curve estimation.
\newblock \emph{The Annals of Statistics}, pages 408--428, 2000.

\bibitem[Chen(2007)]{chen2007large}
Xiaohong Chen.
\newblock Large sample sieve estimation of semi-nonparametric models.
\newblock \emph{Handbook of Econometrics}, 6:\penalty0 5549--5632, 2007.

\bibitem[Delaigle and Meister(2008)]{delaigle2008density}
Aurore Delaigle and Alexander Meister.
\newblock Density estimation with heteroscedastic error.
\newblock \emph{Bernoulli}, 14\penalty0 (2):\penalty0 562--579, 2008.

\bibitem[Dempster et~al.(1977)Dempster, Laird, and Rubin]{dempster1977maximum}
Arthur~P Dempster, Nan~M Laird, and Donald~B Rubin.
\newblock Maximum likelihood from incomplete data via the {EM} algorithm.
\newblock \emph{Journal of the Royal Statistical Society: Series B
  (Methodological)}, 39\penalty0 (1):\penalty0 1--22, 1977.

\bibitem[Deng(2015)]{deng2015objective}
Alex Deng.
\newblock Objective {Bayesian} two sample hypothesis testing for online
  controlled experiments.
\newblock In \emph{Proceedings of the 24th International Conference on World
  Wide Web}. ACM, 2015.

\bibitem[Duchi et~al.(2008)Duchi, Shalev-Shwartz, Singer, and
  Chandra]{duchi2008efficient}
John Duchi, Shai Shalev-Shwartz, Yoram Singer, and Tushar Chandra.
\newblock Efficient projections onto the $l_1$-ball for learning in high
  dimensions.
\newblock In \emph{Proceedings of the 25th International Conference on Machine
  Learning}. ACM, 2008.

\bibitem[Dwivedi et~al.(2018)Dwivedi, Ho, Khamaru, Jordan, Wainwright, and
  Yu]{dwivedi2018singularity}
Raaz Dwivedi, Nhat Ho, Koulik Khamaru, Michael~I Jordan, Martin~J Wainwright,
  and Bin Yu.
\newblock Singularity, misspecification, and the convergence rate of {EM}.
\newblock \emph{arXiv preprint arXiv:1810.00828}, 2018.

\bibitem[Efron(2003)]{efron2003robbins}
Bradley Efron.
\newblock Robbins, empirical {Bayes} and microarrays.
\newblock \emph{The Annals of Statistics}, 31\penalty0 (2):\penalty0 366--378,
  2003.

\bibitem[Efron(2011)]{efron2011tweedie}
Bradley Efron.
\newblock Tweedie's formula and selection bias.
\newblock \emph{Journal of the American Statistical Association}, 106\penalty0
  (496), 2011.

\bibitem[Efron(2014)]{efron2014two}
Bradley Efron.
\newblock Two modeling strategies for empirical {Bayes} estimation.
\newblock \emph{Statistical Science}, 29\penalty0 (2):\penalty0 285, 2014.

\bibitem[Efron and Morris(1977)]{efron1977stein}
Bradley Efron and Carl Morris.
\newblock Stein's paradox in statistics.
\newblock \emph{Scientific American}, 236\penalty0 (5):\penalty0 119--127,
  1977.

\bibitem[Efron and Tibshirani(1996)]{efron1996using}
Bradley Efron and Robert Tibshirani.
\newblock Using specially designed exponential families for density estimation.
\newblock \emph{The Annals of Statistics}, 24\penalty0 (6):\penalty0
  2431--2461, 1996.

\bibitem[Evans(1998)]{Evans2010PDE}
Lawrence~C Evans.
\newblock \emph{Partial Differential Equations}.
\newblock American Mathematical Society, 1st edition, 1998.

\bibitem[Fan(1991{\natexlab{a}})]{fan1991asymptotic}
Jianqing Fan.
\newblock Asymptotic normality for deconvolution kernel density estimators.
\newblock \emph{Sankhy{\=a}: The Indian Journal of Statistics, Series A}, pages
  97--110, 1991{\natexlab{a}}.

\bibitem[Fan(1991{\natexlab{b}})]{fan1991global}
Jianqing Fan.
\newblock Global behavior of deconvolution kernel estimates.
\newblock \emph{Statistica Sinica}, pages 541--551, 1991{\natexlab{b}}.

\bibitem[Gajek(1986)]{gajek1986improving}
Leslaw Gajek.
\newblock On improving density estimators which are not bona fide functions.
\newblock \emph{The Annals of Statistics}, 14\penalty0 (4), 1986.

\bibitem[Goldberg and Johndrow(2017)]{goldberg2017decision}
David Goldberg and James~E Johndrow.
\newblock A decision theoretic approach to {A/B} testing.
\newblock \emph{arXiv preprint arXiv:1710.03410}, 2017.

\bibitem[Grenander(1981)]{grenander1981abstract}
Ulf Grenander.
\newblock \emph{Abstract Inference}.
\newblock Wiley Series, New York, 1981.

\bibitem[Hall(1981)]{hall1981trigonometric}
Peter Hall.
\newblock On trigonometric series estimates of densities.
\newblock \emph{The Annals of Statistics}, 9\penalty0 (3):\penalty0 683--685,
  1981.

\bibitem[Hall(1986)]{hall1986rate}
Peter Hall.
\newblock On the rate of convergence of orthogonal series density estimators.
\newblock \emph{Journal of the Royal Statistical Society: Series B
  (Methodological)}, 48\penalty0 (1):\penalty0 115--122, 1986.

\bibitem[Hyv{\"a}rinen(2005)]{hyvarinen2005estimation}
Aapo Hyv{\"a}rinen.
\newblock Estimation of non-normalized statistical models by score matching.
\newblock \emph{Journal of Machine Learning Research}, 6, 2005.

\bibitem[Kronmal and Tarter(1968)]{kronmal1968estimation}
Richard Kronmal and Michael Tarter.
\newblock The estimation of probability densities and cumulatives by {Fourier}
  series methods.
\newblock \emph{Journal of the American Statistical Association}, 63\penalty0
  (323):\penalty0 925--952, 1968.

\bibitem[Kumaresan(2005)]{kumaresan2005topology}
Somaskandan Kumaresan.
\newblock \emph{Topology of Metric Spaces}.
\newblock Alpha Science International Ltd., Harrow, U.K., 2005.

\bibitem[Lai and Siegmund(1986)]{lai1986contributions}
Tze~Leung Lai and David Siegmund.
\newblock The contributions of {Herbert Robbins} to mathematical statistics.
\newblock \emph{Statistical Science}, 1\penalty0 (2):\penalty0 276--284, 1986.

\bibitem[Laird(1978)]{laird1978nonparametric}
Nan Laird.
\newblock Nonparametric maximum likelihood estimation of a mixing distribution.
\newblock \emph{Journal of the American Statistical Association}, 73\penalty0
  (364):\penalty0 805--811, 1978.

\bibitem[Leonard(1984)]{leonard1984some}
Tom Leonard.
\newblock Some data-analytic modifications to bayes-stein estimation.
\newblock \emph{Annals of the Institute of Statistical Mathematics},
  36\penalty0 (1):\penalty0 11--21, 1984.

\bibitem[Maclaurin et~al.(2015)Maclaurin, Duvenaud, and
  Adams]{maclaurin2015autograd}
Dougal Maclaurin, David Duvenaud, and Ryan~P Adams.
\newblock Autograd: Effortless gradients in numpy.
\newblock In \emph{ICML 2015 AutoML Workshop}, volume 238, 2015.

\bibitem[Meister(2007)]{meister2007deconvolving}
Alexander Meister.
\newblock Deconvolving compactly supported densities.
\newblock \emph{Mathematical Methods of Statistics}, 16\penalty0 (1):\penalty0
  63--76, 2007.

\bibitem[Newey and Powell(2003)]{newey2003instrumental}
Whitney~K Newey and James~L Powell.
\newblock Instrumental variable estimation of nonparametric models.
\newblock \emph{Econometrica}, 71\penalty0 (5):\penalty0 1565--1578, 2003.

\bibitem[O'Donoghue and Candes(2015)]{o2015adaptive}
Brendan O'Donoghue and Emmanuel Candes.
\newblock Adaptive restart for accelerated gradient schemes.
\newblock \emph{Foundations of Computational Mathematics}, 15\penalty0
  (3):\penalty0 715--732, 2015.

\bibitem[Robbins(1956)]{robbins1956empirical}
Herbert Robbins.
\newblock An empirical {Bayes} approach to statistics.
\newblock \emph{Herbert Robbins Selected Papers}, pages 41--47, 1956.

\bibitem[Robbins(1964)]{robbins1964empirical}
Herbert Robbins.
\newblock The empirical {Bayes} approach to statistical decision problems.
\newblock \emph{The Annals of Mathematical Statistics}, 35\penalty0
  (1):\penalty0 1--20, 1964.

\bibitem[Robbins(1977)]{robbins1977prediction}
Herbert Robbins.
\newblock Prediction and estimation for the compound {Poisson} distribution.
\newblock \emph{Proceedings of the National Academy of Sciences}, 74\penalty0
  (7):\penalty0 2670, 1977.

\bibitem[Robbins(1983)]{robbins1983some}
Herbert Robbins.
\newblock Some thoughts on empirical {Bayes} estimation.
\newblock \emph{The Annals of Statistics}, pages 713--723, 1983.

\bibitem[Rudin(1964)]{rudin1964principles}
Walter Rudin.
\newblock \emph{Principles of Mathematical Analysis}.
\newblock McGraw-Hill, 3rd edition, 1964.

\bibitem[Satterthwaite(1946)]{satterthwaite1946approximate}
Franklin~E Satterthwaite.
\newblock An approximate distribution of estimates of variance components.
\newblock \emph{Biometrics Bulletin}, 2\penalty0 (6):\penalty0 110--114, 1946.

\bibitem[Shen(1997)]{shen1997methods}
Xiaotong Shen.
\newblock On methods of sieves and penalization.
\newblock \emph{The Annals of Statistics}, 25\penalty0 (6):\penalty0
  2555--2591, 1997.

\bibitem[Stefanski and Carroll(1990)]{StefanskiCarroll1990}
Leonard Stefanski and Raymond~J. Carroll.
\newblock Deconvoluting kernel density estimators.
\newblock \emph{Statistics}, 21:\penalty0 169--184, 1990.

\bibitem[Stein(1956)]{stein1956inadmissibility}
Charles Stein.
\newblock Inadmissibility of the usual estimator for the mean of a multivariate
  normal distribution.
\newblock In \emph{Proceedings of the Third Berkeley Symposium on Mathematical
  Statistics and Probability, Volume 1: Contributions to the Theory of
  Statistics}, pages 197--206, 1956.

\bibitem[Terrell and Scott(1980)]{terrell1980improving}
George~R Terrell and David~W Scott.
\newblock On improving convergence rates for nonnegative kernel density
  estimators.
\newblock \emph{The Annals of Statistics}, 8\penalty0 (5):\penalty0 1160--1163,
  1980.

\bibitem[van~der Vaart(2000)]{van2000asymptotic}
Aad~W van~der Vaart.
\newblock \emph{Asymptotic Statistics}.
\newblock Cambridge University Press, 2000.

\bibitem[Wager(2014)]{wager2014geometric}
Stefan Wager.
\newblock A geometric approach to density estimation with additive noise.
\newblock \emph{Statistica Sinica}, pages 533--554, 2014.

\bibitem[Wahba(1975)]{wahba1975optimal}
Grace Wahba.
\newblock Optimal convergence properties of variable knot, kernel, and
  orthogonal series methods for density estimation.
\newblock \emph{The Annals of Statistics}, pages 15--29, 1975.

\bibitem[Walter and Blum(1979)]{walter1979probability}
G.~Walter and J.~Blum.
\newblock Probability density estimation using delta sequences.
\newblock \emph{The Annals of Statistics}, 7\penalty0 (2):\penalty0 328--340,
  1979.

\bibitem[Walter(1981)]{walter1981orthogonal}
G.~G. Walter.
\newblock Orthogonal series estimators of the prior distribution.
\newblock \emph{Sankhy{\=a}: The Indian Journal of Statistics, Series A}, pages
  228--245, 1981.

\bibitem[Walter and Hamedani(1991)]{walter1991bayes}
G.~G. Walter and G.~G. Hamedani.
\newblock Bayes empirical {Bayes} estimation for natural exponential families
  with quadratic variance functions.
\newblock \emph{The Annals of Statistics}, pages 1191--1224, 1991.

\bibitem[Wang and Carreira-Perpi{\~n}{\'a}n(2013)]{wang2013projection}
Weiran Wang and Miguel~{\'A}. Carreira-Perpi{\~n}{\'a}n.
\newblock Projection onto the probability simplex: An efficient algorithm with
  a simple proof, and an application.
\newblock \emph{arXiv preprint arXiv:1309.1541}, 2013.

\bibitem[Welch(1947)]{welch1947generalization}
Bernard~L Welch.
\newblock The generalization of `student's' problem when several different
  population variances are involved.
\newblock \emph{Biometrika}, 34:\penalty0 28--35, 1947.

\bibitem[Zygmund(2002)]{zygmund2002trigonometric}
Antoni Zygmund.
\newblock \emph{Trigonometric Series}, volume I, II.
\newblock Cambridge University Press, 2002.

\end{thebibliography}

\end{document}